\PassOptionsToPackage{unicode}{hyperref}
\PassOptionsToPackage{hyphens}{url}
\PassOptionsToPackage{dvipsnames,svgnames,x11names}{xcolor}

\documentclass[a4paper, 11pt]{article}

\usepackage[]{natbib}
\usepackage{bookmark}

\usepackage{graphicx}%
\usepackage{multirow}%
\usepackage{amsmath,amssymb,amsfonts}%
\usepackage{amsthm}%
\usepackage{mathrsfs}%
\usepackage[title]{appendix}%
\usepackage{xcolor}%
\usepackage{textcomp}%
\usepackage{manyfoot}%
\usepackage{booktabs}%
\usepackage{algorithm}%
\usepackage{algorithmicx}%
\usepackage{algpseudocode}%
\usepackage{listings}%
\usepackage{caption}
\usepackage{subcaption}
\usepackage{longtable}
\usepackage{array}
\usepackage{tabularx}
\usepackage{amsmath,amssymb}
\usepackage{iftex}

\theoremstyle{thmstyleone}%
\newtheorem{theorem}{Theorem}
\theoremstyle{thmstyletwo}%
\usepackage{amsthm}
\newtheorem{example}{Example}

\theoremstyle{thmstylethree}%
\newtheorem{definition}{Definition}%

\usepackage{bm}
\newcommand{\bl}[1]{\bm{#1}}

\usepackage[utf8]{inputenc}
\usepackage[T1]{fontenc}
\usepackage[margin=1in]{geometry}

\title{Adaptive online kernel changepoint detection}

\author{Qianqian Jiang  \and Dean A.~Bodenham}

\date{{\normalsize Department of Mathematics, Imperial College London,} \\
{\small \texttt{q.jiang@imperial.ac.uk},
\quad \texttt{dean.bodenham@imperial.ac.uk}}
}

\begin{document}

\maketitle

\begin{abstract}
We propose an adaptive online kernel-based changepoint detection method for streaming data that is capable of detecting a broad range of changes in the underlying data distribution. The method maintains a recursively-weighted reproducing kernel Hilbert space representation of observations and adaptively updates the forgetting factor through a gradient-based procedure driven by a maximum mean discrepancy-type statistic between the current observation and the weighted empirical distribution of the past. This self-tuning mechanism allows the detector to adapt its effective memory and responsiveness to changes in the underlying process. Simulation results demonstrate that the proposed method achieves strong detection performance across a wide range of distributional changes. Further, our proposed approach maintains constant computational and storage cost  through recursive updates, and is very computationally efficient in comparison to competing methods. Experiments on both simulated data and benchmark real-world datasets show improved performance over several other leading kernel-based methods.
\end{abstract}

\noindent%
{\it Keywords:} Adaptive forgetting factor, gradient descent, kernel methods, random Fourier features, online changepoint detection

\section{Introduction}\label{sec1}

Online changepoint detection has a long history dating to the seminal work of \cite{Page1954CONTINUOUSIS}. Subsequent developments have largely followed two main frameworks, represented by the works of \cite{Page1954CONTINUOUSIS} and \cite{Chu1996MonitoringSC}: statistical process control and Type I error control. These ideas are closely related. In the statistical process control framework, procedures are typically designed to achieve a prescribed in-control average run length (ARL), which measures the expected time to a false alarm under the null hypothesis of no change, while seeking to minimize the detection delay after a change occurs. In the Type I error control framework, the probability of false alarms is controlled directly, while the detection delay is minimized subject to this constraint. 
See \cite{tartakovsky2014sequential} and \cite{Aue2023TheSO} for comprehensive reviews of these two frameworks.

Within these general frameworks, an important challenge arises when the form of the change is unknown in advance. A parametric detector designed to identify a shift in the mean may have limited power against changes in variance, dependence structure, tail behavior, or other complex changes in the shape of the underlying distribution. Kernel methods \citep{shawe2004kernel} provide a flexible nonparametric alternative. By mapping observations into a reproducing kernel Hilbert space (RKHS), they represent a probability distribution through its mean embedding. When the kernel is characteristic, distinct distributions have distinct embeddings, allowing changes in a broad range of distributional features to be detected without specifying a particular parametric alternative \citep{fukumizu2008characteristic}. This principle forms the foundation for kernel two-sample tests, most notably the maximum mean discrepancy (MMD), which measures the distance between the mean embeddings of two distributions  \citep{gretton2006kernel,Gretton2012AKT}. Such methods therefore provide a natural foundation for constructing online detectors that are sensitive to general distributional changes.

Kernel-based changepoint detection methods have been developed for both offline settings \citep{Harchaoui2007RetrospectiveMC, Harchaoui2008KernelCA, arlot2019kernel} and online settings. Among the online methods, \cite{li2019scan} develop the Scan \(B\)-statistic, which averages MMD statistics between a fixed-size block of the most recent observations and multiple reference blocks drawn from pre-change data, and standardizes the resulting statistic by its null standard deviation. \cite{keriven2020newma} introduce NEWMA, which compares two exponentially weighted moving averages with different forgetting factors without explicitly storing past observations. \cite{ferrari2023online} propose a nonparametric detector based on kernel density-ratio estimation between consecutive reference and test windows, with an objective designed to yield an unbiased density-ratio estimate under the null hypothesis. \cite{kalinke2025optimal} introduce MMDEW, which uses exponential windows to efficiently compute MMD-based statistics for online change detection with polylogarithmic runtime and logarithmic memory complexity. \cite{kalinke2025optimal} introduce Online RFF MMD, which uses random Fourier features to efficiently approximate MMD over a dyadic grid of candidate changepoint locations.  
\cite{Wei2022OnlineKC} develop an online kernel CUSUM procedure that compares recent observation windows of varying lengths with pre-change reference blocks, constructs self-normalized MMD-based statistics for the corresponding candidate changepoint locations, and uses their maximum as the detection statistic.

Online implementations introduce a second challenge: determining how much of the past, the \emph{memory}, the detector should retain. A long memory reduces estimation variability under the no-change regime, but this can also make the detector less reactive to a change and thereby delay detection. A short memory enables a faster response, but it also increases estimation variability and may lead to more frequent false alarms. Existing methods typically address this trade-off by relying on fixed memory parameters, such as a prespecified window length or fixed exponential forgetting factors. However, such choices can be restrictive when the magnitude, direction, and duration of future changes are unknown, especially if one is monitoring a process where multiple changepoints can occur.  We address this trade-off by learning the value of the forgetting factor directly from the data stream. Adaptive forgetting has previously been studied in online classification \citep{anagnostopoulos2012online}, continuous changepoint detection for mean shifts \citep{bodenham2017continuous}, online changepoint detection for transition matrices \citep{Plasse2021Streaming}, unsupervised streaming detection based on adaptive PCA \citep{hoeltgebaum2021unsupervised}, and online spectral density estimation \citep{kazi2026online}. To the best of our knowledge, the present work is the first adaptive online kernel-based changepoint detection method that incorporates a data-driven forgetting mechanism.

In the proposed method, each incoming observation is compared with a recursively-weighted representation of past observations in the RKHS. The forgetting factor is updated through a stochastic gradient-based algorithm driven by a cost function that measures the squared RKHS distance between the feature representation of the current observation and the recursively-weighted kernel mean embedding of past observations. In this sense, the cost function measures a distributional difference between the current and past data, following the same principle as the MMD. Intuitively, the method retains past information when consecutive observations remain compatible, but discounts older information more aggressively when the current feature representation becomes inconsistent with the recent past. The resulting detector therefore adapts its effective memory to the local dynamics of the data stream rather than fixing it for the entire monitoring horizon. Our experiments on simulated and real data demonstrate that the proposed method maintains high detection power for weak change signals while effectively controlling the average run length. For larger changes, it achieves rapid detection with short expected detection delays. These results indicate that adaptive forgetting provides a practical and flexible mechanism for balancing stability and responsiveness in online changepoint detection.

The remainder of the paper is organized as follows. Section~\ref{prelim} introduces the preliminary concepts and problem setting. Section~\ref{weighted-average} presents the proposed methodology, and Section~\ref{AFFexpr} reports the experimental results. Section~\ref{supmat} provides the proofs and additional experimental results in Supplementary material.

\section{Preliminaries}\label{prelim}

\subsection{Changepoint detection}

\label{sec:prelimcpd}

Consider independent data stream $(X_t)_{t\geq 1}$ taking values in $\mathcal X\subseteq \mathbb R^d$ and following an unknown pre-change distribution $p$:
\[X_1, X_2,\ldots,X_M\stackrel{\mathrm{i.i.d}}{\sim} p.\]
The goal is to perform a sequential test of
\[
H_0:\quad X_1,\ldots,X_t,\ldots \stackrel{\mathrm{i.i.d.}}{\sim} p
\]
against
\begin{align*}
    H_1:\quad \exists k> 1 \text{ such that } X_1,X_2,\ldots,X_k \stackrel{\mathrm{i.i.d.}}{\sim} p,
\quad
X_{k+1},X_{k+2},\ldots\stackrel{\mathrm{i.i.d.}}{\sim} q\neq p.
\end{align*}

We aim to detect a change in the underlying data distribution as soon as possible after the change has occurred, while controlling false alarms.
On the filtered probability space $(\Omega,\mathcal F,(\mathcal F_t)_{t\geq 0},\mathbb P)$ where $(\mathcal F_t)_{t\geq 0}$ is the filtration generated by $(X_t)_{t\geq 1}$, an online changepoint detection procedure aims to determine a stopping time $T$ with respect to the filtration $(\mathcal F_t)_{t\geq 0}$. 
Let \(\mathbb{P}_{\infty}\) denote the probability distribution of \((X_t)_{t\geq 1}\) when no change occurs, and for the changepoint $k<\infty$, let \(\mathbb{P}_k\) denote the probability distribution of \((X_t)_{t\geq 1}\) when the change occurs after time \(k\), so that \(X_1,\dots,X_k\sim p\) and \(X_{k+1},X_{k+2},\dots\sim q\neq p\). Let \(\mathbb{E}_{\infty}\) and \(\mathbb{E}_k\) denote expectations under \(\mathbb{P}_{\infty}\) and \(\mathbb{P}_k\), respectively.

A false alarm occurs when the stopping time is identified before the changepoint occurs. Therefore the false-alarm probability when the changepoint occurs at time $k<\infty$ is
\(
\mathbb{P}_k(T\leq k),
\)
and in the no-change scenario, the false alarm probability is
\(
\mathbb{P}_{\infty}(T<\infty).
\)
The average run length (ARL) under $H_0$ is defined by 
\[
\operatorname{ARL}:=\mathbb{E}_{\infty}[T].
\]
Conditional on no false alarm before the change, one defines the (conditional) expected detection delay as
\[
\operatorname{EDD}:=\mathbb{E}_k[T-k|T>k].
\]
The ARL and EDD are the two performance metrics often used when evaluating the performance of an online changepoint detection method \citep{Page1954CONTINUOUSIS}. Sometimes these are referred to as $\mathrm{ARL}_{0}$ and $\mathrm{ARL}_{1}$, respectively.
One aims to design a changepoint algorithm that has a high average run length while minimizing the expected detection delay.

\subsection{EWMA chart}

The Exponentially Weighted Moving Average (EWMA) chart were first proposed by \cite{Roberts1959} for detecting change in mean for a sequence of random variables. Let $X_1,\ldots,X_n$ be a sequence of independent random variables with the common mean $u_0$ before the changepoint, and $u_1$ after the changepoint. The EWMA chart is defined recursively by $Z_0=u_0$,
\[Z_t=\lambda Z_{t-1}+(1-\lambda)X_t,\quad t\geq 1,\]
where $0<\lambda<1$.
The EWMA statistic assigns greater weight to recent data, with past information decaying exponentially. The parameter $\lambda$ controls the relative weight placed on recent versus older data, smaller values emphasize recent data, while larger values retain more historical information. Before the changepoint, the mean of $Z_t$ is $u_0$, the standard deviation of $Z_t$ is $\sigma_{Z_t}=\sqrt{\frac{(1-\lambda)}{(1+\lambda)}(1-\lambda^{2t})}\sigma_{X}$, where $\sigma_{X}$ is the common standard deviation of $X_t$. Thus, before the changepoint, $Z_t$ will fluctuates around $u_0$. After a change occurs, $Z_t$ will move away from  $u_0$ and toward $u_1$. The changepoint is detected when either 
\(Z_t>u_0+L\sigma_{Z_t}\) or \(Z_t<u_0-L\sigma_{Z_t}\), where $L$ is the control limit which is chosen to achieve a predefined level of detection performance.

Instead of detecting changes in the mean, we aim to detect changes in the underlying distribution. Therefore, we embed the data into a reproducing kernel Hilbert space, where distributional changes can be detected.

\subsection{Reproducing kernel Hilbert Spaces}\label{methdg}
\begin{definition}[Reproducing kernel Hilbert space {\cite{aronszajn1950theory}}]\label{RKHS}
    Let \(\mathcal X\) be an arbitrary set . A Hilbert space \(\mathcal H\) of real-valued functions on \(\mathcal X\) is called a reproducing kernel Hilbert space (RKHS) if, for every \(x\in\mathcal X\), the point-evaluation functional
    \[
    L_x:\mathcal H\to\mathbb R,\qquad L_x(f)=f(x),
    \]
    is continuous. Equivalently, there exists a unique function \(K:\mathcal X\times\mathcal X\to\mathbb R\), called the reproducing kernel of \(\mathcal H\), such that \(K(\cdot,x)\in\mathcal H\) for every \(x\in\mathcal X\) and
    \[
    f(x)=\langle f,K(\cdot,x)\rangle_{\mathcal H},
    \qquad \forall f\in\mathcal H, \forall \ x\in\mathcal X.
    \]
    In particular,
    \[
    K(x,y)=\langle K(\cdot,y),K(\cdot,x)\rangle_{\mathcal H},
    \qquad x,y\in\mathcal X.
    \]
\end{definition}
The existence of the function \(K:\mathcal X\times\mathcal X\to\mathbb R\) is guaranteed by Riesz’s representation theorem as $L_x$ is continuous at every $f\in \mathcal{H}$ \citep{reed1980methods,christmann2008support}. 
It is easy to see that the reproducing kernel \(K\) is symmetric and positive semidefinite, i.e., for every \(n\in\mathbb N\), for every \(x_1,\dots,x_n\in\mathcal X\), and for every \(c_1,\dots,c_n\in\mathbb R\), \(\sum_{i,j=1}^n c_i c_j K(x_i,x_j)\geq 0\). The Moore–Aronszajn theorem \citep{aronszajn1950theory} is a converse to this, it states that if a function 
\(K\) satisfies these conditions then there is a Hilbert space of functions on \(\mathcal X\) for which it is a reproducing kernel.  Denote $\mathcal H(K) $ be the reproducing kernel Hilbert space associated with a positive semidefinite kernel $K:\mathcal{X}\times\mathcal{X}\to \mathbb{R}$.

\begin{definition}[Mean embedding in Hilbert space \cite{Gretton2012AKT}]\label{meanembedding}
    Let $p$ be a probability distribution such that $\mathbb{E}_{X\sim p}K^{1/2}(X,X)<\infty$. Then there exist a $\mu_{p}\in\mathcal H(K)$ such that for every $f\in\mathcal H(K)$,
    \[
    \mathbb{E}_{X\sim p} f(X)=\langle f,\mu_p\rangle_{\mathcal H(K)}.
    \]
    $\mu_p$ is called the mean embedding of $p$ in $\mathcal H(K)$.
\end{definition}
The existence of the embedding $\mu_p$ in $\mathcal H(K)$ is guaranteed by Riesz’s representation theorem as the linear functional $f\mapsto \mathbb{E}_{X\sim p} f(X)$ is bounded under the assumption that $\mathbb{E}_{X\sim p}K^{1/2}(X,X)<\infty$.

The mean embedding provides the idea for detecting change in distributions. Under no changepoint regime with distribution \(p\), the empirical
average of the feature vectors \(K(X_t,\cdot)\) estimates \(\mu_p\). If the
distribution changes from \(p\) to \(q\), the population mean of the incoming
feature vectors changes from \(\mu_p\) to \(\mu_q\). For a characteristic
kernel, \(p\neq q\) implies \(\mu_p\neq\mu_q\) (see, e.g., \cite{fukumizu2007kernel}). Hence, a distributional change
in the original observations becomes a mean change in the RKHS. This observation suggests continuously updating an estimate of the kernel mean embedding as new observations arrive and monitoring it for evidence of a distributional change. An equally weighted average is effective for estimating the mean embedding when the distribution remains unchanged, but it can respond slowly after a change because observations
from the previous regime retain the same weight as recent observations. For
online detection, it is therefore natural to discount older data and place
greater emphasis on observations that better represent the current regime.
The classical exponentially weighted moving average provides the basic
recursive mechanism for achieving this balance between stability and
responsiveness.

\section{Proposed method: OKAFF}\label{weighted-average}
Let \((X_t)_{t\ge 1}\) be a data stream taking values in
\(\mathcal X\subseteq\mathbb R^d\).  Let
\[
K:\mathcal X\times\mathcal X\to\mathbb R
\]
be a symmetric positive semidefinite kernel with the reproducing kernel Hilbert space
\(\mathcal H(K)\), and let
\[
\psi:\mathcal X\to\mathcal H(K),
\qquad
\psi(x):=K(x,\cdot),
\]
be the canonical feature map. Define
\[
Y_t:=\psi(X_t),\quad t\geq 1,
\]
then $(Y_t)_{t\geq 1}$ is the $\mathcal H(K)$-value data stream. In this way, distributional change detection of the original data stream is converted into mean-change detection in \(\mathcal H(K)\).  If \(X_t\sim p\), then
\[
\mathbb E[Y_t]
=\mathbb E_{X\sim p}[\psi(X)]
=\mu_p.
\]
After a change from \(p\) to \(q\), the mean of the feature stream changes
from \(\mu_p\) to \(\mu_q\). If \(K\) is characteristic, then \(p\neq q\)
implies \(\mu_p\neq\mu_q\). We can therefore detect a change in the
distribution of the original stream by monitoring the mean behavior of the
\(\mathcal H(K)\)-valued stream.

At time \(t \in  \{1, 2, \dots\}\), the empirical distribution of the original observations is
\[
 p_t=\frac{1}{t}\sum_{i=1}^t\delta_{X_i},
\]
where $\delta_x$ is the Dirac measure at $x\in\mathcal{X}$, and its mean embedding is
\[
\mu_{p_t}
=\frac{1}{t}\sum_{i=1}^t\psi(X_i).
\]
In what follows, we assign greater weight to recent observations because they better represent the current regime, sharing the same principle as the EWMA scheme. Let
\((\lambda_t)_{t\geq0}\) be a sequence of forgetting factors satisfying
\(0<\lambda_t\leq1\), and define
\begin{align}\label{ait}
    a_{i,t}:=\prod_{s=i}^{t-1}\lambda_s,
\qquad
W_t:=\sum_{j=1}^t a_{j,t},\qquad b_{i,t}:=\frac{a_{i,t}}{W_t},
\end{align}
with the convention that an empty product is equal to \(1\). The forgetting-factor-weighted empirical distribution is
\[
p_{t,\lambda}
=\sum_{i=1}^t\frac{a_{i,t}}{W_t}\,\delta_{X_i},
\]
and its mean embedding is defined as
\begin{equation}
 \bar Y_t :=  \mu_{p_{t,\lambda}}
=\sum_{i=1}^t\frac{a_{i,t}}{W_t}\,\psi(X_i).
\label{eqn:Ybardefn}
\end{equation}
Thus, \(\bar Y_t\) is simultaneously the weighted mean of the
Hilbert space-valued stream and the mean embedding of the weighted empirical
distribution of the original stream.

A direct evaluation of this weighted embedding would require revisiting all
past observations at every time step. To obtain an online procedure, we retain
only the unnormalized weighted feature sum and its total weight. These two
quantities are denoted by \(M_t\) and \(W_t\) below. Their ratio gives the
normalized weighted mean, and both admit one-step recursive updates. Let
\(
M_0=0\), \(W_0=0\).
For \(t\geq1\), define recursively
\begin{equation}
M_{t}=\lambda_{t-1} M_{t-1}+Y_{t}, \qquad
W_{t}=\lambda_{t-1} W_{t-1}+1.
\label{eqn:mwupdate}
\end{equation}
Then, for \(t\ge 1\), we have the normalized weighted average 
\begin{equation}
\bar Y_t(\bl{\lambda}_{t-1}) :=\bar Y_t =\frac{M_t}{W_t},
\label{eqn:ybar}
\end{equation}
which is a function of $\bl{\lambda}_{t-1}=(\lambda_s)_{0\leq s\leq t-1}=(\lambda_0,\lambda_1,\ldots,\lambda_{t-1})$ and the sequence of data $(X_t)_{t\geq 1}$. $\bar Y_t$ usually lives in infinite dimensional reproducing kernel Hilbert space, monitoring it directly is impractical. We therefore define detection statistic as
\begin{align}\label{detdecsta}
  S_t \:=\; \|\bar{Y}_t\|_{\mathcal H(K)}^2,\qquad t\geq1.
\end{align}

When \(\lambda_t=1\) for all \(t\), \(\bar{Y}_t\) is the empirical estimator of the mean embedding \(\mu_p\). It is unbiased, and the law of large numbers implies that $\bar{Y}_t$ converges to $\mu_p$ as $t\to\infty$ (see, e.g., \cite{berlinet2004reproducing}).

Before presenting our main theoretical results, we introduce some notation.
Recall the kernel mean embedding $\mu_p$ is defined as
\begin{align}\label{meanebd}
    \mu_p:=\mathbb E_{X\sim p}\psi(X)\in \mathcal H(K).
\end{align}
The covariance operator $\Sigma_p$ of the RKHS-valued random element \(\psi(X)\) is
defined by
\begin{align}\label{covebd}
\Sigma_p
:=
\mathbb E_{X\sim p}
\left[
\{\psi(X)-\mu_p\}\otimes \{\psi(X)-\mu_p\}
\right],
\end{align}
where, for \(f,g\in\mathcal H(K)\), the rank-one operator
\(f\otimes g:\mathcal H(K)\to\mathcal H(K)\) is defined by
\[
(f\otimes g)h
:=
\langle g,h\rangle_{\mathcal H(K)} f,
\qquad h\in\mathcal H(K).
\]
Then, for every \(h\in\mathcal H(K)\),
\[
\Sigma_p h
=
\mathbb E_{X\sim p}
\left[
\left\langle \psi(X)-\mu_p,h\right\rangle_{\mathcal H(K)}
\{\psi(X)-\mu_p\}
\right].
\]
For all \(h_1,h_2\in\mathcal H(K)\),
\[
\left\langle \Sigma_p h_1,h_2\right\rangle_{\mathcal H(K)}
=
\operatorname{Cov}
\left(
h_1(X),h_2(X)
\right),
\]
where we use the reproducing property
\[
\langle \psi(X),h\rangle_{\mathcal H(K)}=h(X).
\]

The following two results describe the asymptotic behaviour of the detection statistic $S_t$, in the cases when $\mu_p = 0$ and $\mu_p \neq 0$.
\begin{theorem}
\label{CLTffwave}
Let \(\mathcal X\) be a topological space equipped with its Borel \(\sigma\)-algebra \(\mathcal B(\mathcal X)\), and let \((X_t)_{t\ge1}\) be a sequence of i.i.d. \(\mathcal X\)-valued random variables with common law \(p\). Let \(K:\mathcal X\times\mathcal X\to\mathbb R\) be a continuous, symmetric, positive semidefinite kernel, and let \(\mathcal H(K)\) denote its associated RKHS. Assume that \(
\mathbb E_{X\sim p}K(X,X)<\infty
\) and that \(\mathcal H(K)\) is separable. Let \((\lambda_t)_{t\ge1}\) be a deterministic sequence satisfying
\(
0<\lambda_t<1\), and \(t(1-\lambda_t)\to0\), as \(t\to\infty\).
Denote the mean embedding of \(p\) by \(\mu_p\) and 
define the detection statistic $S_t$ as in \eqref{detdecsta}.
Then, if \(\mu_p=0\), as $t\to\infty$,
\[
tS_t
\stackrel{d}{\rightarrow}
\sum_{r\ge1}\eta_{r,p}Z_r^2,
\]
and hence,
\[
t(S_t-\mathbb E S_t)
\stackrel{d}{\rightarrow}
\sum_{r\ge1}\eta_{r,p}(Z_r^2-1),
\]
where \((\eta_{r,p})_{r\ge1}\) are the nonzero eigenvalues of
\(\Sigma_p\) defined in \eqref{covebd}, counted with multiplicity, and
\(Z_1,Z_2,\ldots\) are i.i.d. \(N(0,1)\).
\end{theorem}

\begin{theorem}
\label{cor:CLTffwave}
Under the same conditions as Theorem~\ref{CLTffwave},
if \(\mu_p\neq0\), then as $t\to\infty$,
\[
\sqrt t
\left(
S_t-\|\mu_p\|_{\mathcal H(K)}^2
\right)
\stackrel{d}{\rightarrow}
N\left(
0,
4\left\langle \Sigma_p\mu_p,\mu_p\right\rangle_{\mathcal H(K)}
\right).
\] 
Hence,
\[
\sqrt t
\left(
S_t-\mathbb E S_t
\right)
\stackrel{d}{\rightarrow}
N\left(
0,
4\left\langle \Sigma_p\mu_p,\mu_p\right\rangle_{\mathcal H(K)}
\right),
\]
where
\(
\mathbb E S_t
=
\eta_K\delta_t
+
\bigl(1-\delta_t\bigr)\theta_K,
\)
with \(\eta_K:=\mathbb E_{X\sim p} K(X,X)\),
$\theta_K:=\mathbb E_{X,Y\sim p}  K(X,Y)$, with \(X,Y\) independent random variables with law \(p\), 
and
\begin{align}\label{deltaL}
    \delta_t
:=
\sum_{i=1}^t b_{i,t}^2,
\end{align}
where \(b_{i,t}\) is defined in \eqref{ait}.
\medskip
\end{theorem}

The asymptotic framework requires the forgetting factor \(\lambda_t\) to vary over time and converge to \(1\) as \(t\to\infty\), thereby ensuring that the effective sample size grows asymptotically.  This requirement motivates us to consider a sequence of time-varying forgetting factors, \((\lambda_t)_{t\geq 1}\),  and to develop an adaptive forgetting-factor framework in the next section. The result further shows that, from an asymptotic perspective, the marginal distribution of the detection statistic \(S_t\) is nearly time-invariant. Consequently, a fixed control limit can be determined for the detection procedure for a given nominal in-control ARL, thereby providing practical guidance for threshold selection.

While Theorem~\ref{CLTffwave} assumes that \(t(1-\lambda_t)\to0\) as \(t\to\infty\), in practice \((\lambda_t)_{t\geq 1}\) is a stochastic, rather than a deterministic, sequence. Although we cannot necessarily guarantee that such a stochastic sequence will obey this assumption, this theorem provides practical guidance for setting thresholds for the detection statistic $S_t$.
The following result similarly explores a theoretical approach to setting a threshold which will control the false alarm rate.

\begin{theorem}\label{fap1}
Let \(\mathcal X\) be a topological space equipped with its Borel \(\sigma\)-algebra \(\mathcal B(\mathcal X)\), and let \((X_t)_{t\ge1}\) be a sequence of i.i.d. \(\mathcal X\)-valued random variables with common law \(p\).  Let \(K:\mathcal X\times\mathcal X\to\mathbb R\) be a continous, symmetric, positive semidefinite kernel with associated reproducing kernel Hilbert space \(\mathcal H(K)\). Assume that
\(
\sup_{x\in\mathcal X}K(x,x)\le \kappa 
\)
for some \(\kappa>0\), and that \(\mathcal H(K)\) is separable.  Let \((\lambda_t)_{t\ge1}\) be a deterministic sequence satisfying
\(
0<\lambda_t<1\), and $1-\lambda_t=o(t^{-\beta })$ for some $\beta\geq 1$. For every
\(\alpha\in(0,1)\), let
\begin{align*}
&\varepsilon_{t}=\kappa
\sqrt{\frac{L_t}{2}
\left[
\lfloor \log_2 t\rfloor\log 2
+
2\log(\lfloor \log_2 t\rfloor+1)
+
\log\left(\frac{2\pi^2}{3\alpha}\right)
\right]
}
\end{align*}
with 
\(
L_t
:=
\sum_{i=1}^t
b_{i,t}^2(4-3b_{i,t})^2
\), and \(b_{i,t}\) is defined in \eqref{ait}.
Then, 
\[
\begin{aligned}
\mathbb P_\infty(T<\infty)
<
\alpha,
\end{aligned}
\]
where the stopping time is defined as
\(
T=
\inf\Big\{t:\ |S_t-\mathbb E S_t|>\varepsilon_t\Big\},
\)
and the detection statistic $S_t$ is defined in \eqref{detdecsta} and
\(
\mathbb E S_t
=
\eta_K\delta_t
+
\bigl(1-\delta_t\bigr)\theta_K,
\)
with \(\eta_K:=\mathbb E_{X\sim p} K(X,X)\),
$\theta_K:=\mathbb E_{X,Y\sim p}  K(X,Y)$, \(X,Y\) being independent random variables with law \(p\), 
\(
\delta_t
:=
\sum_{i=1}^t b_{i,t}^2
\), and \(\varepsilon_t=O(\sqrt{t^{-1}\log t})\).
\end{theorem}

{\centering
\includegraphics[
    width=0.9\linewidth
]{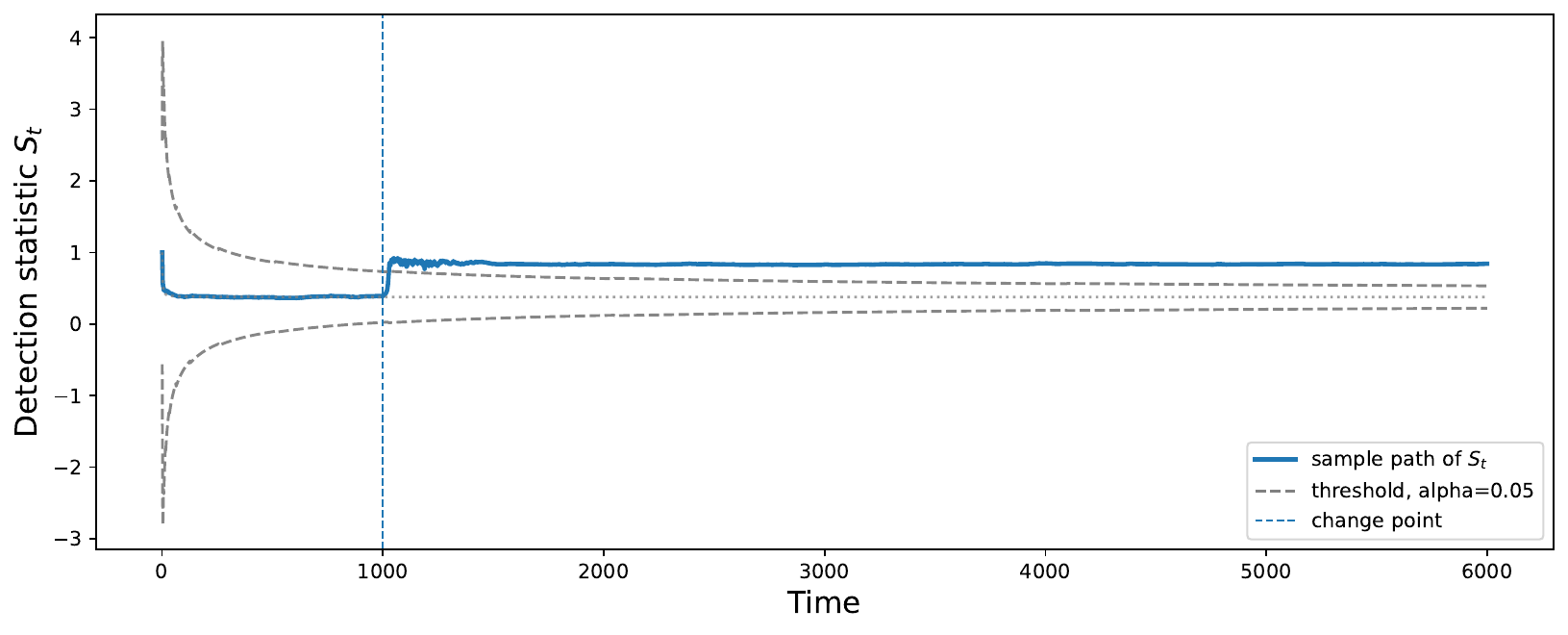}

\captionof{figure}{Example of a sample path of the proposed detection statistic $S_t$, with the time-dependent threshold from Theorem~\ref{fap1}, for data with a changepoint at $\kappa=1000$.}
\label{fig:time-dependent-threshold}
\par}

Figure~\ref{fig:time-dependent-threshold} shows an example of the proposed test statistic and the time-dependent thresholds derived in Theorem~\ref{fap1} for data with a changepoint.
In this example, the data are $20$-dimensional with the pre-change data independently sampled from a multivariate $\mathcal{N}(\mathbf{0}, I)$ distribution, and the post-change data are independently sampled from a $\operatorname{Laplace}(0,0.3)^{20}$ distribution, i.e. each of the $20$ components is independent and follows a $\operatorname{Laplace}(0,0.3)$ distribution. The forgetting factor is set using the adaptive procedure described in the next section.

While the time-dependent threshold in Theorem \ref{fap1} is theoretically sound, there are practical issues when using it with sequences where the changepoint occurs early in the monitoring process. This is discussed in Section \ref{simu}, along with proposals for alternative thresholds that overcome this issue.

A simple example of how Theorem~\ref{CLTffwave} can be applied is to consider the linear kernel and the $\mathbb{R}^d$-valued random vector $X$ with distribution $p$. Then its kernel mean embedding \(\mu_p = 0\), when $\mathbb{E}[X]=0$ and $\mathbb{E}\|X\|^2<\infty$.
We now consider an example of how Theorem~\ref{cor:CLTffwave} can be applied. Note that, for the Gaussian kernel \(K\), the mean embedding \(\mu_p \in \mathcal{H}_K\) is nonzero for every probability distribution \(p\). The calculation of $\|\mu_p\|_{H(K)}$ when $p=N(\mu,\Sigma)$ is derived in \cite{KeCe2019}. Combining this with Theorem~\ref{cor:CLTffwave}, we have the following calculation.

\begin{example}[Gaussian distribution]
Let  \(X_i\stackrel{i.i.d.}{\sim}N(\mu,\Sigma)\) on \(\mathbb R^d\), $i=1,2,\ldots,t$, and let
\[
K(x,y)=\exp\left(-\frac{\|x-y\|^2}{2\gamma^2}\right),
\qquad \gamma>0.
\]
Let \(\mathcal H(K)\) be the corresponding RKHS. Then
\[
\mathbb E S_t
=
\delta_t
+
(1-\delta_t)
\det\left(I+\frac{2\Sigma}{\gamma^2}\right)^{-1/2}
\]
where \(\delta_t\) is defined in \eqref{deltaL}, and
\begin{align*}
    \left\langle \Sigma_p\mu_p,\mu_p\right\rangle_{\mathcal H(K)}
&=
\det\left[
\left(I+\frac{\Sigma}{\gamma^2}\right)
\left(I+\frac{3\Sigma}{\gamma^2}\right)
\right]^{-1/2}-
\det\left(I+\frac{2\Sigma}{\gamma^2}\right)^{-1}.
\end{align*}
\label{eg:gaussian}
\end{example}

\subsection{Adaptive forgetting factors}

The forgetting factor approach defined in \eqref{eqn:mwupdate}--\eqref{detdecsta} controls the detector’s adaptivity and responsiveness, raising the question of how \(\lambda_t\) should be selected to achieve good detection performance. When \(\lambda_t=\lambda=1\), the recursion reduces to the standard unweighted statistic, whereas \(\lambda_t=\lambda=0\) discards the entire data history and retains only the current observation. For \(0<\lambda_t<1\), recent observations receive greater weight, while the influence of older observations decays exponentially. We develop a data-driven approach in which \(\lambda_t\) is allowed to vary over time and is updated according to an empirical measure criterion. This self-tuning procedure determines, in an online and data-dependent manner, how rapidly past information should be discarded.  Specifically, the cost function quantifies the discrepancy between the distribution represented by the current observation and that represented by past observations. This criterion follows the same principle as the maximum mean discrepancy (MMD) \citep{Gretton2012AKT}, which measures differences between current and historical data distributions in a nonparametric manner. When the kernel is characteristic, its kernel mean embedding
is injective \citep{fukumizu2007kernel}. Therefore, this cost function can in principle distinguish any two different probability distributions. In particular, the Gaussian and Laplace kernels are characteristic on $\mathbb R^d$ \citep{fukumizu2007kernel}. Consequently, the resulting detection statistic can identify changes in the underlying distribution of the data without requiring parametric assumptions about either the pre-change or post-change distribution. Although thresholds may depend on the pre-change distribution, we use the median heuristic to select the kernel bandwidth \citep{Gretton2012AKT}. Our numerical results suggest that this choice produces broadly similar thresholds across different pre-change distributions. However, the threshold (or control limit $L$ in \eqref{thre} below) must still be calibrated, and determining the control limit that achieves a desired in-control average run length is a nontrivial question. The adaptive forgetting-factor mechanism is also closely related to ideas from the Kalman filter (see, e.g., \cite{kalman1960new}, \cite{anagnostopoulos2012online}, \cite{bodenham2017continuous}). Inspired by these ideas, we propose an online procedure that adaptively updates the forgetting factor through an approximate gradient-based algorithm.
Define the cost function
\begin{align}\label{costfun}
    f_t(\bl{\tilde{\lambda}}_{t-2}):=\|\bar Y_{t-1}-Y_t\|_{\mathcal H(K)}^2,\quad t\geq2,
\end{align}
with $\bl{\tilde{\lambda}}_{t-2}=(\tilde{\lambda}_0,\ldots,\tilde{\lambda}_{t-2})$, which is discrepancy of the distribution between the current data point and the past data points.
The forgetting factors are updated adaptively using the following recursion for $t \geq 2$, after first specifying initial values for $\lambda_0$ and $\lambda_1$:
\begin{equation}
\lambda_t=\Pi_{[\ell_1,\ell_2]}\bigl(\lambda_{t-1}-\eta\,\nabla_{v}f_t(\bl{\lambda}_{t-2})\bigr),
\label{eqn:lambdat}    
\end{equation}
where $\Pi_{[\ell_1,\ell_2]}(z):=\min\{\ell_2,\max\{\ell_1,z\}\}$, $0<\ell_1<\ell_2<1$, \(\eta>0\) is a learning rate, and
\begin{align*}
\nabla_{v}f_t(\bl{\lambda}_{t-2})
&:= \left.\sum_{i=0}^{t-2}
\frac{\partial f_t}{\partial \tilde{\lambda}_i}\right|_{\bl{\tilde{\lambda}}_{t-2}=\bl{\lambda}_{t-2}}
\end{align*}
is the directional derivative of \(f_t(\bl{\tilde{\lambda}}_{t-2})\) at \(\bl{\lambda}_{t-2}=(\lambda_0,\ldots,\lambda_{t-2})\) in the direction \(v=(1,1,\ldots,1)\in\mathbb R^{t-1}\). The learning rate $\eta$ determines how
quickly the memory can adapt, while projection
prevents a single unusual observation from producing a forgetting factor value outside the range $(0, 1)$. 
The lower bound $\ell_1 > 0$ prevents the method
from discarding nearly all historical information,
the upper bound $\ell_2 < 1$ prevents the effective memory
from becoming arbitrarily long. 
By the chain rule, the directional derivative is computed by
\begin{equation}
\nabla_{v}f_t(\bl{\lambda}_{t-2})
=
2\bigl\langle \nabla_{v}\bar Y_{t-1},\,\bar Y_{t-1}-Y_t\bigr\rangle_{\mathcal H(K)},
\quad t\ge 2,
\label{eqn:dirder}    
\end{equation}

where
\begin{align*}
    \nabla_{v}\bar Y_{t-1} (\bl{\lambda}_{t-2})
&:=\left.\sum_{i=0}^{t-2}
\frac{\partial \bar Y_{t-1}}{\partial \tilde{\lambda}_i}\right|_{\bl{\tilde{\lambda}}_{t-2}=\bl{\lambda}_{t-2}}.
\end{align*}
Since the weighted
average
\[
\bar Y_{t-1}=\frac{M_{t-1}}{W_{t-1}},
\]
it follows that $\nabla_{v}\bar Y_{t-1}$ satisfies 
\begin{equation}
\nabla_{v}\bar Y_{t-1}
=
\frac{\nabla_{v}M_{t-1}-\bar Y_{t-1}\nabla_{v}W_{t-1}}{W_{t-1}},
\quad t\ge 2,
\label{eqn:gradY}
\end{equation}
where
\begin{align*}
\nabla_{v}M_{t-1}(\bl{\lambda}_{t-2})
&:=\left.\sum_{i=0}^{t-2}
\frac{\partial M_{t-1}}{\partial \tilde{\lambda}_i}\right|_{\bl{\tilde{\lambda}}_{t-2}=\bl{\lambda}_{t-2}},\qquad
\nabla_{v}W_{t-1}(\bl{\lambda}_{t-2})
&:=\left.\sum_{i=0}^{t-2}
\frac{\partial W_{t-1}}{\partial \tilde{\lambda}_i}\right|_{\bl{\tilde{\lambda}}_{t-2}=\bl{\lambda}_{t-2}}.
\end{align*}
Moreover, for \(t\ge 2\),
\begin{align}\label{eqn:mw}
M_{t-1}=\lambda_{t-2}M_{t-2}+Y_{t-1},
\qquad
W_{t-1}=\lambda_{t-2}W_{t-2}+1.
\end{align}
Hence, $\nabla_{v}M_{t-1}$ satisfys the recursion
\begin{align}\label{eqn:gradm}
\nabla_{v}M_{t-1}=M_{t-2}+\lambda_{t-2}\nabla_{v}M_{t-2},
\qquad t\ge 2,
\end{align}
with initial value
\(
\nabla_{v}M_0=0,
\)
and similarly $\nabla_{v}W_{t-1}$ satisfys the recursion
\begin{align}\label{eqn:gradw}
\nabla_{v}W_{t-1}=W_{t-2}+\lambda_{t-2}\nabla_{v}W_{t-2},
\qquad t\ge 2,
\end{align}
with initial value
\(
\nabla_{v}W_0=0.
\)
These recursions make online adaptation computationally feasible because the algorithm does not need to revisit the entire observation history whenever a new data point arrives. 
Recall that the detection statistic is
\[
S_t:=\|\bar Y_t\|_{\mathcal H(K)}^2,
\qquad t\geq1.
\]
Its behavior is governed by both the distributional features captured by the
selected kernel and the effective memory induced by the adaptive forgetting
factors. Under the no-change regime, consecutive observations
have the same distribution, so the cost function remains relatively small
and the adaptive forgetting factor tends to stay close to \(1\), preserving a
long effective memory. When the changepoint  occurs, the discrepancy between the
current observation and the weighted historical summary increases. The cost
function responds to this discrepancy by reducing the forgetting factor,
thereby discounting observations from the previous regime and assigning greater
influence to recent data. A change is signalled when \(S_t\) enters a
prescribed rejection region. Figure~\ref{fig:adaptiveforgettingfactor} shows how 
the adaptive forgetting factor is close to $1$ before the changepoint at $t=1000$, 
then its value drops soon after the changepoint occurs, before recovering to a value
close to $1$ during the post-change regime.

{\centering
\includegraphics[
    width=0.9\linewidth
]{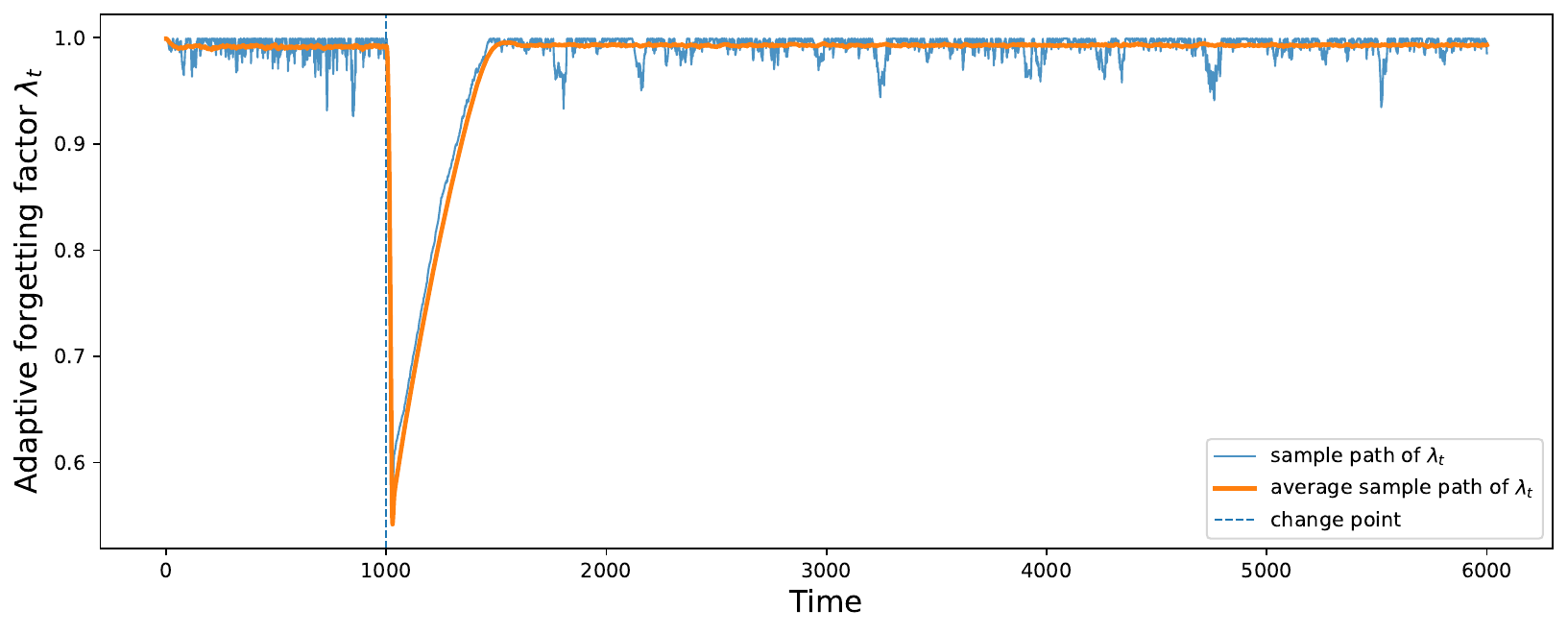}

\captionof{figure}{Sample paths of the adaptive forgetting factor, with a changepoint at $\kappa=1000$.}
\label{fig:adaptiveforgettingfactor}
\par}

\subsection{Computationally efficient detection statistics with Random Fourier Features}

The exact RKHS formulation is useful for analysis, but a direct implementation
can become expensive because the kernel representation of a running mean would
contain one term for every observation seen so far. Since the feature map is typically implicit, the statistic is evaluated through kernel functions,
\[
S_t
=
\|\bar Y_t\|_{\mathcal H(K)}^2
=
\frac{1}{W_t^2}
\sum_{i=1}^t\sum_{j=1}^t
a_{i,t}a_{j,t}K(X_i,X_j)
\]
where $W_t$ and $a_{i,t}$ are defined in \eqref{ait}. Consequently, an exact calculation would require retaining all previous observations and evaluating an increasing number of pairwise kernel terms as the data stream grows.
Its storage and computational
cost would therefore grow with time, which would make the method unsuitable for online changepoint detection. The random Fourier features method \citep{rahimi2007random} replaces the implicit,
possibly infinite-dimensional feature map by an explicit vector of finite
dimension. Once the random frequencies have been drawn, every observation can
be mapped into this feature space and all subsequent updates reduce to ordinary
vector operations. This approximation creates a transparent statistical--computational trade-off.
Increasing the number of random features improves the kernel approximation but requires more memory and computation per observation, whereas using fewer features produces a faster detector at the cost of greater approximation error. In
either case, the cost no longer grows with the monitoring horizon, because the
algorithm stores only fixed-dimensional state vectors rather than the full
sequence of past observations, the procedure is updated via the $O(1)$ recursive equations \eqref{eqn:gradY}--\eqref{eqn:gradw}.

Suppose that \(K : \mathbb{R}^d \times \mathbb{R}^d \to \mathbb{R}\) is a continuous, real-valued, translation-invariant semipositive definite kernel satisfying for all  $n\in\mathbb{N}$, $\ x_1,\dots,x_n\in\mathbb{R}^d$, $\ c_1,\dots,c_n\in\mathbb{R}$.
\[
\sum_{i,j=1}^n c_i c_j\, K(x_i,x_j)\ge 0
\]
By Bochner's theorem \citep{rudin1962fourier}, there exists a finite symmetric positive Borel measure \(\nu\) on \(\mathbb{R}^d\) such that
\begin{align*}
    K(x,y)&=\int_{\mathbb R^d}\psi_\omega(x)\psi_\omega(y)^*\,d\nu(\omega)=
\int_{\mathbb{R}^d} \cos\!\big(\omega^\top(x-y)\big)\,d\nu(\omega),
\end{align*}
where \(
\psi_\omega(x)=e^{i\omega^{\top}x}
\), \(^{*}\) denotes complex conjugation, and the sine term vanishes because \(K\) is symmetric and real-valued.
When \(\nu\) is a probability measure, this representation can be written as the expected inner product of the random Fourier feature \(\psi_\omega\)
\[
K(x,y)
=\mathbb E_{\omega\sim\nu}\big[\psi_\omega(x)\psi_\omega(y)^*\big].
\]
The random Fourier feature approximation proposed by
\cite{rahimi2007random}
replaces \(\nu\) by the empirical measure \(\frac{1}{m}\sum_{j=1}^m \delta_{\omega_j}\), where \(\omega_1,\dots,\omega_m\stackrel{\mathrm{i.i.d.}}{\sim}\nu\) are independent frequencies.
Define the finite-dimensional real-valued random feature map
\begin{align}\label{rff}
\widehat{\psi}(x)
&=
\frac{1}{\sqrt{m}}
\begin{pmatrix}
\cos(Wx)\\
\sin(Wx)
\end{pmatrix}
\in\mathbb{R}^{2m},
\end{align}
where \(W\in\mathbb{R}^{m\times d}\) is a random frequency matrix whose
\(j\)th row is \(\omega_j^\top\), for \(j=1,\ldots,m\), and the
trigonometric functions are applied componentwise. The Euclidean inner product of the finite-dimensional feature maps 
\begin{align*}
    \hat{K}(x,y)
&:=\left\langle \hat{\psi}(x),\hat{\psi}(y)\right\rangle_{\mathbb R^{2m}}=
\frac{1}{m}\sum_{j=1}^m
\cos\!\big(\omega_j^\top(x-y)\big)
\end{align*}
provides an estimator of
\(
K(x,y)
=
\mathbb E_{\omega\sim\nu}
\big[\cos\!\big(\omega^\top(x-y)\big)\big].
\)
For the Gaussian kernel
\[
K(x,y)=\exp\!\left(-\frac{\|x-y\|^2}{2\gamma^2}\right),
\]
the corresponding spectral measure is the Gaussian measure
\(
\nu=\mathcal N(0,\gamma^{-2}I),
\)
so that
\(
K(x,y)=\mathbb E_{\omega\sim \nu}
\big[\psi_\omega(x)\psi_\omega(y)^*\big].
\)
Since we use random Fourier features to approximate the original kernel, we will use the finite-dimensional random feature map \(\hat{\psi}\) in place of the original feature map \(\psi\) in Section~\ref{weighted-average} to define the detection statistic. The detection statistic is defined by
\begin{equation}
S_t^{\mathrm{RFF}}=\|\bar{Y}^{\mathrm{RFF}}_t\|^2,
\qquad 
    \bar{Y}^{\mathrm{RFF}}_t 
    = \sum_{i=1}^t\frac{a_{i,t}}{W_t}\,\hat \psi(X_i),
\label{eq:strff}    
\end{equation}

where \(\|\cdot\|\) denotes the Euclidean norm, \(\bar{Y}^{\mathrm{RFF}}_t\) is defined as \(\bar{Y}_t\) in \eqref{eqn:Ybardefn} with the original infinite-dimensional feature map \(\psi\) in Section~\ref{weighted-average} replaced by the finite dimensional random feature map \(\hat{\psi}\) defined in \eqref{rff}. 

\paragraph*{Algorithm.}
The complete procedure is presented in Algorithm~\ref{alg:aff-gd}.
It raises an alarm whenever the detection statistic enters the rejection
region \(\mathcal{R}_L\) defined in \eqref{thre}. For simplicity, Algorithm~\ref{alg:aff-gd} uses
\(y_t\) and \(\bar y_t\) in place of \(y_t^{\mathrm{RFF}}\) and
\(\bar y_t^{\mathrm{RFF}}\), respectively.
An alternative version of the algorithm uses the adaptive rejection region \(\mathcal{R}_{L, \rho, t}\) defined in \eqref{eqn:RLrhot}.
\begin{algorithm}[!htbp]
\caption{Adaptive kernel detector}
\label{alg:aff-gd}

\footnotesize
\begin{algorithmic}[1]

\Require Data stream \((x_t)_{t\geq 1}\subset\mathbb{R}^d\);
feature map \(\hat{\psi}\) defined in \eqref{rff};
initial forgetting factors \(\lambda_0,\lambda_1\in(0,1)\);
step size \(\eta>0\); bounds \(0<\ell_1<\ell_2<1\).

\State Initialize \(m_0=w_0=\nabla_v m_0=\nabla_v w_0=0\).

\For{\(t=1,2,\ldots\)}

    \State Observe \(x_t\) and set \(y_t=\hat{\psi}(x_t)\). \Comment Random Fourier
    Features step

    \State Update exponentially weighted sums:
    \(m_t=\lambda_{t-1}m_{t-1}+y_t\) and
    \(w_t=\lambda_{t-1}w_{t-1}+1\).

    \State Compute
    \(\bar y_t=m_t/w_t\) and
    \(S_t^{\mathrm{RFF}}=\lVert\bar y_t\rVert^2\).

    \State Update derivative recursions:
    \(\nabla_v m_t=m_{t-1}
      +\lambda_{t-1}\nabla_v m_{t-1}\),  \(\nabla_v w_t=w_{t-1}
      +\lambda_{t-1}\nabla_v w_{t-1}\).

    \State Compute
    \(\displaystyle
      \nabla_v\bar y_t
      =\frac{\nabla_v m_t-\bar y_t\nabla_v w_t}{w_t}\).
      \Comment Using \eqref{eqn:gradY}

    \If{\(t\geq 2\)}

        \State
        \(\displaystyle
        \nabla_v f_t
        =2\left\langle
        \nabla_v\bar y_{t-1},
        \bar y_{t-1}-y_t
        \right\rangle\).
        \Comment Using \eqref{eqn:dirder}

        \State
        \(\displaystyle
        \lambda_t=
        \Pi_{[\ell_1,\ell_2]}
        \left(\lambda_{t-1}-\eta\nabla_v f_t\right)\).
        \Comment Using \eqref{eqn:lambdat}

    \EndIf

    \If{\(S_t^{\mathrm{RFF}}\in\mathcal{R}_L\)}
        \State Raise an alarm and stop.
    \EndIf

\EndFor

\end{algorithmic}
\end{algorithm}

The experiments in Section~\ref{rffnumcom} in the Supplementary Material examine the effect of the number of random features on detection performance. Increasing the number of random features from $500$ to $5000$ slightly lowers the expected detection delay for OKAFF, but the performance is broadly similar across this range.

Moreover, Section~\ref{supp:sec:avruntime} in the Supplementary Material investigates the average runtime of OKAFF with other competitor methods. Figure~\ref{fig:runtime-total-comparison-d20} shows that OKAFF is very computationally efficient, with a runtime similar to NEWMA and lower than all other competitor methods, making it well-suited for online changepoint detection.

\subsection{Behaviour of the detection statistic}

\label{bemostatnullalt}

\begin{figure}[!htbp]
\centering

\includegraphics[
  width=0.32\linewidth
]{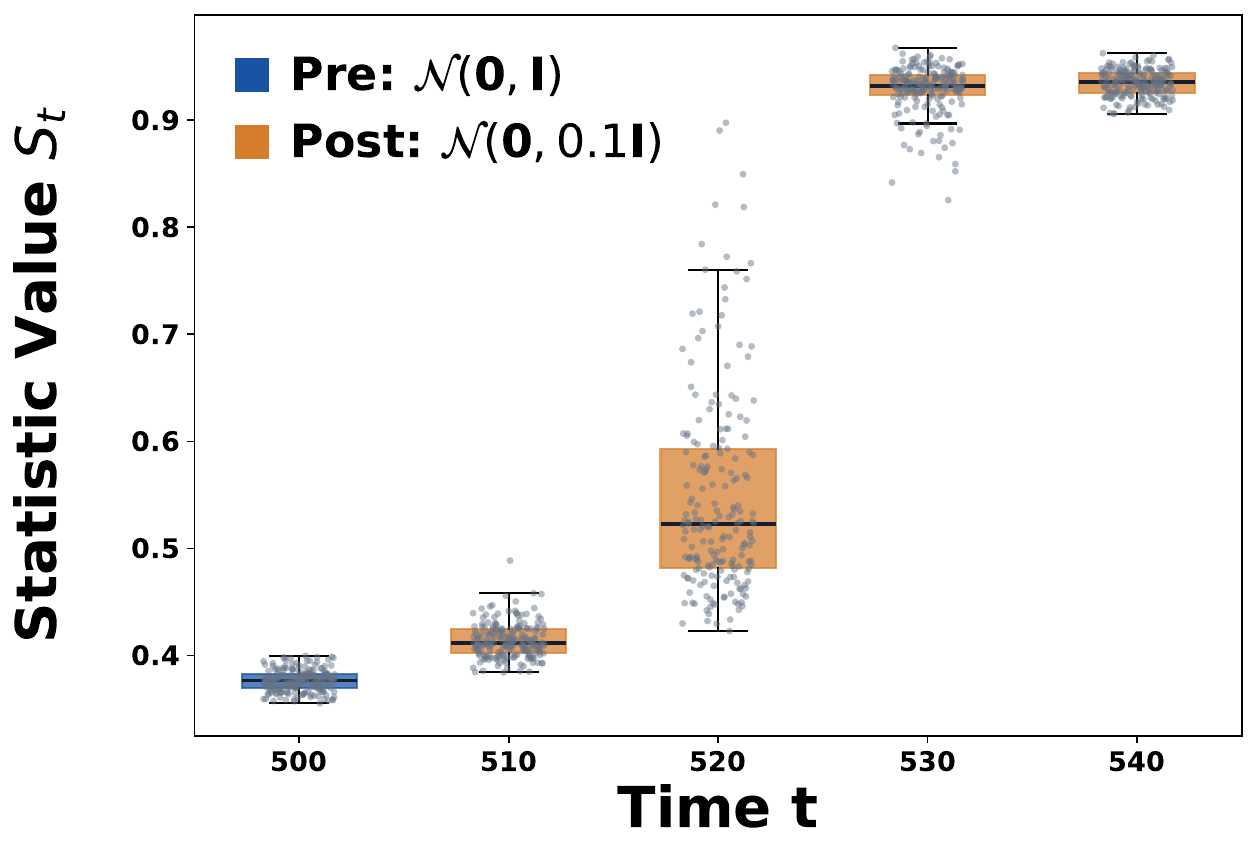}
\hfill
\includegraphics[
  width=0.32\linewidth
]{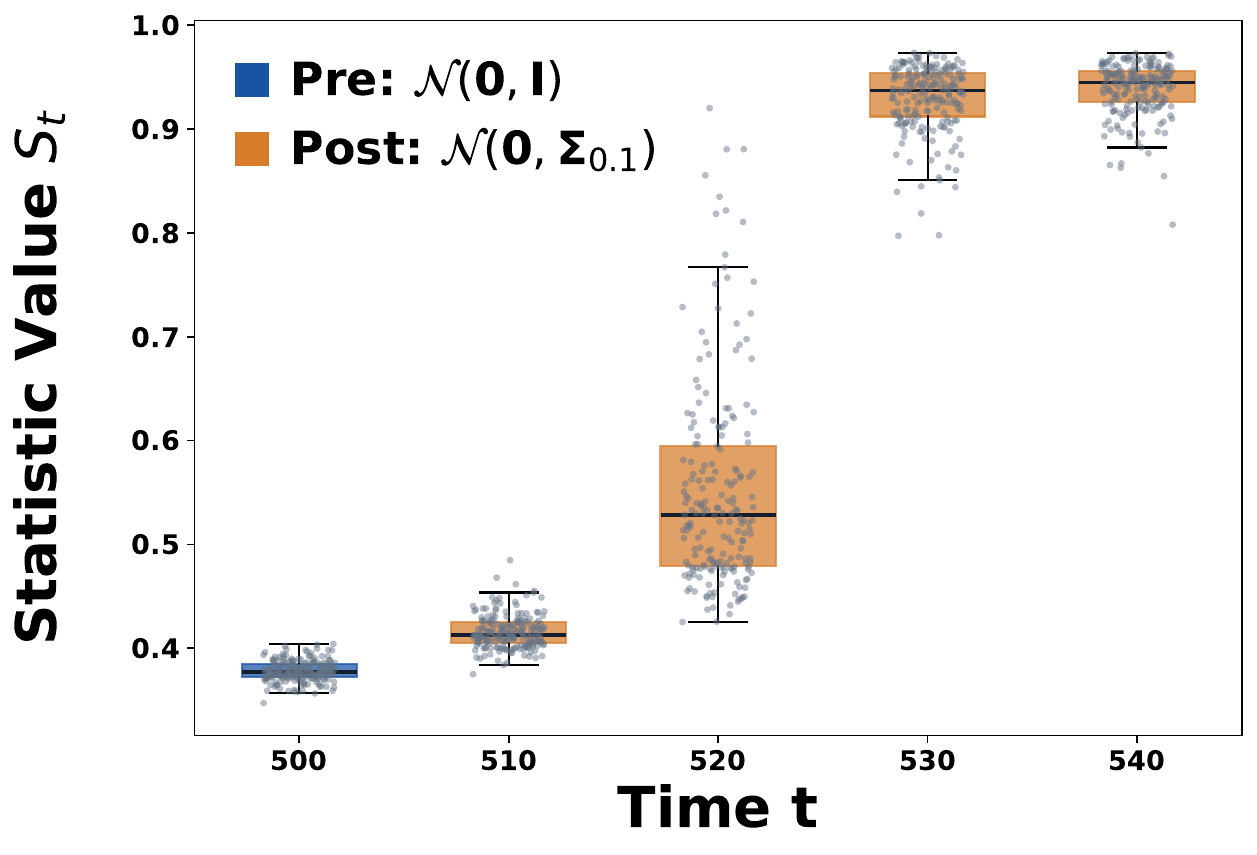}
\hfill
\includegraphics[
  width=0.32\linewidth
]{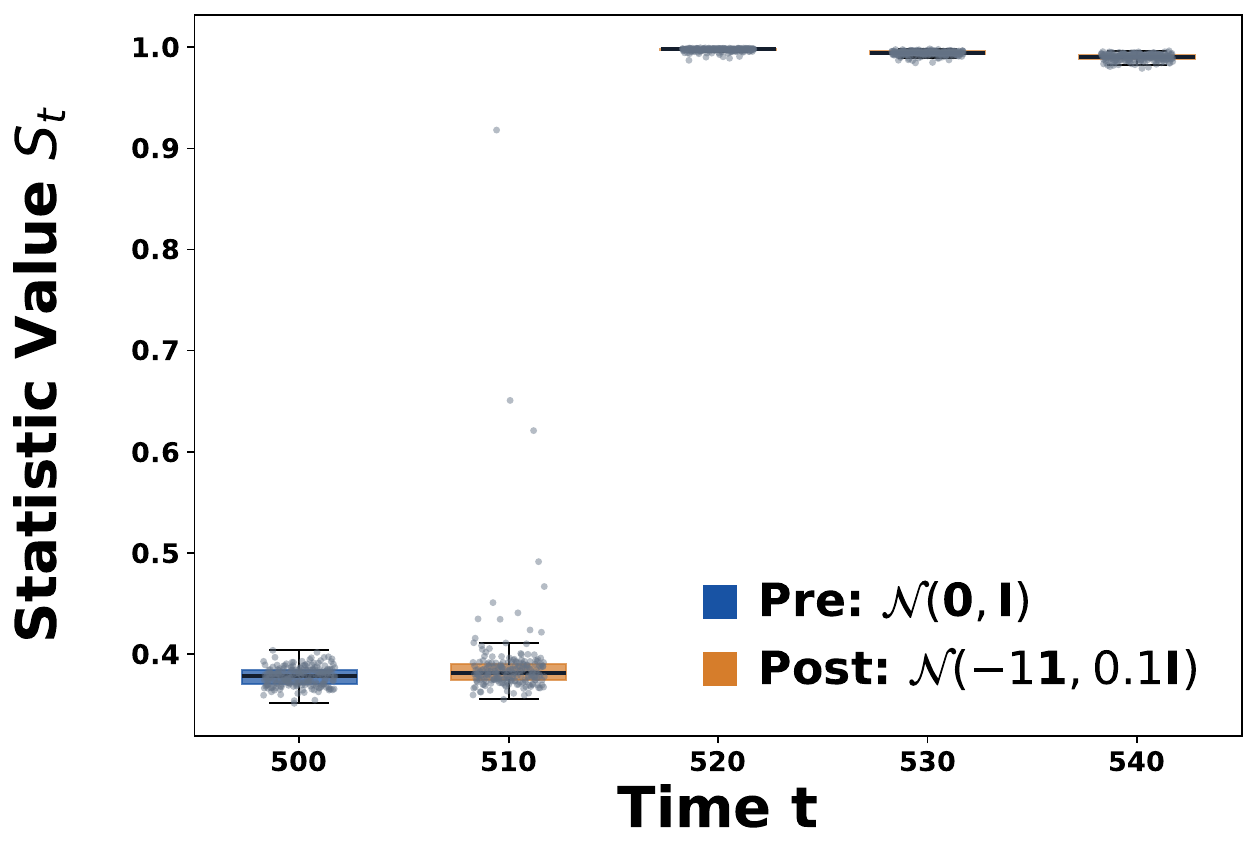}

\smallskip

\includegraphics[
  width=0.32\linewidth
]{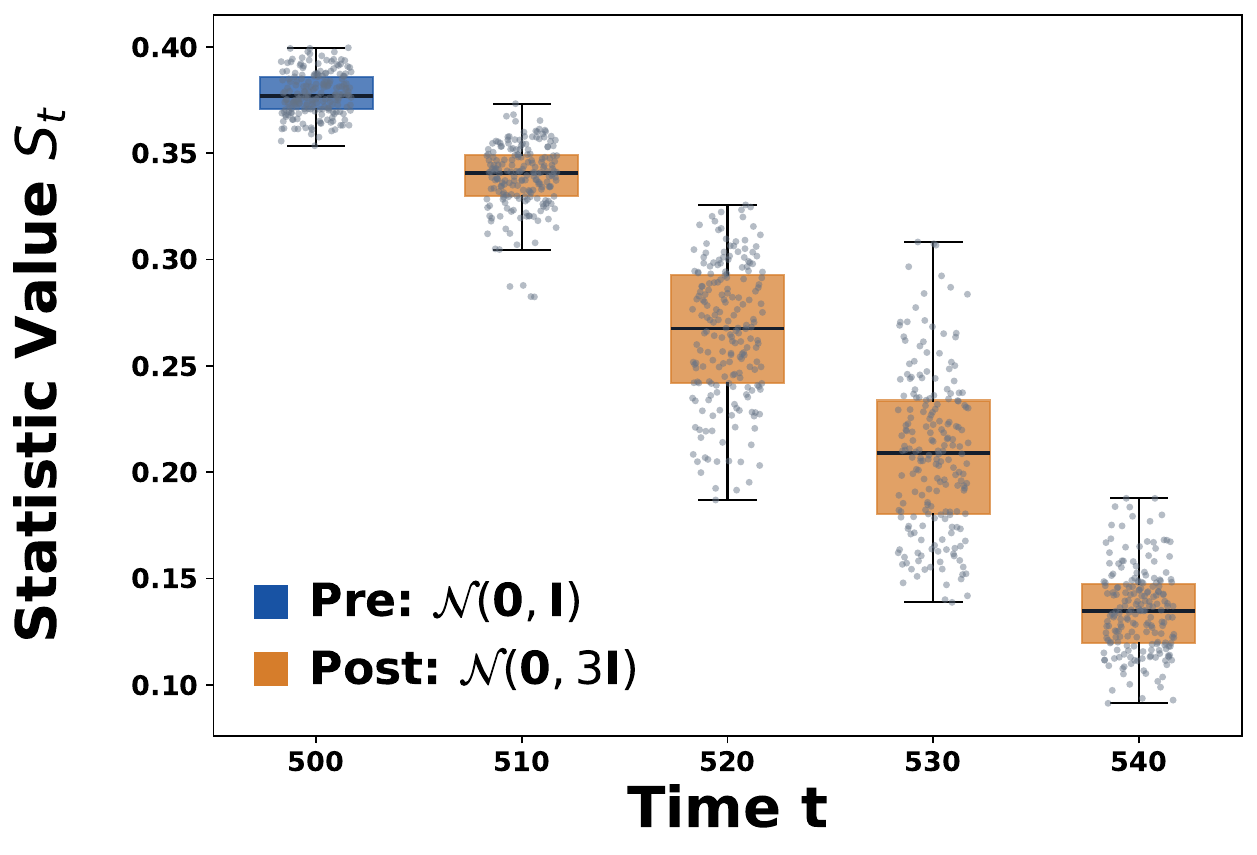}
\hfill
\includegraphics[
  width=0.32\linewidth
]{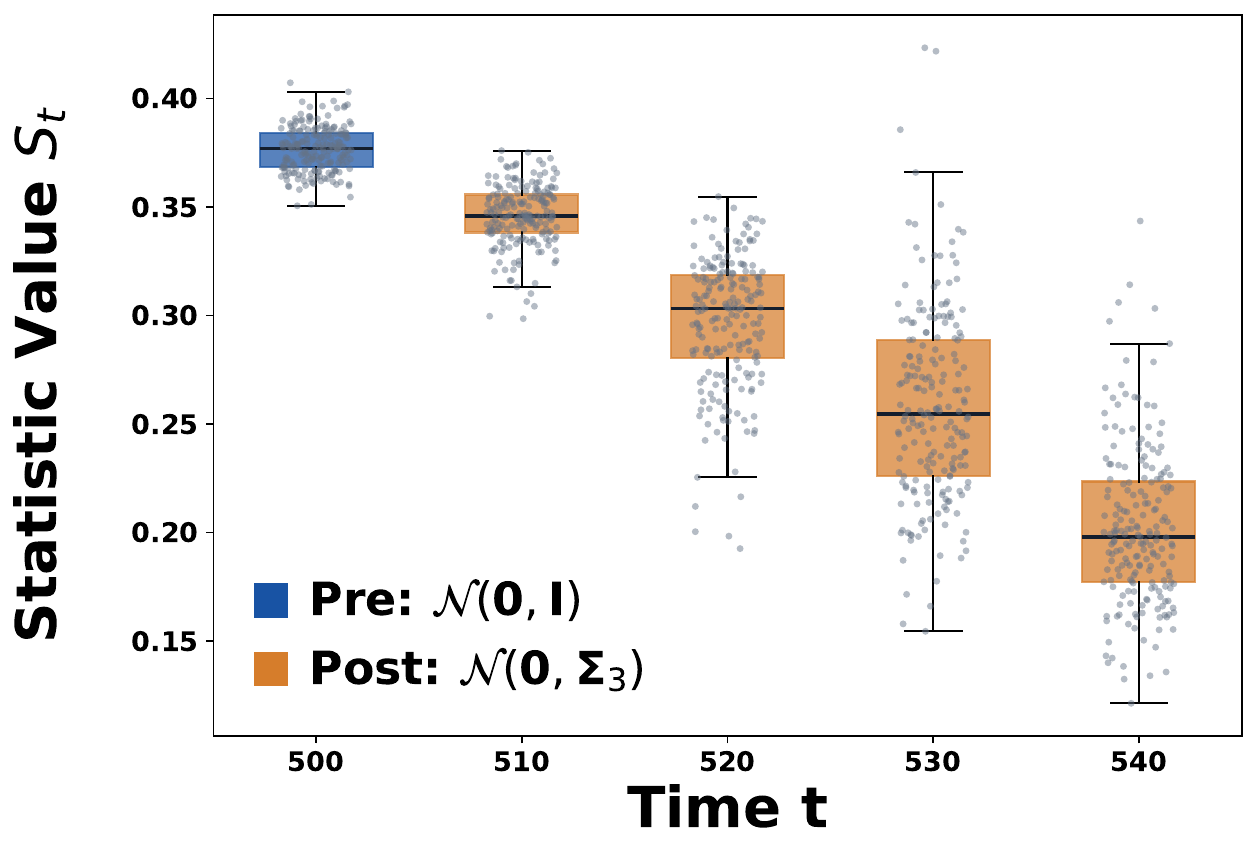}
\hfill
\includegraphics[
  width=0.32\linewidth
]{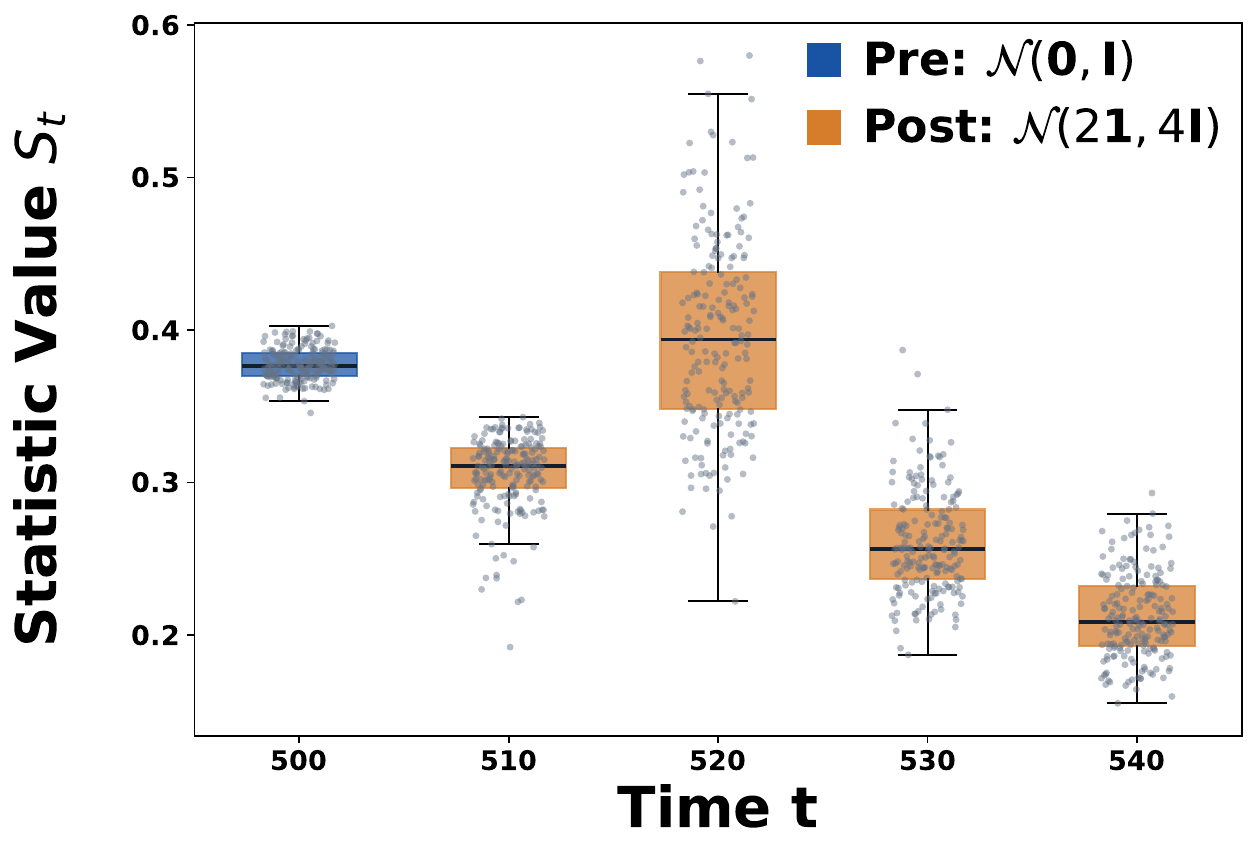}

\caption{Empirical distributions of the monitoring statistic under a
standard Gaussian pre-change distribution and six different Gaussian post-change distributions in dimension \(d=20\),
with the change occurring at \(t=501\). The top row presents, from left to
right, the post-change distributions \(\mathcal{N}(\mathbf{0},0.1I)\),
\(\mathcal{N}(\mathbf{0},\Sigma_{0.1})\), and
\(\mathcal{N}(-\mathbf{1},0.1I)\). The bottom row presents, from left to
right, the the post-change distributions \(\mathcal{N}(\mathbf{0},3I)\),
\(\mathcal{N}(\mathbf{0},\Sigma_3)\), and
\(\mathcal{N}(2\mathbf{1},4I)\).}
\label{fig:gaussian-boxplots-510-540}

\end{figure}

We investigate the behaviour of the detection statistic $S_t$ to see how it reacts to changes for a variety of pre-change and post-change underlying distributions. In all experiments, the changepoint occurs at \(t=501\) and  box plots are used to compare the pre-change values of $S_t$ at $t=500$ with its post-change values at times $t=510, 520, 530$, and $540$, for 100 Monte Carlo trials.
Box plots shaded blue show pre-change values, while orange box plots show post-change values of $S_t$.
A Gaussian kernel is used in all experiments, with its bandwidth selected via the median heuristic.


\par\medskip
\noindent\textbf{Gaussian pre-change distribution.}\enspace
Figure~\ref{fig:gaussian-boxplots-510-540} shows the behaviour of $S_t$ when 
we first fix the pre-change distribution as
\(
p=\mathcal{N}(\boldsymbol{0},I)
\) with \(d = 20\), 
and consider six different Gaussian post-change distributions: $\mathcal{N}(\boldsymbol{0},0.1I)$, $\mathcal{N}(\boldsymbol{0},\Sigma_{0.1})$, $\mathcal{N}(-\boldsymbol{1},0.1I)$, $\mathcal{N}(\boldsymbol{0},3I)$, $\mathcal{N}(\boldsymbol{0},\Sigma_{3})$, $\mathcal{N}(2\boldsymbol{1},4I)$.
Here,
\(
\Sigma_{0.1}
=0.05I+0.05\boldsymbol{1}\boldsymbol{1}^{\top},
\)
and 
\(
\Sigma_3
=1.5I+1.5\boldsymbol{1}\boldsymbol{1}^{\top}.
\)
The results for each post-change distribution are shown in a separate panel as time increases. While $S_t$ reacts more quickly to some changes than others, in all cases the empirical distribution of $S_t$ has changed significantly by $t=540$.

\begin{figure}[!htbp]
\centering

\includegraphics[
  width=0.48\linewidth
]{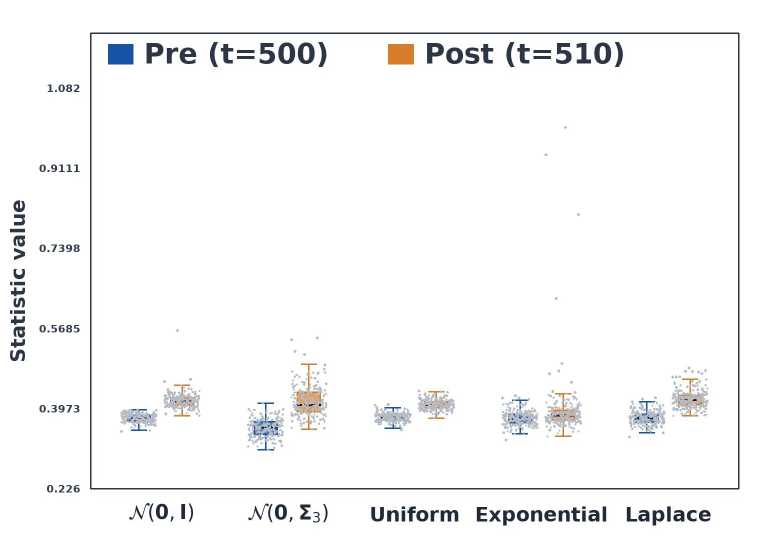}
\hfill
\includegraphics[
  width=0.48\linewidth
]{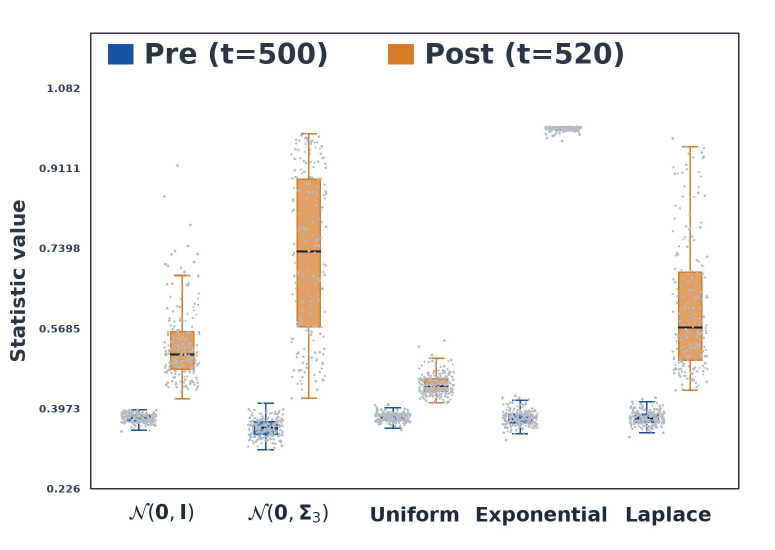}
\hfill
\includegraphics[
  width=0.48\linewidth
]{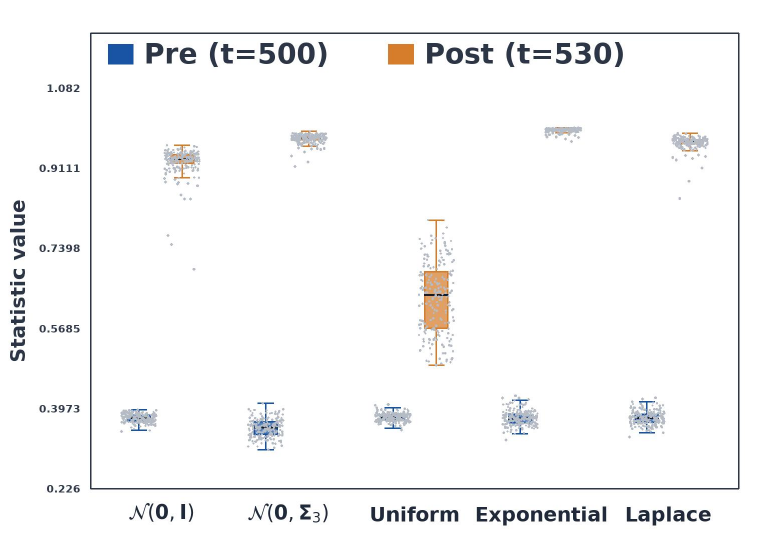}
\hfill
\includegraphics[
  width=0.48\linewidth
]{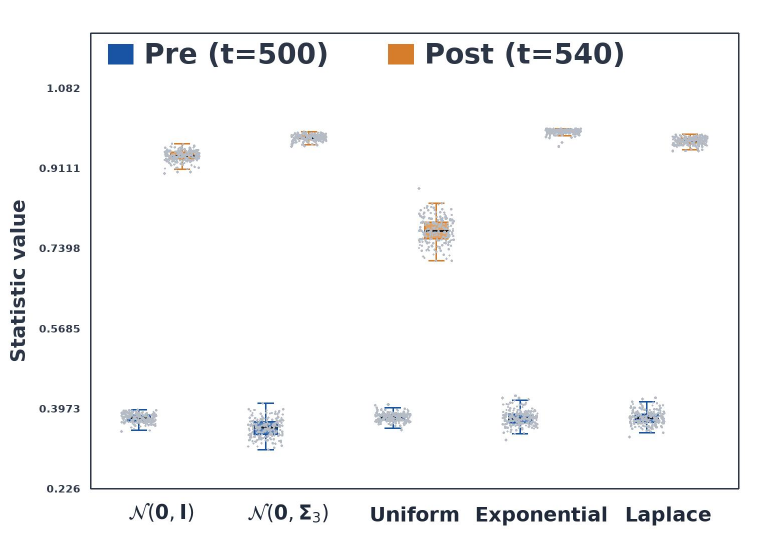}

\caption{Empirical distributions of the monitoring statistic under
different pre-change distributions and the same post-change distribution, with
the changepoint occurring at \(t=501\). The statistic under the pre-change
distribution is evaluated at \(t=500\), while the statistic under the
post-change distribution is evaluated at \(t=510\) (upper left), \(t=520\)
(upper right), \(t=530\) (bottom left) and \(t=540\) (bottom right).}
\label{fig:boxplot-thresholds-510-540-difpre}

\end{figure}

\par\medskip
\noindent\textbf{Different pre-change distributions.}\enspace 
Figure~\ref{fig:boxplot-thresholds-510-540-difpre} shows the behaviour of $S_t$ for five different pre-change distributions: 
$\mathcal{N}(\mathbf{0},I)$,
$\mathcal{N}(\mathbf{0},\Sigma_3)$,
$\operatorname{Uniform}([-1,1]^d)$,
$\operatorname{Exp}(1)^d$,
$\operatorname{Laplace}(0,1)^d$. Here, \(d=20\) and \(\boldsymbol{0},\boldsymbol{1}\in\mathbb{R}^{d}\) denote the zero
vector and the vector of ones, respectively, \(I\in\mathbb R^{d\times d}\) denotes the identity matrix, and
\(
\Sigma_3=1.5I+1.5\boldsymbol{1}\boldsymbol{1}^{\top}.
\)
The $\mathcal{N}(\mathbf{0},I)$ distribution is included as a reference, while the post-change distribution in all cases is \(\mathcal{N}(\mathbf{0},0.1I)\).
The layout for this figure is slightly different, and each panel shows the empirical post-change distributions of $S_t$ for a different time $t \in \{510, 520, 530, 540\}$, with the empirical pre-change distributions of $S_t$ at $t=500$ included alongside for reference. Again, the results show that $S_t$ is reactive to the changepoint in all cases.



We also plot the average values of $S_t$ in both scenarios as line plots in Figure~\ref{fig:average-gaussian-difpre-samplepath}. These figures collectively illustrate
how significantly the statistic differs between the pre-change and post-change regimes. 
This suggests that a properly calibrated threshold can yield high detection power while controlling false
alarms. 
Additional experimental results under different null distributions are reported in Section~\ref{exprappd} of Supplementary Material.

\begin{figure}[!htbp]
\centering

\includegraphics[
  width=0.49\linewidth
]{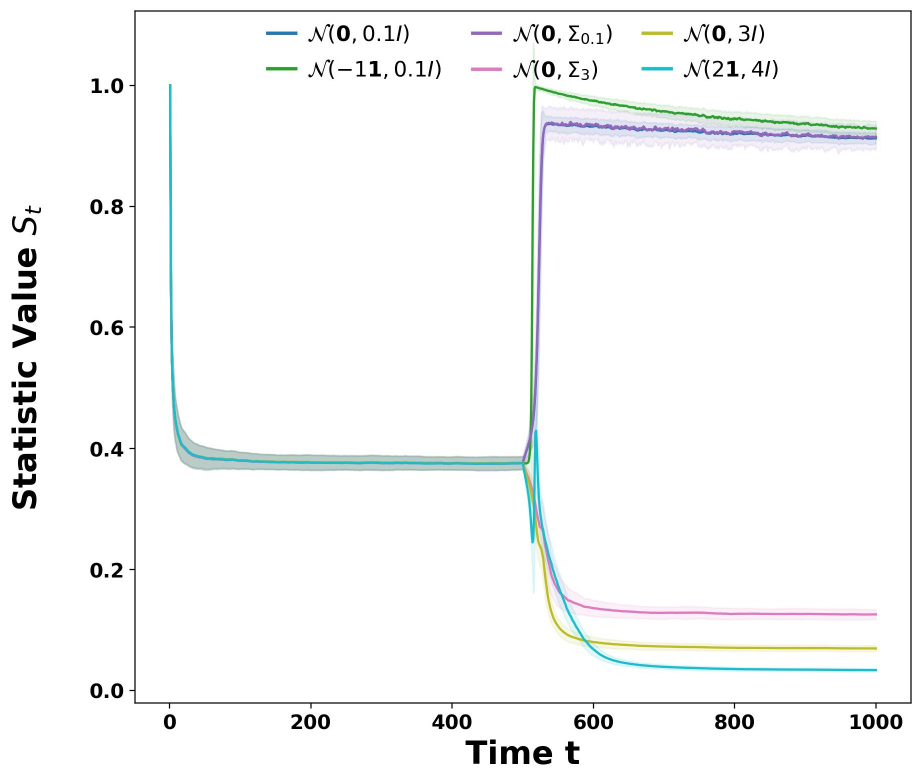}
\hfill
\includegraphics[
  width=0.49\linewidth
]{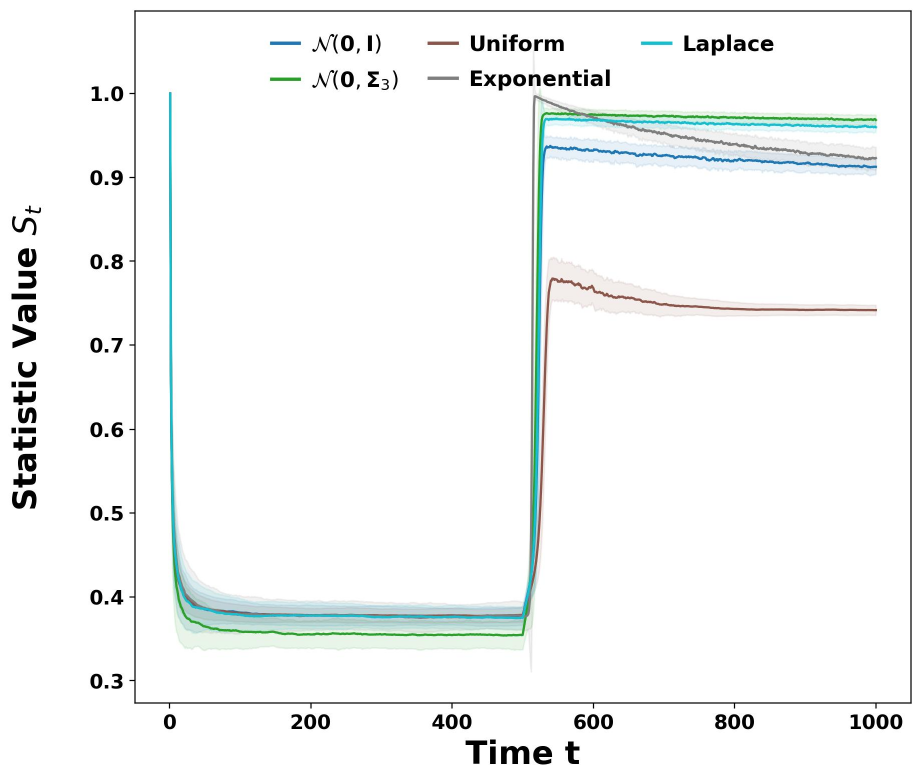}

\caption{Left: Sample paths of the monitoring statistic under a common Gaussian pre-change distribution
and different Gaussian post-change distributions, with the changepoint occurring at
\(t=501\). The panel shows the average sample paths with the corresponding
estimated standard-deviation bands. Right: Sample paths of the monitoring statistic under different pre-change 
distributions and a common post-change distribution, with the changepoint
occurring at \(t=501\). The panel shows the average sample paths together
with the corresponding estimated standard-deviation bands.}
\label{fig:average-gaussian-difpre-samplepath}

\end{figure}


\clearpage

\section{Experiments}\label{AFFexpr}
\textbf{Threshold calibration.} 
We use Theorem~\ref{cor:CLTffwave} as guidance to construct the threshold for the detection statistic $S_t^\mathrm{RFF}$ defined in \eqref{eq:strff}. Suppose that \(K\) is normalized so that
\(K(x,x)=1\). Under the no-change regime and using deterministic forgetting factors, with the detection statistic $S_t$ that does not use random Fourier features,
\[
\mathbb E S_t
=
\delta_t
+
\bigl(1-\delta_t\bigr)\theta_K,
\qquad
\theta_K:=\mathbb E_{X,Y\sim p}  K(X,Y),
\]
where \(X\) and \(Y\) are independent copies from \(p\), and
\(\delta_t\) is the sum of the squared normalized weights, as defined in
\eqref{deltaL}. Thus, \(\theta_K\) describes the average similarity between
two independent pre-change observations, whereas \(\delta_t\) describes the
concentration of the weighting scheme. The corresponding variance is
\[
\operatorname{Var}(S_t)
=
4(\delta_t-2\rho_t+\kappa_t)\nu_K
+
2(\delta_t^2-\kappa_t)\zeta_K,
\]
where
\(
\delta_t:=\sum_{i=1}^tb_{i,t}^2\),
\(\rho_t:=\sum_{i=1}^tb_{i,t}^3\),
\(\kappa_t:=\sum_{i=1}^tb_{i,t}^4\),
and 
\begin{align}
\nu_K
&:=
\left\langle\Sigma_p\mu_p,\mu_p
\right\rangle_{\mathcal H(K)}
=
\mathbb E_{X\sim p}[\mu_p(X)^2]
-\left(\mathbb E_{X,Y\sim p}K(X,Y)\right)^2,    
\nonumber \\
\zeta_K
&:=\mathbb E_{X,Y\sim p}[K(X,Y)^2]
-2\mathbb E_{X\sim p}[\mu_p(X)^2]+\left(\mathbb E_{X,Y\sim p}K(X,Y)\right)^2.
\nonumber
\end{align}
Here 
\(\mu_p:=\mathbb E_{Y\sim p}K(Y,\cdot)\), and
\(\mu_p(X):=\mathbb E_{Y\sim p}K(X,Y)\). The quantities \(\nu_K\) and
\(\zeta_K\) capture the variability of the kernel under the null
distribution, while \(\delta_t,\rho_t,\kappa_t\) capture the effect of the
forgetting weights. For a constant forgetting factor, \(\lambda_t=\lambda\) for all \(t\geq1\), the mean
and variance approach the limits $\lim_{t\to\infty}\mathbb E S_t
=\delta_\lambda+(1-\delta_\lambda)\theta_K$ and $\lim_{t\to\infty}\operatorname{Var}(S_t)
=4(\delta_\lambda-2\rho_\lambda+\kappa_\lambda)\nu_K+2(\delta_\lambda^2-\kappa_\lambda)\zeta_K$,
where
\begin{equation}
\delta_\lambda=\frac{1-\lambda}{1+\lambda},
\qquad
\rho_\lambda=\frac{(1-\lambda)^2}{1+\lambda+\lambda^2},
\qquad
\kappa_\lambda=
\frac{(1-\lambda)^3}{1+\lambda+\lambda^2+\lambda^3}.   
\label{eqn:deltarhokappalambda}
\end{equation}
Each quantity is the limit of its corresponding time-indexed weight
summary. A similar approach, using the constant forgetting factor case to guide the development of an adaptive forgetting factor procedure, was adopted in \cite{anagnostopoulos2012online}.

Under the no-change regime, but using random Fourier features, the distribution of \(S_t^{\mathrm{RFF}}\) becomes
approximately stable after the initial burn-in period, and the statistic then fluctuates
around its asymptotic null center. When the changepoint occurs, \(S_t^{\mathrm{RFF}}\) will deviate from the asymptotic null center. The numerical experiments exhibit this
behavior: before a change, the statistic remains concentrated around a stable
level \(u\), whereas after a change it is more likely to leave the calibrated
acceptance band; see Figure~\ref{fig:time-dependent-threshold} for an example. This observation supports the use of a fixed asymptotic center and
scale for long-run monitoring. The asymptotic argument in Theorem~\ref{cor:CLTffwave} determines the form of the threshold, but it does not
imply exact finite-sample false-alarm control for the adaptive procedure.
Moreover, Theorem~\ref{cor:CLTffwave} is stated for deterministic forgetting factors, whereas
the proposed method updates the value of the forgetting factor based on the data. We therefore treat the
asymptotic approximation as a calibration guide and select the final
control-limit multiplier \(L\) using pre-change data and Monte Carlo simulation.

Using the initial value of $\lambda=\lambda_0$ and the quantities in \eqref{eqn:deltarhokappalambda}, define the asymptotic null center and standard deviation by
\begin{align}
u&:=\delta_\lambda+(1-\delta_\lambda)\theta_K,\label{u}\\
\sigma&:=
\left[
4(\delta_\lambda-2\rho_\lambda+\kappa_\lambda)\nu_K
+
2(\delta_\lambda^2-\kappa_\lambda)\zeta_K
\right]^{1/2}.\label{v}
\end{align}
For a control-limit multiplier \(L>0\), we use the two-sided rejection region 
\begin{align}\label{thre}
\mathcal{R}_L
=
[0,u-L\sigma)\cup(u+L\sigma,\infty).
\end{align}
Equivalently, monitoring continues while \(S^{\mathrm{RFF}}_t\in[u-L\sigma,u+L\sigma]\). Increasing
\(L\) produces a wider acceptance region and fewer false alarms, but may also
increase detection delay. In practice, \(L\) is selected to achieve a desired
in-control ARL, using either pre-change observations and Monte Carlo calibration. 

For the Gaussian kernel and a Gaussian pre-change null distribution, let
\[
K(x,y)
=
\exp\left(-\frac{\|x-y\|^2}{2\gamma^2}\right),
\qquad
X\sim N(\mu,\Sigma).
\]
Then using Example~\ref{eg:gaussian},
\begin{align}\label{thetagauss}
    \theta_K
=
\det\left(I+\frac{2\Sigma}{\gamma^2}\right)^{-1/2},
\end{align}
and
\begin{align}\label{nugauss}
\nu_K
&=
\det\left[
\left(I+\frac{\Sigma}{\gamma^2}\right)
\left(I+\frac{3\Sigma}{\gamma^2}\right)
\right]^{-1/2}-
\det\left(I+\frac{2\Sigma}{\gamma^2}\right)^{-1},
\end{align}
and
\begin{align}\label{zetagauss}
\zeta_K
&=
\det\left(I+\frac{4\Sigma}{\gamma^2}\right)^{-1/2}-
2\det\left[
\left(I+\frac{\Sigma}{\gamma^2}\right)
\left(I+\frac{3\Sigma}{\gamma^2}\right)
\right]^{-1/2}+
\det\left(I+\frac{2\Sigma}{\gamma^2}\right)^{-1}.
\end{align}
Although the resulting thresholds defined by the rejection region $\mathcal{R}_{L}$ in \eqref{thre} may depend on the null distribution, we select the kernel bandwidth using the median heuristic (see, e.g., \cite{fukumizu2009kernel}), 
\begin{align*}  
\gamma=\sqrt{\frac{1}{2}\operatorname{median}\{\|X_i-X_j\|^2:1\le i <j\le n\}}.
\end{align*}


As demonstrated in the experiments below, this choice yields broadly similar thresholds across different null distributions. Nevertheless, the control limit $L$ must still be calibrated to ensure a prescribed in-control average run length. So far, we have considered fixed thresholds. In Section~\ref{adpthr}, we use Theorem~1 to guide the construction of the adaptive threshold. 
Below, we use $S_t$ to denote $S_t^{\mathrm{RFF}}$, since the OKAFF detection statistic will always use random Fourier features in practice.

\subsection{Simulation Study}\label{simu}

\subsubsection{Simulations with fixed thresholds}\label{fixthr}
The following experiments show the expected detection delay (EDD) versus the average run length (ARL) for different methods under different post-change alternatives; recall that these metrics are defined in Section~\ref{sec:prelimcpd}.
To illustrate the EDD at a given target ARL, we conduct a simulation study similar to the one described in \cite{kalinke2025optimal} and \cite{Wei2022OnlineKC}, and include our proposed method in the comparison using fixed thresholds.
All figures show the values of EDD on the vertical axis versus ARL on the horizontal axis. However, in order to make performance
differences visible over the wide range of ARL and EDD values considered, the axes are both log-scaled.


\paragraph*{Thresholds for ARL values} 
The methods in \cite{kalinke2025optimal,kalinke2025maximum,Wei2022OnlineKC,li2019scan} do not provide the procedure for selecting thresholds. Although NEWMA in \cite{keriven2020newma} provides guidance for threshold selection, it does not establish how to choose a threshold that attains a prescribed average run length or false-alarm probability. We therefore calibrate the threshold for each method separately using Monte Carlo simulations. For each method, we start by selecting a range of thresholds, and then run 200 Monte Carlo trials to determine the ARL values corresponding to these thresholds. In each trial, NEWMA, online RFF MMD, MMDEW and OKAFF each use 250 observations as a burn-in period, and OKCUSUM and ScanB use 250 observations as a reference sample, which is similar to a burn-in period. Note that the 250 observations in the burn-in period allow the test statistics to stabilize, before starting to monitor for a changepoint.


\paragraph*{Bandwidth parameter}
All methods use the 250 observations in the burn-in period (or reference sample) to set the bandwidth using the median heuristic.


\paragraph*{Other algorithm parameters} We choose the remaining algorithmic parameters as follows: we set \(B_{\min}=2\), \(B_{\max}=50\), and \(N=5\) for OKCUSUM; \(B=50\) and \(N=5\) for ScanB; and \(B=50\) for NEWMA. For OKAFF, we set the step size to \(\eta=0.001\) and $\lambda_0=\lambda_1=0.999$ (intial values of forgetting factors). NEWMA, Online RFF MMD, OKAFF each use \(500\) random features. All methods use a Gaussian kernel. We select the bandwidths for all methods using the same median heuristic.

\begin{figure}[htbp]
\centering

\begin{subfigure}{0.98\textwidth}
  \centering
  \includegraphics[
    width=\linewidth,
    height=0.2\textheight,
    keepaspectratio
  ]{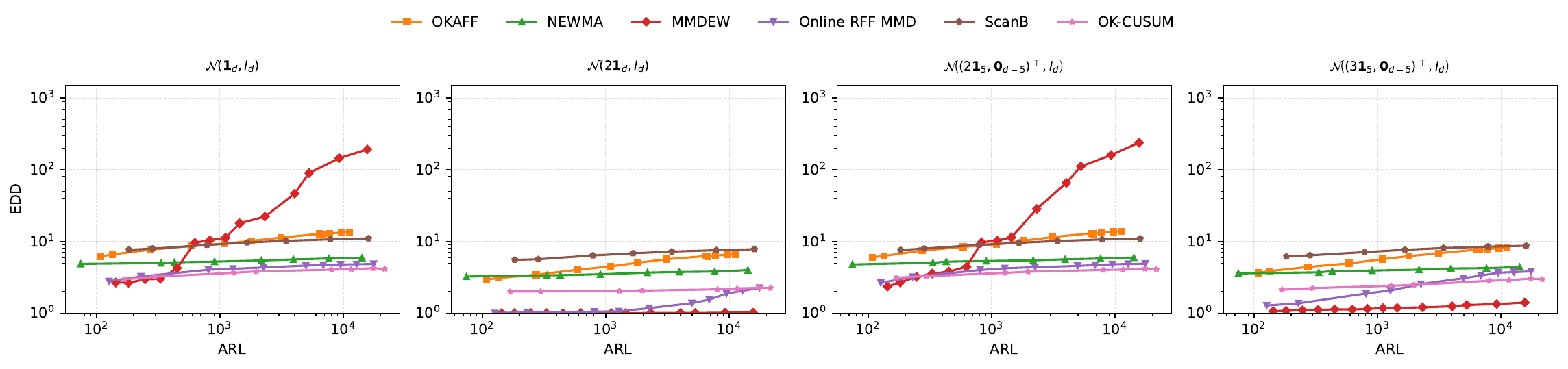}
  \caption{Dense and sparse mean shifts.}
  \label{fig:edd-arl-d20-mean}
\end{subfigure}

\medskip

\begin{subfigure}{0.98\textwidth}
  \centering
  \includegraphics[
    width=\linewidth,
    height=0.2\textheight,
    keepaspectratio
  ]{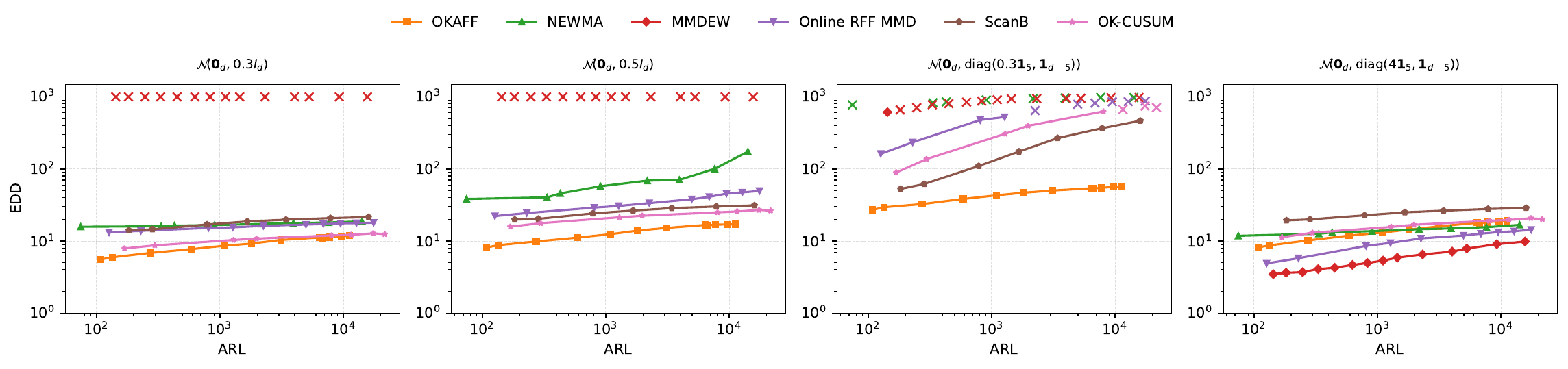}
  \caption{Dense and sparse covariance shifts.}
  \label{fig:edd-arl-d20-covariance}
\end{subfigure}

\medskip

\begin{subfigure}{0.98\textwidth}
  \centering
  \includegraphics[
    width=\linewidth,
    height=0.2\textheight,
    keepaspectratio
  ]{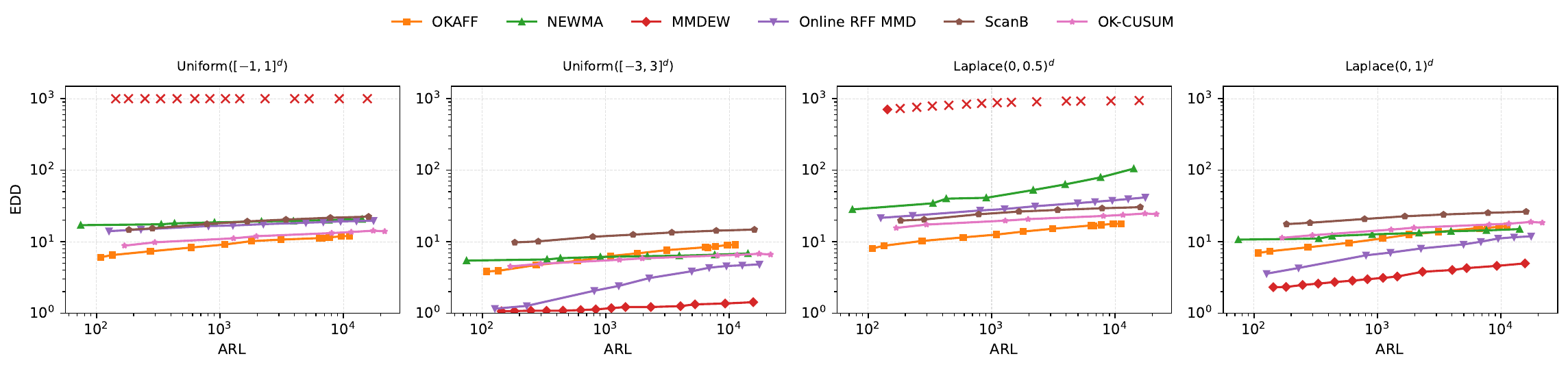}
  \caption{Uniform and Laplace alternatives.}
  \label{fig:edd-arl-d20-uniform-laplace}
\end{subfigure}

\caption{EDD versus ARL under different alternatives in dimension
\(d=20\), when the change occurs at \(k=100\). The three panels show
dense and sparse mean shifts, dense and sparse covariance shifts, and
uniform and Laplace alternatives, respectively.}
\label{fig:edd-arl-d20-mean-cov-unif-laplace}

\end{figure}

\clearpage 

We first conduct experiments for multivariate data to match the simulation setting of 
\citep{kalinke2025maximum,kazi2026online,Wei2022OnlineKC}, and then conduct experiments for univariate data.

\paragraph*{Multivariate data.}
We consider simulations with a \(20\)-dimensional standard normal pre-change distribution and evaluate a range of post-change alternatives, including mean shifts, covariance changes, Gaussian mixtures, uniform distributions, and Laplace distributions. Some of these distributional alternatives are also considered in \citep{kalinke2025maximum,kazi2026online,Wei2022OnlineKC}.  Complete distributional specifications and parameter settings are provided in Section~\ref{exprappd} of the Supplementary Material. At a comparable ARL, a lower EDD indicates faster detection. No single method dominates under every alternative. Instead, the relative
performance depends on both the type and magnitude of the distributional
change. Each EDD experiment is truncated after $1000$ post-change alternative
observations. A simulation replication is considered a failure if the method does not raise an alarm within this horizon. The failure rate is therefore the
proportion of the \(200\) simulation replications in which the change is not detected
within \(1000\) post-change observations. A cross in a plot marks a
configuration with a failure rate of at least \(40\%\), that is, the
method fails to detect the change within the simulation horizon in at least
\(80\) of the \(200\) replications. For failed replications, the detection delay is set to \(1000\), and and the EDD is computed as the average detection delay across all \(200\) replications, including both successful and failed detections. The experimental results are reported in Figure~\ref{fig:edd-arl-d20-mean-cov-unif-laplace}, with additional results presented in Figures~\ref{fig:edd-arl-covariance-shifts-d20}--\ref{fig:edd-arl-distribution-shifts-d20} in the Supplementary Material.

For both dense and sparse mean shifts, the EDD values of most methods are
approximately below \(10\), especially when the mean change is moderate or
large. This indicates that all of these methods detect mean changes rapidly,
and their performance is broadly similar under the settings considered here.
The main exception is MMDEW, whose performance is noticeably less stable. For
the weaker dense and sparse mean changes, its 
EDD increases sharply and can be
several orders of magnitude larger than those of the other methods. Thus,
although MMDEW performs well for sufficiently strong mean shifts, its detection
performance deteriorates substantially when the change in the mean is small.
Note that both axes are displayed on the log scale, so when the EDD increases above $10^1$, this represents a significant increase.

For covariance changes, the proposed method (OKAFF) performs well across all of the alternatives and
is the most stable detector among the methods considered. This advantage is particularly pronounced when the covariance scale decreases substantially, such as from \(I\) to \(0.3I\) or \(0.5I\). In these settings, MMDEW performs poorly, exhibiting very large detection delays and a high proportion of failed replications in which no change is detected within the simulation horizon. For the covariance alternatives
\[\operatorname{diag}(\underbrace{\delta,\ldots,\delta}_{5},\underbrace{1,\ldots,1}_{d-5},),\qquad \delta \in \{0.3,0.5\}, \qquad d=20,\]
the competing methods also exhibit comparatively large detection delays, whereas OKAFF achieves substantially shorter delays across the considered ARL range. These results indicate that OKAFF is more effective than the competing methods at detecting covariance changes and provides more stable performance across the covariance-change settings considered.

For the uniform and Laplace post-change alternative distributions,
OKAFF continues to perform well and achieves the smallest detection delay for
several alternatives. ScanB and OK-CUSUM also exhibit relatively stable
performance, but OKAFF has a lower EDD in most of the cases considered.
Moreover, the runtime results in Figure \ref{fig:runtime-total-comparison-d20} show that ScanB and OK-CUSUM have
substantially larger computational costs, which is an important practical
limitation for online detection. MMDEW is again less stable: although it achieves the smallest EDD for some alternatives, it exhibits very large detection delays and a high proportion of failed detections for others. For strong change signals, the EDD values of most methods are below approximately \(10\). While MMDEW performs poorly in some of the experiments above, 
Section~\ref{exprappd} in the Supplementary Material shows that, for the Gaussian-mixture alternatives, MMDEW generally produces the smallest
EDD, with Online RFF MMD and OKAFF also performing competitively for the more pronounced mixtures.

Overall, the multivariate \(d=20\) experiments demonstrate that OKAFF is a robust detector across a wide range of distributional changes. Its advantages are particularly evident for weak signals and covariance changes, which for some competing methods exhibit rapidly increasing detection delays or substantial failures. For sufficiently strong change signals, several methods detect the change almost immediately, and the differences in their detection performance consequently become less pronounced.

\begin{figure}[htbp]
\centering

\begin{subfigure}{0.98\textwidth}
  \centering
  \includegraphics[
    width=\linewidth,
    height=0.2\textheight,
    keepaspectratio
  ]{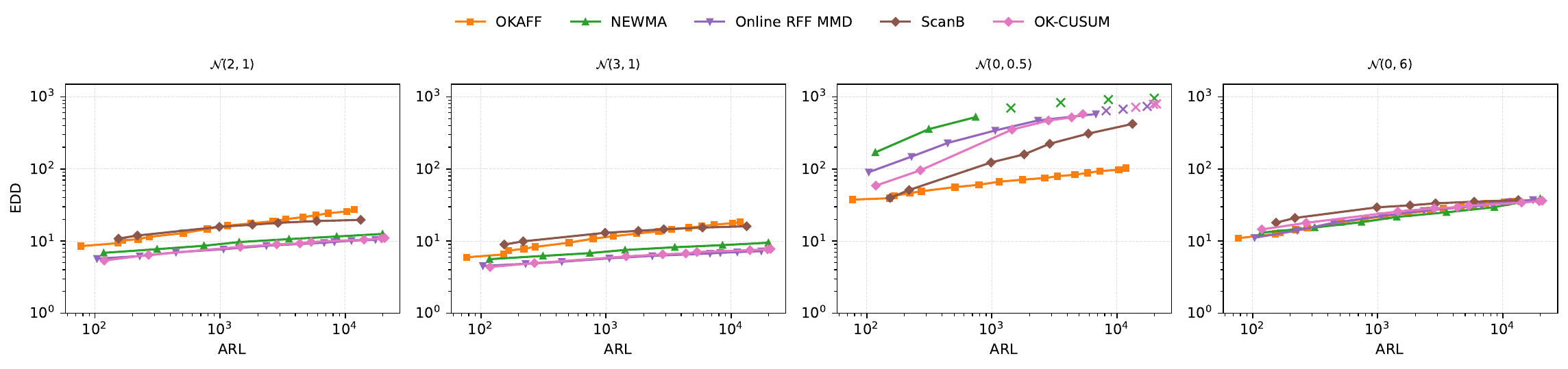}
  \caption{Mean and variance shifts.}
  \label{fig:edd-arl-d1-mean-variance}
\end{subfigure}

\medskip

\begin{subfigure}{0.98\textwidth}
  \centering
  \includegraphics[
    width=\linewidth,
    height=0.2\textheight,
    keepaspectratio
  ]{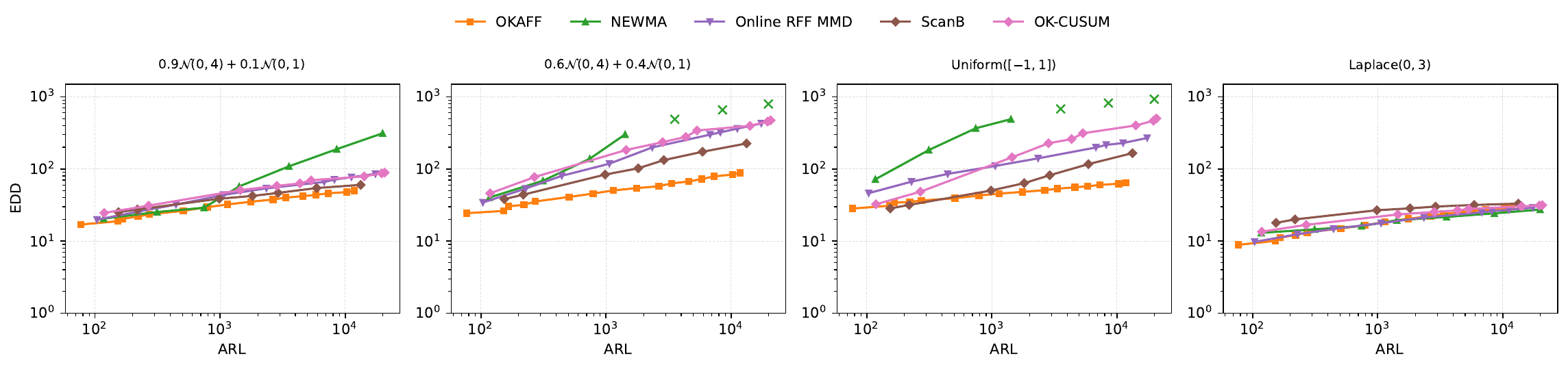}
  \caption{Gaussian-mixture, uniform, and Laplace alternatives.}
  \label{fig:edd-arl-d1-mixture-uniform-laplace}
\end{subfigure}

\caption{EDD versus ARL under different alternatives for univariate
data when the changepoint occurs at \(k=100\).}
\label{fig:edd-arl-d1-mean-cov-unif-laplace}

\end{figure}

\paragraph*{Univariate data.} Figure~\ref{fig:edd-arl-d1-mean-cov-unif-laplace} reports the EDD-versus-ARL results for the univariate setting under different alternatives, with additional results presented in Figures~\ref{fig:edd-arl-mean-shifts-d1}--\ref{fig:edd-arl-distribution-shifts-d1} in Section~\ref{exprappd} of the Supplementary Material.
The null distribution is \(p=\mathcal{N}(0,1)\). The alternatives include changes in the mean, changes in the variance, 
and more general distributional changes represented by Gaussian-mixture, uniform, and Laplace distributions as follows,
\begin{itemize}\setlength{\itemsep}{0pt}\setlength{\parskip}{0pt}
    \item \textbf{Mean shift:}
    \(\mathcal{N}(\delta,1)\), 
    \(\delta\in\{1,2,3,4\}\).

    \item \textbf{Variance shift:}
    \(\mathcal{N}(0,\delta)\), 
    \(\delta\in\{0.3,0.5,2,6\}\).

    \item \textbf{Gaussian mixture:}
    \((1-\delta)\mathcal{N}(0,4)+\delta\mathcal{N}(0,1)\), 
    \(\delta\in\{0.1,0.2,0.3,0.4\}\).

    \item \textbf{Laplace:}
    \(\operatorname{Laplace}(0,\delta)\),
    \(\delta\in\{0.5,1,3,4\}\).

    \item \textbf{Uniform:}
    \(\operatorname{Uniform}([-\delta,\delta])\), 
    \(\delta\in\{0.5,1,3,4\}\).
\end{itemize}
For these univariate experiments, MMDEW exhibited unstable behavior, which prevented reliable estimation of its ARL and EDD. We therefore exclude MMDEW from Figure~\ref{fig:edd-arl-d1-mean-cov-unif-laplace} .

All methods detect moderate and large mean shifts rapidly. For variance changes, OKAFF provides the most consistently low delays. OKAFF also performs well for Gaussian-mixture alternatives and for the uniform and Laplace alternatives, whereas NEWMA and Online RFF MMD deteriorate substantially in several alternatives. When the change is large, the performance differences among the methods become less pronounced. Overall, OKAFF provides a favorable balance of detection speed and stability across the univariate alternatives.

\subsubsection{Simulations with adaptive thresholds} 
\label{sec:adapthreshsimstudy}

The simulations in Section~\ref{fixthr} consider the case where there is a single changepoint and one can use a burn-in period of $250$ observations. In many real-world applications, there are (i) multiple changepoints and (ii) fewer observations between changepoints, which would preclude the possibility of assuming we can use a long burn-in period to estimate parameters and stabilize the statistics.

We therefore consider the adaptive threshold approach proposed in the NEWMA method \citep{keriven2020newma}, which allows the threshold to evolve continuously during online monitoring. This approach is particularly useful when the pre-change distribution is unknown and only a limited number of pre-change observations are available, or when multiple changepoints may occur, as the threshold can adapt to the evolving behavior of the monitoring statistic. 

\begingroup
\scriptsize
\setlength{\tabcolsep}{2.5pt}
\renewcommand{\arraystretch}{0.95}
\begin{longtable}{@{}>{\hspace{0.4em}}p{0.265\textwidth}p{0.11\textwidth}*{4}{>{\centering\arraybackslash}p{0.145\textwidth}}@{}}
\caption{Performance comparison of four methods under 20-dimensional
post-change alternatives, with the in-control ARL constrained to
\(500\pm40\). Bold entries indicate the smallest EDD, highest Success, and the lowest False alarm and Failure proportions.}
\label{tab:main:d20-mid-arl-edd-sff}\\
\toprule
& Metric & OKAFF & NEWMA & Online RFF MMD & MMDEW\\
\midrule
\endfirsthead
\multicolumn{6}{c}{\tablename~\thetable{} continued}\\
\toprule
& Metric & OKAFF & NEWMA & Online RFF MMD & MMDEW\\
\midrule
\endhead
\midrule
\multicolumn{6}{r}{Continued on next page}\\
\endfoot
\bottomrule
\endlastfoot

\multicolumn{6}{@{}l}{\hspace{0.4em}\bfseries\boldmath
\(\mathcal{N}(\boldsymbol{0},I)
\to\mathcal{N}(\delta\boldsymbol{1},I)\)}\\
\addlinespace[2pt]

\(\delta=1\) & EDD
& \makebox[5.0em][l]{\makebox[3.4em][r]{9.11}}
& \makebox[5.0em][l]{\makebox[3.4em][r]{764.28}}
& {\bfseries \makebox[5.0em][l]{\makebox[3.4em][r]{4.07}}}
& \makebox[5.0em][l]{\makebox[3.4em][r]{179.08}}\\
& Success
& {\bfseries \makebox[5.0em][l]{\makebox[3.4em][r]{0.950}}}
& \makebox[5.0em][l]{\makebox[3.4em][r]{0.582}}
& \makebox[5.0em][l]{\makebox[3.4em][r]{0.912}}
& \makebox[5.0em][l]{\makebox[3.4em][r]{0.782}}\\
& False alarm
& {\bfseries \makebox[5.0em][l]{\makebox[3.4em][r]{0.050}}}
& \makebox[5.0em][l]{\makebox[3.4em][r]{0.414}}
& \makebox[5.0em][l]{\makebox[3.4em][r]{0.088}}
& \makebox[5.0em][l]{\makebox[3.4em][r]{0.142}}\\
& Failure
& {\bfseries \makebox[5.0em][l]{\makebox[3.4em][r]{0.000}}}
& \makebox[5.0em][l]{\makebox[3.4em][r]{0.004}}
& {\bfseries \makebox[5.0em][l]{\makebox[3.4em][r]{0.000}}}
& \makebox[5.0em][l]{\makebox[3.4em][r]{0.076}}\\

\addlinespace[3pt]

\(\delta=-4\) & EDD
& \makebox[5.0em][l]{\makebox[3.4em][r]{3.39}}
& \makebox[5.0em][l]{\makebox[3.4em][r]{71.74}}
& \makebox[5.0em][l]{\makebox[3.4em][r]{1.01}}
& {\bfseries \makebox[5.0em][l]{\makebox[3.4em][r]{1.00}}}\\
& Success
& {\bfseries \makebox[5.0em][l]{\makebox[3.4em][r]{0.934}}}
& \makebox[5.0em][l]{\makebox[3.4em][r]{0.580}}
& \makebox[5.0em][l]{\makebox[3.4em][r]{0.888}}
& \makebox[5.0em][l]{\makebox[3.4em][r]{0.844}}\\
& False alarm
& {\bfseries \makebox[5.0em][l]{\makebox[3.4em][r]{0.066}}}
& \makebox[5.0em][l]{\makebox[3.4em][r]{0.414}}
& \makebox[5.0em][l]{\makebox[3.4em][r]{0.112}}
& \makebox[5.0em][l]{\makebox[3.4em][r]{0.156}}\\
& Failure
& {\bfseries \makebox[5.0em][l]{\makebox[3.4em][r]{0.000}}}
& \makebox[5.0em][l]{\makebox[3.4em][r]{0.006}}
& {\bfseries \makebox[5.0em][l]{\makebox[3.4em][r]{0.000}}}
& {\bfseries \makebox[5.0em][l]{\makebox[3.4em][r]{0.000}}}\\

\addlinespace[4pt]

\multicolumn{6}{@{}l}{\hspace{0.4em}\bfseries\boldmath
\(\mathcal{N}(\boldsymbol{0},I)
\to\mathcal{N}(\boldsymbol{0},\delta I)\)}\\
\addlinespace[2pt]

\(\delta=0.3\) & EDD
& {\bfseries \makebox[5.0em][l]{\makebox[3.4em][r]{12.67}}}
& \makebox[5.0em][l]{\makebox[3.4em][r]{861.44}}
& \makebox[5.0em][l]{\makebox[3.4em][r]{25.03}}
& \makebox[5.0em][l]{\makebox[3.4em][r]{671.29}}\\
& Success
& {\bfseries \makebox[5.0em][l]{\makebox[3.4em][r]{0.942}}}
& \makebox[5.0em][l]{\makebox[3.4em][r]{0.592}}
& \makebox[5.0em][l]{\makebox[3.4em][r]{0.880}}
& \makebox[5.0em][l]{\makebox[3.4em][r]{0.828}}\\
& False alarm
& {\bfseries \makebox[5.0em][l]{\makebox[3.4em][r]{0.058}}}
& \makebox[5.0em][l]{\makebox[3.4em][r]{0.386}}
& \makebox[5.0em][l]{\makebox[3.4em][r]{0.106}}
& \makebox[5.0em][l]{\makebox[3.4em][r]{0.168}}\\
& Failure
& {\bfseries \makebox[5.0em][l]{\makebox[3.4em][r]{0.000}}}
& \makebox[5.0em][l]{\makebox[3.4em][r]{0.022}}
& \makebox[5.0em][l]{\makebox[3.4em][r]{0.014}}
& \makebox[5.0em][l]{\makebox[3.4em][r]{0.004}}\\

\addlinespace[3pt]

\( \delta=2.5\) & EDD
& \makebox[5.0em][l]{\makebox[3.4em][r]{6.33}}
& \makebox[5.0em][l]{\makebox[3.4em][r]{808.32}}
& \makebox[5.0em][l]{\makebox[3.4em][r]{4.61}}
& {\bfseries \makebox[5.0em][l]{\makebox[3.4em][r]{2.06}}}\\
& Success
& {\bfseries \makebox[5.0em][l]{\makebox[3.4em][r]{0.938}}}
& \makebox[5.0em][l]{\makebox[3.4em][r]{0.618}}
& \makebox[5.0em][l]{\makebox[3.4em][r]{0.906}}
& \makebox[5.0em][l]{\makebox[3.4em][r]{0.896}}\\
& False alarm
& {\bfseries \makebox[5.0em][l]{\makebox[3.4em][r]{0.062}}}
& \makebox[5.0em][l]{\makebox[3.4em][r]{0.382}}
& \makebox[5.0em][l]{\makebox[3.4em][r]{0.094}}
& \makebox[5.0em][l]{\makebox[3.4em][r]{0.102}}\\
& Failure
& {\bfseries \makebox[5.0em][l]{\makebox[3.4em][r]{0.000}}}
& {\bfseries \makebox[5.0em][l]{\makebox[3.4em][r]{0.000}}}
& {\bfseries \makebox[5.0em][l]{\makebox[3.4em][r]{0.000}}}
& \makebox[5.0em][l]{\makebox[3.4em][r]{0.002}}\\

\end{longtable}
\endgroup

Theorem~\ref{cor:CLTffwave} provides further motivation for applying this adaptive-threshold approach to our proposed statistic $S_t$ for OKAFF. Specifically, under a Gaussian kernel and a deterministic forgetting factor, the monitoring statistic converges to a Gaussian limit under the null hypothesis. This result motivates constructing the threshold using the evolving mean and variance of the statistic. We therefore implement an adaptive threshold that is updated continuously throughout the monitoring process for OKAFF, NEWMA, Online RFF MMD and MMDEW. Table~\ref{tab:main:d20-mid-arl-edd-sff} below shows some results for these methods when a burn-in period of only $50$ is used, and there is a single changepoint occurring after 100 observations, with a total sequence length of $1100$. Since ScanB and OK-CUSUM require a large amount of pre-change data to construct a reference block, we omit them from this experiment. All methods are given adaptive thresholds to ensure their ARL is approximately $500$. 
The performance metric reported are \(\operatorname{EDD}\), False alarm  (proportion of sequences with an alarm raised before the true
changepoint) Success (the proportion of successful detections, with no false alarm and the changepoint is detected within $1000$ post-change observations), and Failure (failure to detect the
change within 1000 post-change observations). Note that the values of Success, False alarm and Failure sum to $1$. The adaptive threshold construction and additional simulation scenarios, including for non-Gaussian data, for both univariate and multivariate data, are provided in Section~\ref{adpthr} of Supplementary Material.
Table~\ref{tab:main:d20-mid-arl-edd-sff} shows that OKAFF performs well across scenarios considering changes in mean and variance. We also notice that NEWMA does not perform well with a burn-in period of $50$. Online RFF MMD has a smaller EDD when there is a small change in the mean ($\delta=1$) and MMDEW has a smaller EDD when there is a large change in the mean or variance, but a high EDD for small changes. However, comparing across all scenarios shown in the tables in Section~\ref{adpthr} of Supplementary Material, OKAFF generally has the best overall performance, with comparitively low EDD and high Success, and low False alarm and low Failure.

\subsubsection{Learning the threshold from the data: a nonparametric approach}

In Section~\ref{fixthr}, the fixed threshold was learned using Theorem~\ref{cor:CLTffwave} and assuming the pre-change distribution was Gaussian. However, Figure~\ref{fig:boxplot-thresholds-510-540-difpre} shows that the statistic will reflect any changes for a variety of pre-change distributions. This is because the cost function in \eqref{costfun} measures a difference of distributions, sharing the same principle as MMD. If we would like to use the fixed threshold approach, without assuming knowledge of the pre-change distribution, we can use the adaptive threshold approach to learn a threshold from the burn-in period, and then use this as a fixed threshold for the monitoring period.

\subsection{Real data experiments}
In the real data experiments, we use the adaptive threshold briefly described in the previous section, and described in more detail in Section~\ref{adpthr} of Supplementary Material. 

In the single-changepoint experiments, we evaluate the performance of each changepoint detection method using the average detection delay for a given in-control ARL. For the multiple-changepoint and real-data experiments, the $F_1$-score is a more appropriate performance measure. We therefore use the datasets and $F_1$-score evaluation framework introduced in \cite{Burg2020AnEO}. These datasets were independently annotated by multiple human annotators, who identified the presence and locations of perceived changepoints. For each method, the first 50 observations are used to estimate the Gaussian kernel bandwidth using the median heuristic. These observations also serve as burn-in data, and monitoring of the detection statistic begins from the 51st observation onward.

Let $T_k$ be the set of changepoints supplied by annotator $k$, let $X$ be the set of predicted changepoints, and let
\begin{equation*}
    T^{*}=\bigcup_{k=1}^{K}T_k
\end{equation*}
be the union of the $K$ annotation sets.  Following the implementation in \cite{Burg2020AnEO}, the
initial location $0$ is added to every $T_k$ and to $X$ before evaluation. Let $\operatorname{TP}(T,X;\Delta)$ denote the set of true changepoints in $T$
that can be matched one-to-one to predictions in $X$ within the tolerance
margin $\Delta$.  A prediction cannot match more than one true changepoint; when
several predictions are eligible, the closest prediction is used.  The experiment use $\Delta=5$ observations as in \cite{Burg2020AnEO}. In this experiment, only annotations occurring after the first 45 observations are retained because monitoring starts after the first 50 data and the matching tolerance is $\pm5$ observations. Precision, annotation-averaged recall and $F_1$-score
are then defined by
\begin{align*}
    P &= \frac{\lvert\operatorname{TP}(T^{*},X;5)\rvert}{\lvert X\rvert},\\
    R &= \frac{1}{K}\sum_{k=1}^{K}
         \frac{\lvert\operatorname{TP}(T_k,X;5)\rvert}{\lvert T_k\rvert}, \\
    F_1&=\frac{2PR}{P+R}.
\end{align*}

Adding location 0 makes the denominators nonzero and treats the start of the
series as a changepoint shared by the annotations and predictions.  Thus a
dataset with no annotated changes and no alarms receives $P=R=F_1=1$.

\begin{longtable}{llrrrr}
\caption{Per-dataset Best F1 score comparison via parameters grid search. The best
score in each dataset is bold.}
\label{tab:oracle-f1-dataset}\\
\toprule
Dataset & Type & OKAFF & NEWMA & MMDEW & Online RFF MMD \\
\midrule
\endfirsthead
\caption[]{Per-dataset Best F1 score comparison via parameters grid search. (continued).}\\
\toprule
Dataset & Type & OKAFF & NEWMA & MMDEW & Online RFF MMD \\
\midrule
\endhead
\midrule
\multicolumn{6}{r}{Continued on next page}\\
\endfoot
\bottomrule
\endlastfoot
bank$^{\ddagger}$ & U & \textbf{1.0000} & \textbf{1.0000} & \textbf{1.0000} & \textbf{1.0000} \\
bitcoin & U & 0.4791 & \textbf{0.5063} & 0.4496 & 0.4496 \\
brent\_spot & U & \textbf{0.5437} & 0.3607 & 0.3392 & 0.3152 \\
businv & U & \textbf{0.9189} & 0.5882 & 0.5882 & 0.5882 \\
children\_per\_woman & U & 0.5075 & 0.5075 & 0.5075 & \textbf{0.6784} \\
co2\_canada & U & 0.5027 & 0.3610 & \textbf{0.6905} & 0.3610 \\
construction & U & 0.6957 & \textbf{0.8889} & 0.6957 & 0.6957 \\
gdp\_argentina & U & 0.8889 & 0.8889 & \textbf{1.0000} & 0.8889 \\
global\_co2 & U & \textbf{0.8462} & \textbf{0.8462} & \textbf{0.8462} & \textbf{0.8462} \\
homeruns & U & \textbf{0.9474} & 0.6957 & 0.8679 & 0.6957 \\
iceland\_tourism & U & 0.9474 & \textbf{1.0000} & 0.9474 & 0.9474 \\
jfk\_passengers & U & 0.7234 & 0.7234 & \textbf{0.9091} & 0.7234 \\
lga\_passengers & U & \textbf{0.8468} & 0.5386 & 0.5685 & 0.5386 \\
measles & U & \textbf{0.9474} & \textbf{0.9474} & \textbf{0.9474} & \textbf{0.9474} \\
nile$^{\dagger}$ & U & \textbf{1.0000} & \textbf{1.0000} & \textbf{1.0000} & \textbf{1.0000} \\
ratner\_stock & U & \textbf{0.6496} & 0.5714 & 0.5714 & 0.5714 \\
scanline\_126007 & U & \textbf{0.7547} & 0.6523 & 0.7066 & 0.6444 \\
scanline\_42049 & U & 0.7016 & 0.3849 & \textbf{0.7619} & 0.3018 \\
seatbelts & U & 0.6829 & \textbf{0.8235} & \textbf{0.8235} & \textbf{0.8235} \\
shanghai\_license & U & \textbf{0.9655} & 0.6496 & 0.6364 & 0.6364 \\
unemployment\_nl & U & \textbf{0.8416} & 0.6799 & 0.7960 & 0.5663 \\
us\_population & U & \textbf{0.8889} & \textbf{0.8889} & \textbf{0.8889} & \textbf{0.8889} \\
usd\_isk & U & \textbf{0.9474} & 0.9381 & \textbf{0.9474} & 0.5542 \\
well\_log & U & \textbf{0.5878} & 0.4254 & 0.4787 & 0.3141 \\
apple & M & 0.6341 & 0.6154 & \textbf{0.6957} & 0.5938 \\
bee\_waggle\_6 & M & \textbf{0.9286} & \textbf{0.9286} & \textbf{0.9286} & \textbf{0.9286} \\
occupancy & M & 0.7895 & 0.4375 & \textbf{0.9091} & 0.3425 \\
run\_log & M & \textbf{0.7163} & 0.4483 & 0.6061 & 0.4483 \\
\end{longtable}

Table~\ref{tab:oracle-f1-dataset} presents the comparison after excluding all datasets with 50 or fewer observations, as well as \texttt{gdp\_iran} (58), \texttt{gdp\_japan} (58), \texttt{ozone} (54), and \texttt{robocalls} (52), whose monitoring periods are also relatively short. 
A dagger ($\dagger$) marks the \texttt{nile} dataset which has no annotated changepoint outside the initial $50$ observations which are used as a burn-in period. A double dagger ($\ddagger$) marks the \texttt{bank} dataset which has no
annotated changepoints. We retain these two datasets in the analysis since it is still possible for methods to falsely identify changepoints.
Table \ref{tab:oracle-f1-summary} reports average of the best \(F_1\)-scores over all
datasets. These results were obtained by running each method across a grid of parameters, shown in Table~\ref{tab:oracle-parameter-grids}, and reporting the results for the best-performing parameter configuration.

Reviewing the results, we see that OKAFF has the best $F_1$-score in $17/28$ of the datasets, compared to MMDEW ($14/28$), NEWMA ($10/28$) and Online FDD MMD $(8/28)$. In the cases where OKAFF does not have the best $F_1$-score, it has the second-best or joint second-best $F_1$-score. We note however, that in many cases the performance is the same or very similar for some datasets, although there are notable exceptions.

\begin{table}[!htbp]
\centering
\setlength{\tabcolsep}{5.7pt}
\caption{Average of the best $F_1$ scores over all datasets.
The best score is bold.}
\label{tab:oracle-f1-summary}
\begin{tabular}{@{}lrrr@{}}
\toprule
Method & Univariate & Multivariate & Overall \\
\midrule
NEWMA          & 0.7028 & 0.6074 & 0.6892 \\
MMDEW          & 0.7487 & \textbf{0.7848} & 0.7538 \\
Online RFF MMD & 0.6657 & 0.5783 & 0.6532 \\
OKAFF          & \textbf{0.7840} & 0.7671 & \textbf{0.7815} \\
\bottomrule
\end{tabular}
\end{table}

\begin{table}[htbp]
\centering
\small
\caption{Parameter grids used for per-dataset best $F_1$ selection.}
\label{tab:oracle-parameter-grids}
\begin{tabularx}{\textwidth}{lp{4.0cm}X}
\toprule
Method & Tuned parameter & Grid values \\
\midrule

MMDEW 
& Adaptation rate $\rho$,
& $\rho\in\{0.005,0.01,0.05,0.10,0.20\}$ \\
& $\theta$
& $\theta\in\{0.90,0.95,0.99,0.995\}$ \\

\addlinespace
Online RFF MMD 
& Adaptation rate $\rho$,
& $\rho\in\{0.005,0.01,0.05,0.10,0.20\}$ \\
& $\theta$
& $\theta\in\{0.90,0.95,0.99,0.995\}$ \\

\addlinespace
NEWMA 
& Window $B$,
& $B\in\{25,50,100,150,250\}$ \\
& Adaptation rate $\rho$,
& $\rho\in\{\lambda_2(B),0.005,0.01,0.05,0.10,0.20\}$  \\
& $\theta$
& $\theta\in\{0.90,0.95,0.99,0.995\}$ \\

\addlinespace
OKAFF 
& Learning rate $\eta$,
& $\eta\in\{10^{-5},10^{-4},10^{-3},10^{-2}\}$ \\
& Adaptation rate $\rho$,
& $\rho\in\{0.005,0.01,0.05,0.10,0.20\}$ \\
& $\theta=1-(1-\tilde{\theta})/2$
& $\tilde{\theta}\in\{0.90,0.95,0.99,0.995\}$ \\

\bottomrule
\end{tabularx}
\end{table}

\section{Conclusion}\label{sec-conc}

We propose a computationally efficient, online changepoint detection method named OKAFF that is suitable for both univariate and multivariate data. Section~\ref{bemostatnullalt} shows that the method's test statistic $S_t$ reacts to a variety of changes, for a range of pre-change and post-change distributions. This is due to the cost function in \eqref{costfun} that measures a difference in distribution between the current and past data.
Theorem~\ref{cor:CLTffwave} provides a way to set the threshold for the test statistic, if one can assume knowledge of the pre-change distribution. A simulation study following this approach in Section~\ref{fixthr} shows good performance for the proposed OKAFF compared to several competitor methods. However, this approach requires a relatively long burn-in period to allow the test statistic to stabilize. A second approach to setting the threshold, based on the adaptive threshold of NEWMA in \cite{keriven2020newma}, only requires a short burn-in period, and using this approach also yields good performance, as shown in Section~\ref{sec:adapthreshsimstudy}. Indeed, this second approach provides a data-driven way to setting a threshold for $S_t$, which would make OKAFF a truly nonparametric online changepoint method. Results on real-world benchmark datasets from \cite{Burg2020AnEO} show that OKAFF has slightly better overall performance compared to several competitor methods.

\section{Code availability}

The code will be soon be available as a Python package \texttt{okaff}.

\section{Disclosure statement}\label{disclosure-statement}

The authors declare that they have no competing
interests.



\section{Declaration of Generative AI Use}

\label{sec:generative-ai-use}

Generative AI used as a coding assistant, however all code was reviewed and checked by the authors for accuracy. It was also used for copyediting purposes, to check the spelling and grammar in the text.

\bibliography{bibliography.bib}

@article{fukumizu2009kernel,
  title={Kernel choice and classifiability for RKHS embeddings of probability distributions},
  author={Sriperumbudur, Bharath K and Fukumizu, Kenji and Gretton, Arthur and Lanckriet, Gert and Sch{\"o}lkopf, Bernhard},
  journal={Advances in Neural Information Processing Systems},
  volume={22},
  year={2009}
}

@book{rudin1962fourier,
  title={Fourier analysis on groups},
  author={Rudin, Walter and others},
  volume={12},
  year={1962},
  publisher={Interscience publishers New York}
}

@book{christmann2008support,
  title     = {Support Vector Machines},
  author    = {Steinwart, Ingo and Christmann, Andreas},
  series    = {Information Science and Statistics},
  year      = {2008},
  publisher = {Springer},
  address   = {New York},
  doi       = {10.1007/978-0-387-77242-4}
}

@book{reed1980methods,
  author    = {Reed, Michael and Simon, Barry},
  title     = {Methods of Modern Mathematical Physics. Vol. 1: Functional Analysis},
  publisher = {Academic Press},
  address   = {San Diego},
  year      = {1980}
}

@article{aronszajn1950theory,
  author  = {Aronszajn, N.},
  title   = {Theory of Reproducing Kernels},
  journal = {Transactions of the American Mathematical Society},
  volume  = {68},
  number  = {3},
  pages   = {337--404},
  year    = {1950}
}

@inbook{vandervaart1998stochastic,
  author    = {van der Vaart, A. W.},
  title     = {Stochastic Convergence},
  booktitle = {Asymptotic Statistics},
  series    = {Cambridge Series in Statistical and Probabilistic Mathematics},
  publisher = {Cambridge University Press},
  year      = {1998},
  pages     = {5--24}
}

@article{gretton2006kernel,
  title={A kernel method for the two-sample-problem},
  author={Gretton, Arthur and Borgwardt, Karsten and Rasch, Malte and Sch{\"o}lkopf, Bernhard and Smola, Alex},
  journal={Advances in Neural Information Processing Systems},
  volume={19},
  year={2006}
}

@article{fukumizu2008characteristic,
  title={Characteristic kernels on groups and semigroups},
  author={Fukumizu, Kenji and Gretton, Arthur and Sch{\"o}lkopf, Bernhard and Sriperumbudur, Bharath K},
  journal={Advances in Neural Information Processing Systems},
  volume={21},
  year={2008}
}

@book{shawe2004kernel,
  title     = {Kernel Methods for Pattern Analysis},
  author    = {Shawe-Taylor, John and Cristianini, Nello},
  year      = {2004},
  publisher = {Cambridge University Press},
  isbn      = {9780521813976}
}

@article{Page1954CONTINUOUSIS,
  title={CONTINUOUS INSPECTION SCHEMES},
  author={E. S. Page},
  journal={Biometrika},
  year={1954},
  volume={41},
  pages={100-115}
}

@article{Chu1996MonitoringSC,
  title={Monitoring Structural Change},
  author={Chia-Shang James Chu and Maxwell B. Stinchcombe and Halbert L. White},
  journal={Econometrica},
  year={1996},
  volume={64},
  pages={1045-1065}
}

@book{tartakovsky2014sequential,
  author    = {Tartakovsky, Alexander G. and Nikiforov, Igor V. and Basseville, Mich{\`e}le},
  title     = {Sequential Analysis: Hypothesis Testing and Changepoint Detection},
  series    = {Chapman \& Hall/CRC Monographs on Statistics \& Applied Probability},
  volume    = {136},
  publisher = {Chapman \& Hall/CRC, Taylor \& Francis Group},
  year      = {2014}
}

@article{Aue2023TheSO,
  title={The state of cumulative sum sequential changepoint testing 70 years after Page},
  author={Aue, Alexander and Kirch, Claudia},
  journal={Biometrika},
  volume={111},
  number={2},
  pages={367--391},
  year={2024},
  publisher={Oxford University Press}
}

@article{Harchaoui2007RetrospectiveMC,
  title={Retrospective Mutiple Change-Point Estimation with Kernels},
  author={Za{\"i}d Harchaoui and Olivier Capp{\'e}},
  journal={2007 IEEE/SP 14th Workshop on Statistical Signal Processing},
  year={2007},
  pages={768-772}
}

@article{arlot2019kernel,
  title={A kernel multiple change-point algorithm via model selection},
  author={Arlot, Sylvain and Celisse, Alain and Harchaoui, Zaid},
  journal={Journal of Machine Learning Research},
  volume={20},
  number={162},
  pages={1--56},
  year={2019}
}

@article{Harchaoui2008KernelCA,
  title={Kernel change-point analysis},
  author={Harchaoui, Zaid and Moulines, Eric and Bach, Francis},
  journal={Advances in Neural Information Processing Systems},
  volume={21},
  year={2008}
}

@article{ferrari2023online,
  title={Online change-point detection with kernels},
  author={Ferrari, Andr{\'e} and Richard, C{\'e}dric and Bourrier, Anthony and Bouchikhi, Ikram},
  journal={Pattern Recognition},
  volume={133},
  pages={109022},
  year={2023},
  publisher={Elsevier}
}

@article{Plasse2021Streaming,
  author  = {Plasse, J. and Hoeltgebaum, H. and Adams, N. M.},
  title   = {Streaming changepoint detection for transition matrices},
  journal = {Data Mining and Knowledge Discovery},
  year    = {2021},
  volume  = {35},
  pages   = {1287--1316}
}

@article{Burg2020AnEO,
  title={An Evaluation of Change Point Detection Algorithms},
  author={Gerrit J. J. van den Burg and Christopher K. I. Williams},
  journal={ArXiv},
  year={2020},
  volume={abs/2003.06222}
}

@article{kalman1960new,
  author  = {Kalman, Rudolf E.},
  title   = {A New Approach to Linear Filtering and Prediction Problems},
  journal = {Journal of Basic Engineering},
  volume  = {82},
  number  = {1},
  pages   = {35--45},
  year    = {1960},
  month   = mar
}

@article{bodenham2017continuous,
  author  = {Bodenham, Dean A. and Adams, Niall M.},
  title   = {Continuous Monitoring for Changepoints in Data Streams Using Adaptive Estimation},
  journal = {Statistics and Computing},
  volume  = {27},
  number  = {5},
  pages   = {1257--1270},
  year    = {2017}
}

@article{hoeltgebaum2021unsupervised,
author  = {Hoeltgebaum, Henrique and Adams, Niall and Lau, F. Din-Houn},
title   = {Unsupervised Streaming Anomaly Detection for Instrumented Infrastructure},
journal = {The Annals of Applied Statistics},
volume  = {15},
number  = {3},
pages   = {1101--1125},
year    = {2021},
month   = sep
}

@inproceedings{fukumizu2007kernel,
  author    = {Fukumizu, Kenji and Gretton, Arthur and Sun, Xiaohai and Sch{\"o}lkopf, Bernhard},
  title     = {Kernel Measures of Conditional Dependence},
  booktitle = {Advances in Neural Information Processing Systems},
  volume    = {20},
  pages     = {489--496},
  year      = {2007},
  publisher = {Curran Associates, Inc.}
}

@book{berlinet2004reproducing,
  author    = {Berlinet, Alain and Thomas-Agnan, Christine},
  title     = {Reproducing Kernel Hilbert Spaces in Probability and Statistics},
  year      = {2004},
  publisher = {Springer},
  address   = {New York, NY},
  edition   = {1},
  pages     = {355},
  isbn      = {978-1-4020-7679-4}
}

@article{Roberts1959,
  author  = {Roberts, S. W.},
  title   = {Control Chart Tests Based on Geometric Moving Averages},
  journal = {Technometrics},
  year    = {1959},
  volume  = {1},
  number  = {3},
  pages   = {239--250}
}

@article{Gretton2012AKT,
  title   = {A Kernel Two-Sample Test},
  author  = {Gretton, Arthur and Borgwardt, Karsten M. and
             Rasch, Malte J. and Sch{\"o}lkopf, Bernhard and
             Smola, Alexander},
  journal = {Journal of Machine Learning Research},
  year    = {2012},
  volume  = {13},
  number  = {25},
  pages   = {723--773}
}

@article{KeCe2019,
  author = {J{\'e}r{\'e}mie Kellner and Alain Celisse},
  title = {{A one-sample test for normality with kernel methods}},
  volume = {25},
  journal = {Bernoulli},
  number = {3},
  publisher = {Bernoulli Society for Mathematical Statistics and Probability},
  pages = {1816 -- 1837},
  year = {2019}
}

@article{KUNDU2000265,
  title   = {Central Limit Theorems revisited},
  author  = {Kundu, Subrata and Majumdar, Suman and Mukherjee, Kanchan},
  journal = {Statistics \& Probability Letters},
  volume  = {47},
  number  = {3},
  pages   = {265--275},
  year    = {2000}
}

@article{li2019scan,
  title   = {Scan {B}-statistic for kernel change-point detection},
  author  = {Li, Shuang and Xie, Yao and Dai, Hanjun and Song, Le},
  journal = {Sequential Analysis},
  volume  = {38},
  number  = {4},
  pages   = {503--544},
  year    = {2019}
}

@article{Wei2022OnlineKC,
  author  = {Wei, Song and Xie, Yao},
  title   = {Online Kernel {CUSUM} for Change-Point Detection},
  journal = {Journal of the Royal Statistical Society Series B: Statistical Methodology},
  year    = {2026},
  volume  = {88},
  number  = {4},
  pages   = {1251--1277},
  doi     = {10.1093/jrsssb/qkag020}
}

@article{kalinke2025maximum,
  author  = {Kalinke, Florian and Heyden, Marco and Gntuni, Georg
             and Fouch{\'e}, Edouard and B{\"o}hm, Klemens},
  title   = {Maximum Mean Discrepancy on Exponential Windows for Online Change Detection},
  journal = {Transactions on Machine Learning Research},
  year    = {2025},
  issn    = {2835-8856}
}

@article{anagnostopoulos2012online,
  author  = {Anagnostopoulos, Christoforos and Tasoulis, Dimitrios K. and Adams, Niall M. and Pavlidis, Nicos G. and Hand, David J.},
  title   = {Online Linear and Quadratic Discriminant Analysis with Adaptive Forgetting for Streaming Classification},
  journal = {Statistical Analysis and Data Mining},
  volume  = {5},
  number  = {2},
  pages   = {139--166},
  year    = {2012}
}

@article{kazi2026online,
  author  = {Kazi, S. H. and Adams, N. and Cohen, E. A. K.},
  title   = {Online Spectral Density Estimation},
  journal = {Journal of Computational and Graphical Statistics},
  pages   = {1--11},
  year    = {2026}
}

@article{keriven2020newma,
  title   = {{NEWMA}: A New Method for Scalable Model-Free Online Change-Point Detection},
  author  = {Keriven, Nicolas and Garreau, Damien and Poli, Iacopo},
  journal = {IEEE Transactions on Signal Processing},
  volume  = {68},
  pages   = {3515--3528},
  year    = {2020}
}

@article{kalinke2025optimal,
  title={Optimal Online Change Detection via Random Fourier Features},
  author={Kalinke, Florian and Gavioli-Akilagun, Shakeel},
  journal={Advances in Neural Information Processing Systems},
  volume={38},
  pages={9866--9901},
  year={2026}
}

@inproceedings{rahimi2007random,
  author    = {Rahimi, Ali and Recht, Benjamin},
  title     = {Random Features for Large-Scale Kernel Machines},
  booktitle = {Advances in Neural Information Processing Systems},
  volume    = {20},
  pages     = {1177--1184},
  year      = {2007},
  publisher = {Curran Associates, Inc.}
}

@inbook{McDiarmid1989,
  author    = {McDiarmid, Colin},
  title     = {On the method of bounded differences},
  booktitle = {Surveys in Combinatorics, 1989: Invited Papers at the Twelfth British Combinatorial Conference},
  editor    = {Siemons, J.},
  series    = {London Mathematical Society Lecture Note Series},
  volume    = {141},
  pages     = {148--188},
  publisher = {Cambridge University Press},
  address   = {Cambridge},
  year      = {1989},
}

\clearpage

\phantomsection\label{supplementary-material}
\bigskip




\section{Supplementary Material}\label{supmat}
The supplementary material provides the proof of Theorems~\ref{CLTffwave} and \ref{cor:CLTffwave}, Theorem~\ref{fap1}, and additional experimental results that complement the preceding sections, including detailed simulation settings, simulation study with fixed thresholds, simulation study with adaptive thresholds, comparisons of runtime, comparisons of the performance under different number of random Fourier features, the behavior of the monitoring statistic under null and alternative distributions.

\subsection{Proof of Theorems \ref{CLTffwave} and \ref{cor:CLTffwave}}\label{prfappd}

Since Theorems \ref{CLTffwave} and \ref{cor:CLTffwave} are closely related, we prove them together below.

\begin{proof}

By the reproducing property,
\[
\|\psi(X)\|_{\mathcal H(K)}^2
=
\langle K(X,\cdot),K(X,\cdot)\rangle_{\mathcal H(K)}
=
K(X,X).
\]
Hence the assumption \(\mathbb E_{X\sim p} K(X,X)<\infty\) implies
\[
\mathbb E_{X\sim p}\|\psi(X)\|_{\mathcal H(K)}^2<\infty.
\]
In particular,
\[
\mathbb E_{X\sim p}\|\psi(X)\|_{\mathcal H(K)}
\le
\left(\mathbb E_{X\sim p}\|\psi(X)\|_{\mathcal H(K)}^2\right)^{1/2}
<\infty.
\]
Moreover, \(\psi(X)\) is strongly measurable. So \(\mu_p=\mathbb E_{X\sim p}\psi(X)\) is well-defined in \(\mathcal H(K)\).
Now define
\[
\xi_i:=\psi(X_i)-\mu_p,
\qquad i\ge1.
\]
Thus, we have
\(
    \mathbb E_{X_i\sim p}\xi_i=0.
\)
By Jensen's inequality,
\[
\|\mu_p\|_{\mathcal H(K)}^2
\le
\mathbb E_{X\sim p}\|\psi(X)\|_{\mathcal H(K)}^2
<\infty.
\]
Therefore,
\begin{align}\label{esqfit}
\mathbb E_{X_i\sim p}\|\xi_i\|_{\mathcal H(K)}^2
=
\mathbb E_{X_i\sim p}\|\psi(X_i)-\mu_p\|_{\mathcal H(K)}^2 \le
2\mathbb E_{X_i\sim p}\|\psi(X_i)\|_{\mathcal H(K)}^2
+
2\|\mu_p\|_{\mathcal H(K)}^2 <\infty.
\end{align}
Let
\[
\delta_t
=
\sum_{i=1}^tb_{i,t}^2
=
\frac{\sum_{i=1}^t a_{i,t}^2}{W_t^2},
\]
and 
\(
c_{i,t}:=\frac{b_{i,t}}{\sqrt{\delta_t}}
\), $1\le i\le t$, then
\(
\sum_{i=1}^t c_{i,t}^2=1.
\)
We now apply the Hilbert-space central limit theorem. Define
the triangular array
\[
\zeta_{t,i}:=c_{i,t}\xi_i,
\qquad 1\le i\le t.
\]
Then the row sum is
\[
R_t:=\sum_{i=1}^t \zeta_{t,i}
=
\frac{\bar Y_t-\mu_p}{\sqrt{\delta_t}}.
\]
We have for every $1\leq i\leq t$,
\begin{align*}
    \mathbb E_{X_i\sim p}\|\zeta_{t,i}\|_{\mathcal H(K)}^2=c_{i,t}^2\mathbb E_{X_i\sim p}\|\xi_i\|_{\mathcal H(K)}^2\leq \mathbb E_{X_i\sim p}\|\xi_i\|_{\mathcal H(K)}^2<\infty.
\end{align*}
Let \((e_k)_{k\ge1}\) be an orthonormal basis of \(\mathcal H(K)\). Since
\(\mathbb E_{X_i\sim p}\xi_i=0\), we have
\begin{align}\label{zeroexp}
    \mathbb E_{X_i\sim p}\langle \zeta_{t,i},e_k\rangle_{\mathcal H(K)}=
c_{i,t}
\mathbb E_{X_i\sim p}\left[
\left\langle \xi_i,e_k\right\rangle_{\mathcal H(K)}
\right]=0
\end{align}
for every \(i,t,k\). 
Let \(C_t\) be the covariance operator of \(R_t\). Since $\xi_i$, $\xi_j$ are independent and centered for $i\neq j$, and $\xi_i$ ($i=1,\ldots,t$) are identically distribution, and $\sum_{i=1}^t c_{i,t}^2=1$, for
\(u,v\in\mathcal H(K)\), we have
\begin{align}
\mathbb E
\left[
\langle R_t,u\rangle_{\mathcal H(K)}
\langle R_t,v\rangle_{\mathcal H(K)}
\right]=
\sum_{i=1}^t c_{i,t}^2
\mathbb E_{X_i\sim p}
\left[
\langle \xi_i,u\rangle_{\mathcal H(K)}
\langle \xi_i,v\rangle_{\mathcal H(K)}
\right] =
\langle \Sigma_p u,v\rangle_{\mathcal H(K)},\label{cov}
\end{align}
where 
\begin{align*}
\Sigma_p
:=
\mathbb E_{X\sim p}
\left[
\{\psi(X)-\mu_p\}\otimes \{\psi(X)-\mu_p\}
\right]
\end{align*}
is the covariance operator of random element \(\psi(X)\).
Thus
\[
C_t=\Sigma_p
\]
for every \(t\). Hence, for every \(k,l\ge1\),
\begin{align}\label{bas}
\langle C_t e_k,e_l\rangle_{\mathcal H(K)}
=
\langle \Sigma_p e_k,e_l\rangle_{\mathcal H(K)}.
\end{align}
Next, for every \(t\geq1\),
\[
\sum_{k=1}^\infty
\langle C_t e_k,e_k\rangle_{\mathcal H(K)}
=
\sum_{k=1}^\infty
\langle \Sigma_p e_k,e_k\rangle_{\mathcal H(K)}
=
\operatorname{tr}(\Sigma_p).
\]
Since \(\Sigma_p\) is positive, its trace is
\[
\operatorname{tr}(\Sigma_p)
:=
\sum_{k=1}^\infty
\langle \Sigma_p e_k,e_k\rangle_{\mathcal H(K)}.
\]
For each \(k\ge1\),
\begin{align*}
\langle \Sigma_p e_k,e_k\rangle_{\mathcal H(K)}
&=
\left\langle
\mathbb E_{X\sim p}\left[(\xi\otimes \xi)e_k\right],
e_k
\right\rangle_{\mathcal H(K)} \\
&=\left\langle
\mathbb E_{X\sim p}\left[
\langle e_k,\xi\rangle_{\mathcal H(K)}\xi
\right],
e_k
\right\rangle_{\mathcal H(K)} \\
&=
\mathbb E_{X\sim p}\left[
\langle \xi,e_k\rangle_{\mathcal H(K)}^2
\right],
\end{align*}
where
\(
\xi:=\psi(X)-\mu_p
\).
Hence, by Tonelli's theorem and \eqref{esqfit}, we have
\begin{align}\label{tracefit}
\operatorname{tr}(\Sigma_p)
&=\sum_{k=1}^\infty
\mathbb E_{X\sim p}\left[
\langle \xi,e_k\rangle_{\mathcal H(K)}^2
\right]\nonumber\\
&=
\mathbb E_{X\sim p}\left[
\sum_{k=1}^\infty
\langle \xi,e_k\rangle_{\mathcal H(K)}^2
\right] \nonumber\\
&=
\mathbb E_{X\sim p}\|\xi\|_{\mathcal H(K)}^2\nonumber\\
&=
\mathbb E_{X\sim p}\|\psi(X)-\mu_p\|_{\mathcal H(K)}^2<\infty.
\end{align}
Next, for every \(k\ge1\) and \(\varepsilon>0\). We need to prove that as $t\to\infty$,
\[
\sum_{i=1}^t
\mathbb E_{X_i\sim p}
\left[
\left\langle \zeta_{t,i},e_k\right\rangle_{\mathcal H(K)}^2
\mathbf 1
\left\{
\left|
\left\langle \zeta_{t,i},e_k\right\rangle_{\mathcal H(K)}
\right|
>
\varepsilon
\right\}
\right]
\]
converge to \(0\).
Note that
\begin{align*}
&\sum_{i=1}^t
\mathbb E_{X_i\sim p}
\left[
\left\langle \zeta_{t,i},e_k\right\rangle_{\mathcal H(K)}^2
\mathbf 1
\left\{
\left|
\left\langle \zeta_{t,i},e_k\right\rangle_{\mathcal H(K)}
\right|
>
\varepsilon
\right\}
\right] 
\end{align*}
equal to
\begin{align*}
\sum_{i=1}^t c_{i,t}^2
\mathbb E_{X_i\sim p}
\left[
\left\langle \xi_i,e_k\right\rangle_{\mathcal H(K)}^2
\mathbf 1
\left\{
|c_{i,t}|
\left|
\left\langle \xi_i,e_k\right\rangle_{\mathcal H(K)}
\right|
>
\varepsilon
\right\}
\right].
\end{align*}
Since \(X_1,X_2,\ldots\) are i.i.d., the random variables
\(\left\langle \xi_i,e_k\right\rangle_{\mathcal H(K)}
\), \( i\ge1,\) are identically distributed. 
Hence
\begin{align*}
&\sum_{i=1}^t c_{i,t}^2
\mathbb E_{X_i\sim p}
\left[
Z_{i,k}^2
\mathbf 1
\left\{
|c_{i,t}|
\left|
Z_{i,k}
\right|
>
\varepsilon
\right\}
\right] \le
\mathbb E_{X_1\sim p}
\left[
Z_{1,k}^2
\mathbf 1
\left\{
\left|
Z_{1,k}
\right|
>
\frac{\varepsilon}{\max_{1\le j\le t}|c_{j,t}|}
\right\}
\right]
\end{align*}
where \(Z_{i,k}=\left\langle \xi_i,e_k\right\rangle_{\mathcal H(K)}\),
\(
\sum_{i=1}^t c_{i,t}^2=1.
\)
Furthermore,
\[
\mathbb E_{X_1\sim p}
\left[
\left\langle \xi_1,e_k\right\rangle_{\mathcal H(K)}^2
\right]
\le
\mathbb E_{X_1\sim p}\|\xi_1\|_{\mathcal H(K)}^2
<\infty,
\]
which implies $\left|
\left\langle \xi_1,e_k\right\rangle_{\mathcal H(K)}
\right|<\infty$ almost surely. 
Recall that for $1\leq i\leq t$,
\(
a_{i,t}
=
\prod_{s=i}^{t-1}\lambda_s
\).
Since \(0<\lambda_s<1\), we have
\(
0<a_{i,t}\le1.
\)
Using the inequality
\[
1-\prod_{r=1}^m(1-x_r)\le \sum_{r=1}^m x_r,
\qquad 0\le x_r\le1,
\]
we obtain
\[
1-a_{i,t}
\le
\sum_{s=i}^{t-1}(1-\lambda_s).
\]
Hence,
\[
0
\le
\frac1t\sum_{i=1}^t(1-a_{i,t})
\le
\frac1t\sum_{i=1}^t\sum_{s=i}^{t-1}(1-\lambda_s)=
\frac1t\sum_{s=1}^{t-1}s(1-\lambda_s).
\]
Since \(s(1-\lambda_s)\to0\), as $s\to\infty$, Cesaro's theorem implies 
\[
\lim_{t\to\infty}\frac1t\sum_{s=1}^{t-1}s(1-\lambda_s)=0.
\]
Therefore, 
\[
\lim_{t\to\infty}\frac{W_t}{t}=1.
\]
Similarly, since \(0<a_{i,t}\le1\),
\[
0
\le
1-\frac1t\sum_{i=1}^t a_{i,t}^2
=
\frac1t\sum_{i=1}^t(1-a_{i,t}^2)
\le
\frac2t\sum_{i=1}^t(1-a_{i,t})
\to0.
\]
Therefore,
\begin{align}\label{sumaitlim}
    \lim_{t\to\infty}\frac1t\sum_{i=1}^t a_{i,t}^2=1.
\end{align}
Thus, we obtain, as $t\to\infty$,
\(
t\delta_t\to1.
\)
Moreover, since \(a_{t,t}=1\) and \(0<a_{i,t}\le1\),
\[
\max_{1\le i\le t}c_{i,t}^2
=
\frac{\max_{1\le i\le t}b_{i,t}^2}{\delta_t}
=
\frac{1}{\sum_{i=1}^t a_{i,t}^2}.
\]
Hence, by \eqref{sumaitlim}, we have
\begin{align}\label{clim0}
    \lim_{t\to\infty}\max_{1\le i\le t}|c_{i,t}|=0.
\end{align}
Since $\left|
\left\langle \xi_1,e_k\right\rangle_{\mathcal H(K)}
\right|<\infty$ almost surely, and \eqref{clim0}, we have as $t\to\infty$
\[
\left\langle \xi_1,e_k\right\rangle_{\mathcal H(K)}^2
\mathbf 1
\left\{
\left|
\left\langle \xi_1,e_k\right\rangle_{\mathcal H(K)}
\right|
>
\frac{\varepsilon}{\max_{1\le j\le t}|c_{j,t}|}
\right\}\to0, 
\]
almost surely. Therefore, by dominated convergence theorem, as $t\to\infty$
\[
\mathbb E_{X_1\sim p}
\left[
\left\langle \xi_1,e_k\right\rangle_{\mathcal H(K)}^2
\mathbf 1
\left\{
\left|
\left\langle \xi_1,e_k\right\rangle_{\mathcal H(K)}
\right|
>
\frac{\varepsilon}{\max_{1\le j\le t}|c_{j,t}|}
\right\}
\right]
\]
converges to \(0\).
Consequently, for every \(k\ge1\) and every \(\varepsilon>0\), as $t\to\infty$,
\[
\sum_{i=1}^t
\mathbb E_{X_i\sim p}
\left[
\left\langle \zeta_{t,i},e_k\right\rangle_{\mathcal H(K)}^2
\mathbf 1
\left\{
\left|
\left\langle \zeta_{t,i},e_k\right\rangle_{\mathcal H(K)}
\right|
>
\varepsilon
\right\}
\right]
\to0.
\]
Combine this with \eqref{bas} and \eqref{tracefit}, and applying the Hilbert-space CLT (see, e.g., Theorem 1.1 in \cite{KUNDU2000265}),
\[
R_t
=
\frac{\bar Y_t-\mu_p}{\sqrt{\delta_t}}
\stackrel{d}{\rightarrow}
G
\qquad
\text{in }\mathcal H(K),
\]
where (here $\stackrel{d}{\rightarrow}$ denotes convergence in distribution)
\[
G\sim N(0,\Sigma_p)\;\;\;\;\text{in} \;\mathcal H(K).
\]
If \(dim(\mathcal{H}(K))<\infty\), it follows from the usual multivariate Lindeberg–Feller CLT (see, e.g., Proposition 2.27 in \cite{vandervaart1998stochastic}). Now write
\[
\bar Y_t
=
\mu_p+\sqrt{\delta_t}R_t.
\]
Then
\begin{align*}
S_t
&=
\|\bar Y_t\|_{\mathcal H(K)}^2=
\|\mu_p\|_{\mathcal H(K)}^2
+
2\sqrt{\delta_t}\langle \mu_p,R_t\rangle_{\mathcal H(K)}
+
\delta_t\|R_t\|_{\mathcal H(K)}^2.
\end{align*}
If \(\mu_p\neq0\), then
\[
\frac{
S_t-\|\mu_p\|_{\mathcal H(K)}^2
}{\sqrt{\delta_t}}
=
2\langle \mu_p,R_t\rangle_{\mathcal H(K)}
+
\sqrt{\delta_t}\|R_t\|_{\mathcal H(K)}^2.
\]
Since
\[
R_t\stackrel{d}{\rightarrow} G
\qquad
\text{in }\mathcal H(K),
\]
by the continuous mapping
theorem 
\[
\|R_t\|_{\mathcal H(K)}^2
\stackrel{d}{\rightarrow}
\|G\|_{\mathcal H(K)}^2.
\]
Then, as $t\to\infty$
\[
\|R_t\|_{\mathcal H(K)}^2=O_{\mathbb P}(1).
\]
Since \(\delta_t\to0\) as $t\to\infty$, we have as $t\to\infty$
\[
\sqrt{\delta_t}\|R_t\|_{\mathcal H(K)}^2\to0
\]
in probability. 
Hence, by Slutsky's theorem, as $t\to\infty$
\[
\frac{
S_t-\|\mu_p\|_{\mathcal H(K)}^2
}{\sqrt{\delta_t}}
\stackrel{d}{\rightarrow}
2\langle \mu_p,G\rangle_{\mathcal H(K)}.
\]
Since
\[
\langle \mu_p,G\rangle_{\mathcal H(K)}
\sim
N\left(
0,
\langle \Sigma_p\mu_p,\mu_p\rangle_{\mathcal H(K)}
\right),
\]
and \(t\delta_t\to1\),  by Slutsky’s theorem
we get as $t\to\infty$
\[
\sqrt t
\left(
S_t-\|\mu_p\|_{\mathcal H(K)}^2
\right)
\stackrel{d}{\rightarrow}
N\left(
0,
4\langle \Sigma_p\mu_p,\mu_p\rangle_{\mathcal H(K)}
\right).
\]
Next, we aim to prove that
if \(\mu_p=0\), then as $t\to\infty$
\[
tS_t
\stackrel{d}{\rightarrow}
\sum_{r\ge1}\eta_{r,p}Z_r^2.
\]
Note that if \(\mu_p=0\), then
\[
S_t=\delta_t\|R_t\|_{\mathcal H(K)}^2.
\]
Therefore, by the continuous mapping theorem
\[
\frac{S_t}{\delta_t}
=
\|R_t\|_{\mathcal H(K)}^2
\stackrel{d}{\rightarrow}
\|G\|_{\mathcal H(K)}^2.
\]
Since \(t\delta_t\to1\), this implies
\[
tS_t\stackrel{d}{\rightarrow} \|G\|_{\mathcal H(K)}^2.
\]
Since \(G\sim N(0,\Sigma_p)\) in \(\mathcal H(K)\), the covariance operator \(\Sigma_p\) is positive, self-adjoint and trace-class. Hence its nonzero
eigenvalues \((\eta_{r,p})_{r\ge1}\), counted with multiplicity, satisfy
\begin{align}\label{trsp}
\sum_{r\ge1}\eta_{r,p}=\operatorname{tr}(\Sigma_p)<\infty.
\end{align}
By the Karhunen--Lo\`eve expansion,
\begin{align}\label{G2}
    \|G\|_{\mathcal H(K)}^2
\overset{d}{=}
\sum_{r\ge1}\eta_{r,p}Z_r^2,
\end{align}
where \(Z_1,Z_2,\ldots\) are i.i.d. \(N(0,1)\).
Therefore, as $t\to\infty$
\[
tS_t
\stackrel{d}{\rightarrow}
\sum_{r\ge1}\eta_{r,p}Z_r^2.
\]
Finally, note that
\[
\mathbb E S_t
=
\|\mu_p\|_{\mathcal H(K)}^2
+
\delta_t\operatorname{tr}(\Sigma_p)=
\eta\delta_t
+
\bigl(1-\delta_t\bigr)\theta_K.
\]
Thus, in the nonzero mean case ($\mu_p\neq0$),
\[
\sqrt t
\left(
\mathbb E S_t-\|\mu_p\|_{\mathcal H(K)}^2
\right)
=
\sqrt t\,\delta_t\operatorname{tr}(\Sigma_p)
\to0,
\]
thus, by Slutsky’s theorem, we have
as $t\to\infty$
\[
\sqrt t
\left(
S_t-\mathbb E S_t
\right)
\stackrel{d}{\rightarrow}
N\left(
0,
4\langle \Sigma_p\mu_p,\mu_p\rangle_{\mathcal H(K)}
\right).
\]
In the zero mean case ($\mu_p=0$),
\[
\mathbb E S_t=\delta_t\operatorname{tr}(\Sigma_p),
\]
and hence
\[
t(S_t-\mathbb E S_t)
=
t\delta_t
\left(
\frac{S_t}{\delta_t}-\operatorname{tr}(\Sigma_p)
\right).
\]
Since \(t\delta_t\to1\), by Slutsky’s theorem, we have
\[
t(S_t-\mathbb E S_t)
\stackrel{d}{\rightarrow}
\|G\|_{\mathcal H(K)}^2-\operatorname{tr}(\Sigma_p).
\]
By \eqref{trsp} -- \eqref{G2}
we have
\[
t(S_t-\mathbb E S_t)
\stackrel{d}{\rightarrow}
\sum_{r\ge1}\eta_{r,p}(Z_r^2-1).
\]
\end{proof}

\subsection{Proof of Theorem \ref{fap1}}
\begin{proof}
Recall that $Y_i=\psi(X_i)$, and
\[
S_t
=
\left\|\sum_{i=1}^t b_{i,t}Y_i\right\|_{\mathcal H(K)}^2
=
\sum_{i=1}^t b_{i,t}^2K(X_i,X_i)
+
2\sum_{1\le i<j\le t}b_{i,t} b_{j,t}K(X_i,X_j).
\]
Note that \(X_i\) are i.i.d.,
\[
\mathbb E S_t
=
\eta_K\sum_{i=1}^t b_{i,t}^2
+
2\theta_K\sum_{1\le i<j\le t}b_{i,t} b_{j,t}=
\eta_K\delta_t+(1-\delta_t)\theta_K,
\]
where
\(
\eta_K=\mathbb E_{X_i\sim p} K(X_i,X_i)
\),
and
\(
\theta_K=\mathbb E_{X_i,X_j\sim p} K(X_i,X_j).
\)
Since \(K\) is positive semidefinite,
\(
K(x,x)\ge 0
\)
for every \(x\). Moreover, by the Cauchy--Schwarz inequality,
\begin{align}\label{Kbound}
    |K(x,z)|
=
|\langle K(x,\cdot),K(z,\cdot)\rangle_{\mathcal H(K)}|
\le
\sqrt{K(x,x)K(z,z)}
\le
\kappa.
\end{align}
Replace \(X_i\) by \(X_i'\), the diagonal term changes by at most
\[
b_{i,t}^2|K(X_i,X_i)-K(X_i',X_i')|
\le
\kappa b_{i,t}^2.
\]
The off-diagonal terms involving \(X_i\) change by at most
\[
2b_{i,t}\sum_{j\ne i}b_{j,t}
\left|
K(X_i,X_j)-K(X_i',X_j)
\right|.
\]
By \eqref{Kbound}, we have
\(
\left|
K(X_i,X_j)-K(X_i',X_j)
\right|
\le
2\kappa.
\)
Therefore,
\[
\left|
S_t(X)-S_t(X^{(i)})
\right|
\le
\kappa b_{i,t}^2
+
4\kappa b_{i,t}\sum_{j\ne i}b_{j,t}=
\kappa b_{i,t}(4-3b_{i,t}),
\]
where \(X:=(X_1,\ldots,X_t)\), \(X^{(i)}\) is the vector obtained by replacing the i-th component \(X_i\) of \(X\) by \(X_i'\), and
\(
S_t(X):=S_t(X_1,\ldots,X_t):=\sum_{k,\ell=1}^t b_{k,t} b_{\ell,t} K(X_k,X_\ell),
\)
and
\(
S_t(X^{(i)}):=S_t(X_1,\ldots,X_{i-1},X_i',X_{i+1},\ldots,X_t).
\)
Therefore,
\begin{align}\label{bound2}
\sup_{X\in\mathcal{X}^t,X_i'\in \mathcal{X}}\left|
S_t(X)-S_t(X^{(i)})
\right|
\le
\kappa b_{i,t}(4-3b_{i,t}).
\end{align}
By \eqref{bound2} and McDiarmid's inequality \citep{McDiarmid1989}, we have
\begin{align}\label{bound1}
    \mathbb P\left(
|S_t-\mathbb E S_t|>\varepsilon
\right)
\le
2\exp\left(
-\frac{2\varepsilon^2}{\kappa^2L_t}
\right),
\end{align}
where 
\(
L_t
=
\sum_{i=1}^t b_{i,t}^2(4-3b_{i,t})^2.
\)
Let \(j(t):=\lfloor \log_2 t\rfloor\). For every
\(\alpha\in(0,1)\), define
\[
\varepsilon_{t}
:=
\kappa
\sqrt{
\frac{L_t}{2}
\log\left(
\frac{2\pi^2\,2^{j(t)}(j(t)+1)^2}{3\alpha}
\right)
}.
\]
That is,
\begin{align*}
    \varepsilon_{t}
&=\kappa
\sqrt{
\frac{L_t}{2}
\left[
\lfloor \log_2 t\rfloor\log 2
+
2\log(\lfloor \log_2 t\rfloor+1)
+
\log\left(\frac{2\pi^2}{3\alpha}\right)
\right]
}.
\end{align*}
Then, for \(2^j\le t<2^{j+1}\), we have \(j(t)=j\). Therefore, by \eqref{bound1}, we have
\[
\begin{aligned}
\mathbb P\left(
|S_t-\mathbb E S_t|>\varepsilon_{t}
\right)
&\le
2\exp\left(
-\frac{2\varepsilon_{t}^2}
{\kappa^2L_t}
\right)=
2\exp\left[
-\log\left(
\frac{2\pi^2\,2^j(j+1)^2}{3\alpha}
\right)
\right] =
\frac{3\alpha}{\pi^2\,2^j(j+1)^2}.
\end{aligned}
\]
Therefore,
\[
\begin{aligned}
\mathbb P_{\infty}\left(T<\infty\right)
&\le
\mathbb P\left(
\exists t<\infty \text{ such that }
|S_t-\mathbb E S_t|>\varepsilon_{t}
\right) \\
&=
\mathbb P\left(\cup_{t=1}^{\infty}\{
|S_t-\mathbb E S_t|>\varepsilon_{t}\}
\right) \\
&=
\mathbb P\left(\cup_{j=0}^{\infty}\cup_{2^j\leq t< 2^{j+1}}\{
|S_t-\mathbb E S_t|>\varepsilon_{t}\}
\right) \\
&\le
\sum_{j=0}^\infty
2^j
\max_{2^j\le t<2^{j+1}}
\mathbb P\left(
|S_t-\mathbb E S_t|>\varepsilon_{t}
\right) \\
&\le
\sum_{j=0}^\infty
2^j
\frac{3\alpha}{\pi^2\,2^j(j+1)^2} \\
&=
\frac{3\alpha}{\pi^2}
\sum_{j=0}^\infty \frac{1}{(j+1)^2} \\
&=
\frac{3\alpha}{\pi^2}\cdot \frac{\pi^2}{6}
=
\frac{\alpha}{2}
<
\alpha.
\end{aligned}
\]
Note that for all $t\geq 1$,
\[
L_t
=
\sum_{i=1}^t
b_{i,t}^2(4-3b_{i,t})^2
\leq
16\sum_{i=1}^t b_{i,t}^2
\leq
\frac{16}{W_t}.
\]
Since $0<\lambda_t<1$, and $1-\lambda_t=o(t^{-\beta})$ for some $\beta>1$, for all sufficiently large $t$, and for  $\left\lfloor \frac{t}{2}\right\rfloor\leq i\leq t$, we have
\[
a_{i,t}
\geq
\prod_{s=i}^{t-1}(1-s^{-\beta})>0.
\]
Therefor, for all sufficiently large $t$, and  $\left\lfloor \frac{t}{2}\right\rfloor\leq i\leq t$, we have
\begin{align*}
    \log a_{i,t}
&\geq
\sum_{s=i}^{t-1}\log(1-s^{-\beta})
\geq
-2\sum_{s=i}^{t-1}s^{-\beta},\\
&\geq-2t\left(\frac{t}{4}\right)^{-\beta}
\geq
-4^{\beta+1} t^{1-\beta}\geq -4^{\beta+1} (\beta\geq1)
\end{align*}
where the second inequality is derived from
\(
\log(1-x)\geq -2x\) for
\(0\leq x\leq \frac12.
\)
Therefore, 
\(
a_{i,t}\geq e^{-4^{\beta+1}}.
\)
Consequently, for all sufficiently large $t$, we have
\[
W_t
=
\sum_{i=1}^t a_{i,t}
\geq
\sum_{i=\left\lfloor \frac{t}{2}\right\rfloor}^{t} a_{i,t}
\geq
(t-\left\lfloor \frac{t}{2}\right\rfloor)e^{-4^{\beta+1}}
\geq
c t
\]
for some constant $c>0$. Therefore, \(\varepsilon_t=O(\sqrt{t^{-1}\log t}).\)

\end{proof}

\subsection{Experiments}\label{exprappd}

This section contains additional experiments which are referenced in the main text. Specifically, Section~\ref{sec:supp:addfixedexp} provides additional simulation study results for the fixed threshold approach. Section~\ref{adpthr} provides the mathematical details for the adaptive threshold approach, along with additional simulation study results. Section~\ref{supp:sec:avruntime} provides the results comparing the runtime of the proposed OKAFF method to competitor methods.
Section~\ref{rffnumcom} shows the performance of OKAFF for different numbers of random Fourier features.
Section~\ref{sec:supp:monstat} provides additional plots showing the behaviour of the OKAFF statistic $S_t$ under different conditions.

\subsubsection{Additional simulation with fixed thresholds}

\label{sec:supp:addfixedexp}

\paragraph*{Simulation settings.} This subsection provides the detailed simulation settings for the \(20\)-dimensional experiments
reported in Section \ref{fixthr}, while the univariate setting has already been described in Section \ref{fixthr}. It specifies the pre-change distribution and the post-change alternatives used to evaluate detection performance in Section \ref{fixthr}.

Let \((X_t)_{t\geq1}\) be a sequence of independent \(\mathbb{R}^d\)-valued observations with common distribution
\[
p=\mathcal{N}(\boldsymbol{0},I).
\]
After the change, \(X_t\sim q\neq p\). For \(d=20\), we consider the
following alternatives:

\begin{itemize}
    \item \textbf{Dense mean shift:}
    \[
    \begin{aligned}
    q=\mathcal{N}(\delta\mathbf{1},I),\quad
    \delta\in\{1,2,3,4\}.
    \end{aligned}
    \]
    Here, \(\mathbf{1}=(1,\ldots,1)^\top\in\mathbb{R}^d\). 

    \item \textbf{Sparse mean shift:}
    \[
    q=\mathcal{N}\!\left(
    \left(
    \underbrace{\delta,\ldots,\delta}_{5},
    \underbrace{0,\ldots,0}_{d-5}
    \right)^\top,I
    \right),\qquad
    \delta\in\{1,2,3,4\}.
    \]

    \item \textbf{Dense covariance shift:}
    \[
    q=\mathcal{N}(\boldsymbol{0},\delta I),
    \qquad \delta\in\{0.3,0.5,2,4\}.
    \]

    \item \textbf{Sparse covariance shift:}
    \[
    q=\mathcal{N}\!\left(
    \boldsymbol{0},
    \operatorname{diag}\!\left(
    \underbrace{\delta,\ldots,\delta}_{5},
    \underbrace{1,\ldots,1}_{d-5}
    \right)
    \right),\qquad \delta\in\{0.3,0.5,2,4\}.
    \]

    \item \textbf{Gaussian mixture:}
    \[
    q=(1-\delta)\mathcal{N}(\boldsymbol{0},4I)
       +\delta\,\mathcal{N}(\boldsymbol{0},I),\qquad
    \delta\in\{0.3,0.5,0.7,0.9\}.
    \]

     \item \textbf{Uniform distribution:}
    \[
    q=\operatorname{Uniform}([-\delta,\delta]^d),
    \qquad \delta\in\{1,2,3,4\},
    \]
    is the \(d\)-dimensional Uniform distribution with mutually independent coordinates, each following an \(\operatorname{Uniform}([-\delta,\delta])\) distribution.
\end{itemize}

\begin{itemize}
    \item \textbf{Laplace distribution:}
    \[
    q=\operatorname{Laplace}(0,\delta)^{d},
    \qquad \delta\in\{0.3,0.5,1,2\},
    \]
    is the \(d\)-dimensional Laplace distribution with mutually independent coordinates, each following an \(\operatorname{Laplace}(0,\delta)\) distribution.
\end{itemize}

\begin{figure}[!htbp]
\centering

\begin{subfigure}[t]{0.96\linewidth}
  \centering
  \includegraphics[
    width=\linewidth
  ]{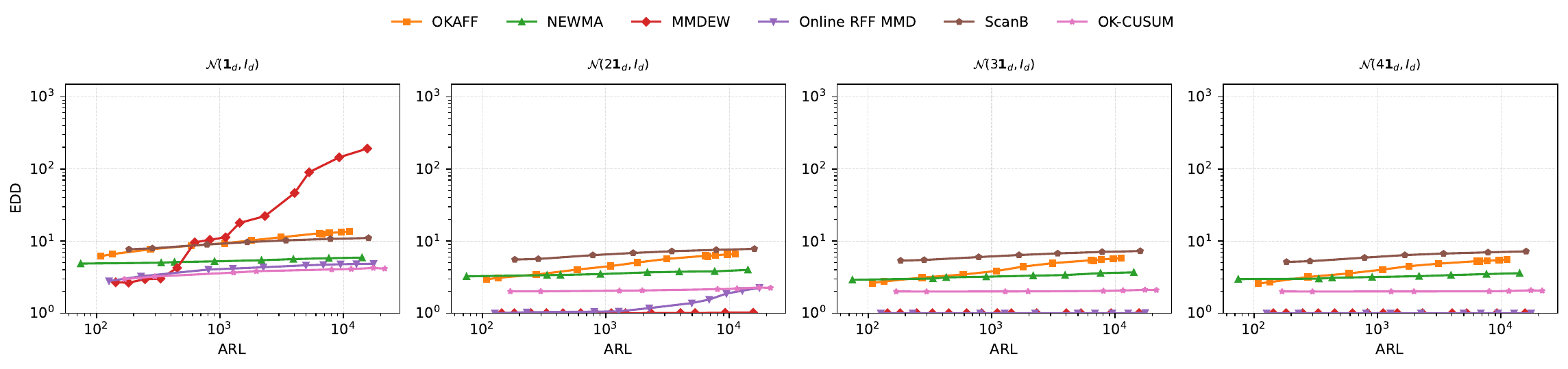}
  \caption{Dense mean shift in \(d=20\):
  \(q=\mathcal{N}(\delta\mathbf{1},I)\), where
  \(\delta\in\{1,2,3,4\}\).}
\end{subfigure}

\medskip

\begin{subfigure}[t]{0.96\linewidth}
  \centering
  \includegraphics[
    width=\linewidth
  ]{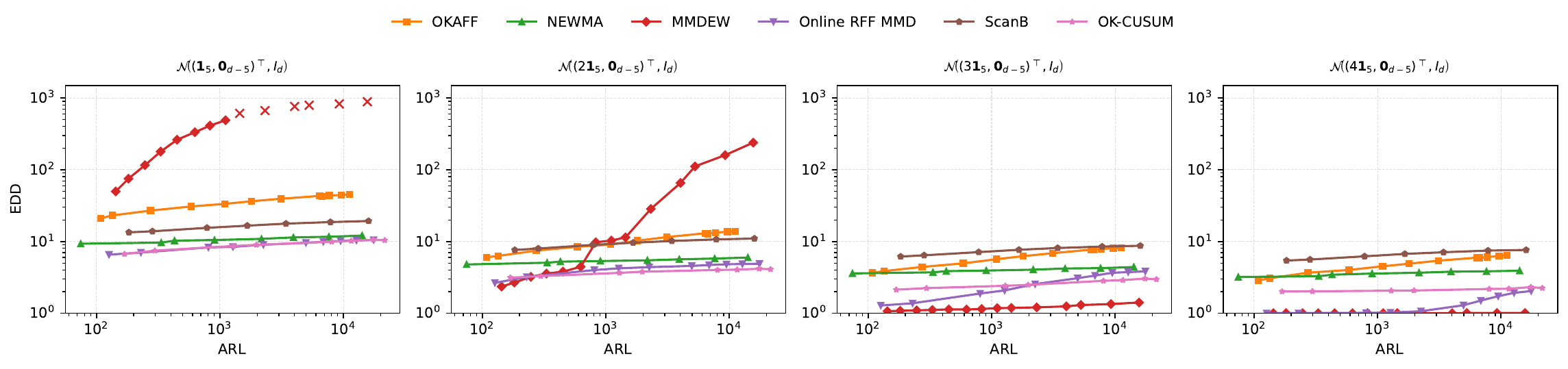}
  \caption{Sparse mean shift in \(d=20\): the first five coordinates
  have mean \(\delta\in\{1,2,3,4\}\), while the remaining coordinates
  have mean zero.}
\end{subfigure}

\caption{EDD versus ARL for the dense and sparse mean-shift settings in
\(d=20\), when the change occurs at \(k=100\).}
\label{fig:edd-arl-mean-shifts-d20}

\end{figure}

\begin{figure}[!htbp]
\centering

\begin{subfigure}[t]{0.96\linewidth}
  \centering
  \includegraphics[
    width=\linewidth
  ]{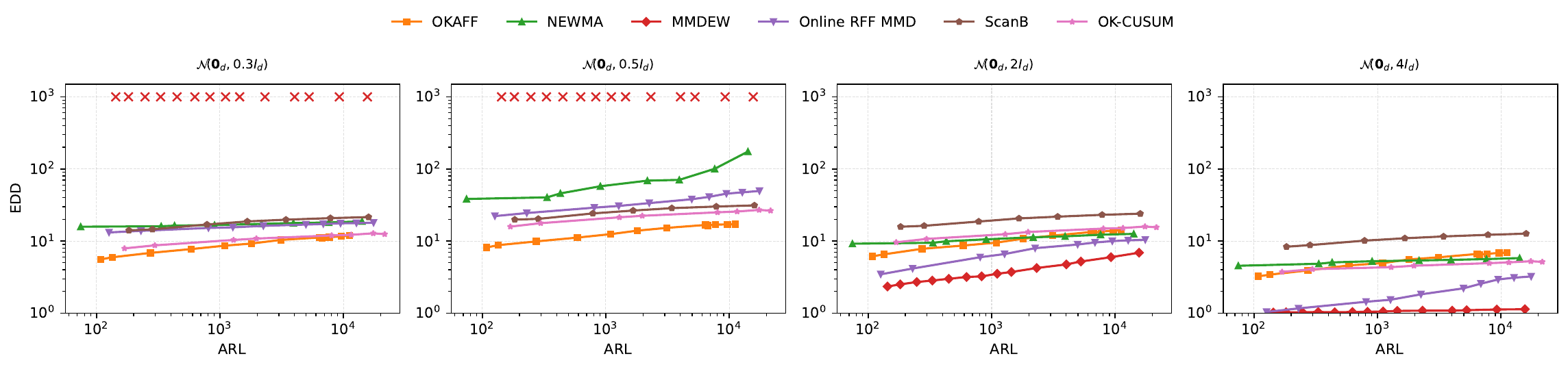}
  \caption{Dense covariance shift in \(d=20\):
  \(q=\mathcal{N}(\boldsymbol{0},\delta I)\), where
  \(\delta\in\{0.3,0.5,2,4\}\).}
\end{subfigure}

\medskip

\begin{subfigure}[t]{0.96\linewidth}
  \centering
  \includegraphics[
    width=\linewidth
  ]{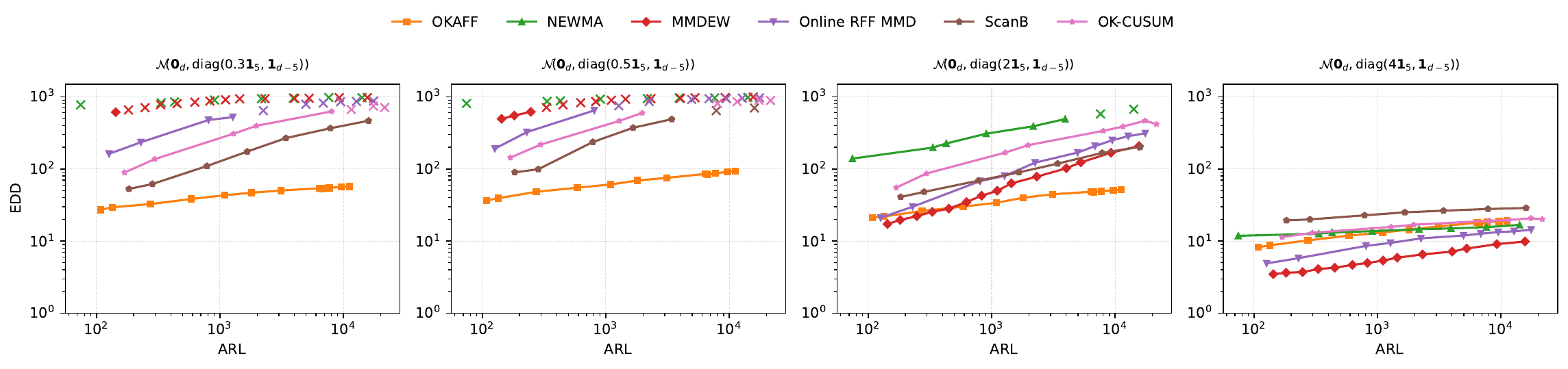}
  \caption{Sparse covariance shift in \(d=20\): the first five marginal
  variances are \(\delta\in\{0.3,0.5,2,4\}\), while the remaining
  marginal variances equal one.}
\end{subfigure}

\caption{EDD versus ARL for the dense and sparse covariance-shift
settings in \(d=20\), when the change occurs at \(k=100\).}
\label{fig:edd-arl-covariance-shifts-d20}

\end{figure}

\begin{figure}[!htbp]
\centering

\begin{subfigure}[t]{0.96\linewidth}
  \centering
  \includegraphics[
    width=\linewidth
  ]{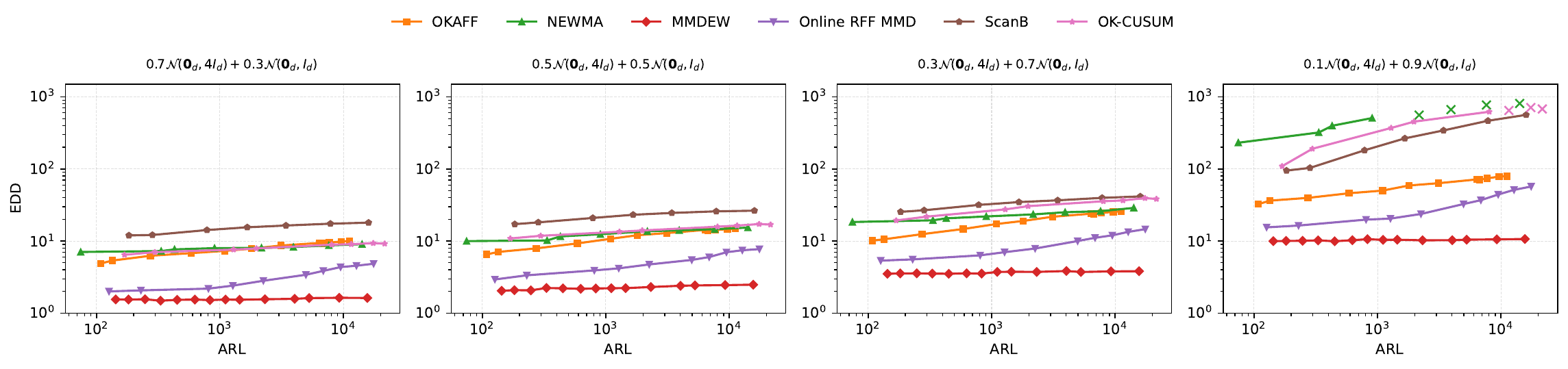}
  \caption{Gaussian-mixture setting in \(d=20\).}
\end{subfigure}

\medskip

\begin{subfigure}[t]{0.96\linewidth}
  \centering
  \includegraphics[
    width=\linewidth
  ]{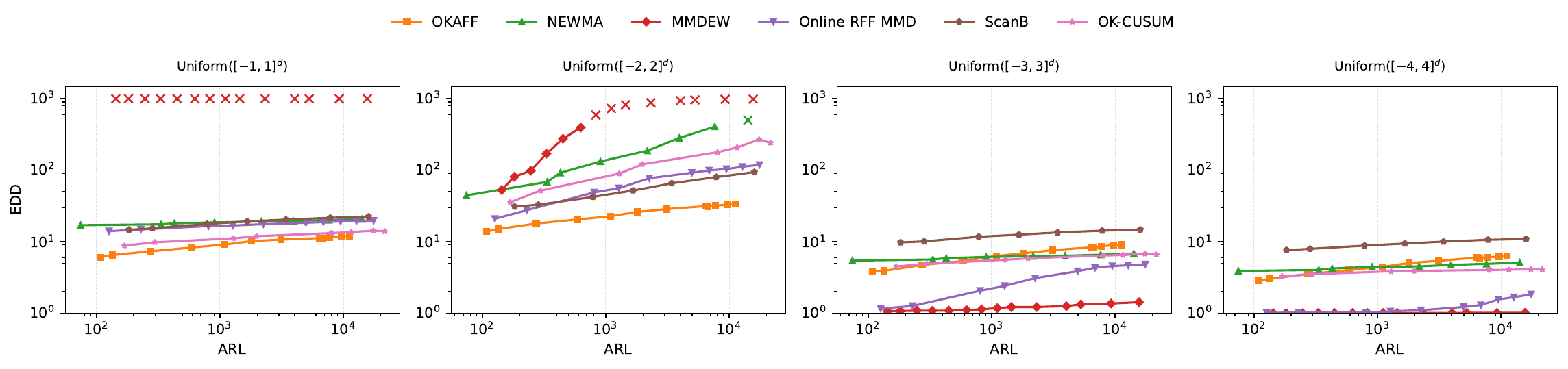}
  \caption{Uniform setting in \(d=20\).}
\end{subfigure}

\medskip

\begin{subfigure}[t]{0.96\linewidth}
  \centering
  \includegraphics[
    width=\linewidth
  ]{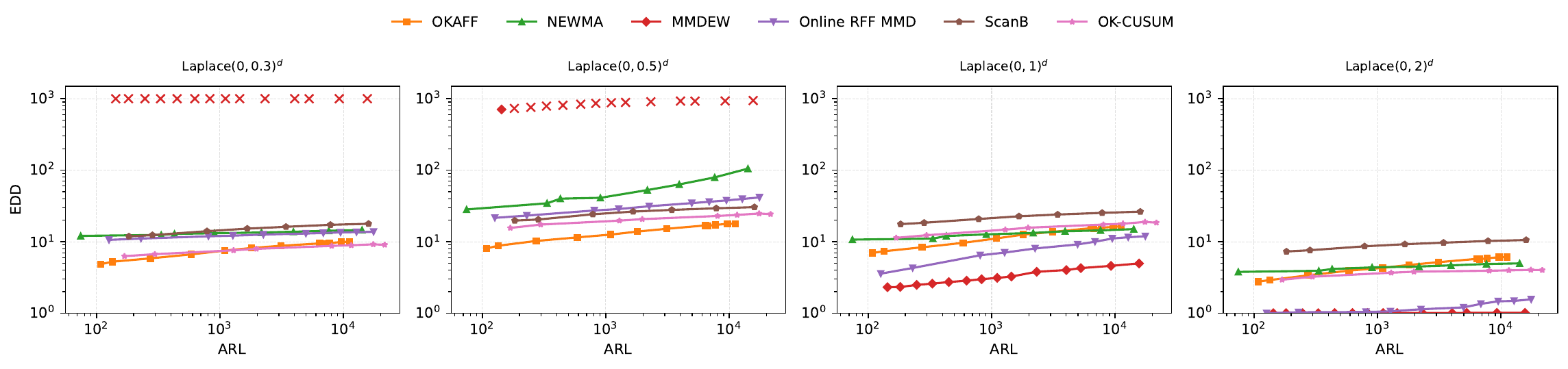}
  \caption{Laplace setting in \(d=20\).}
\end{subfigure}

\caption{EDD versus ARL for the Gaussian-mixture, uniform, and Laplace
settings in \(d=20\), when the change occurs at \(k=100\).}
\label{fig:edd-arl-distribution-shifts-d20}

\end{figure}

\begin{figure}[!htbp]
\centering

\begin{subfigure}[t]{0.96\linewidth}
  \centering
  \includegraphics[
    width=\linewidth
  ]{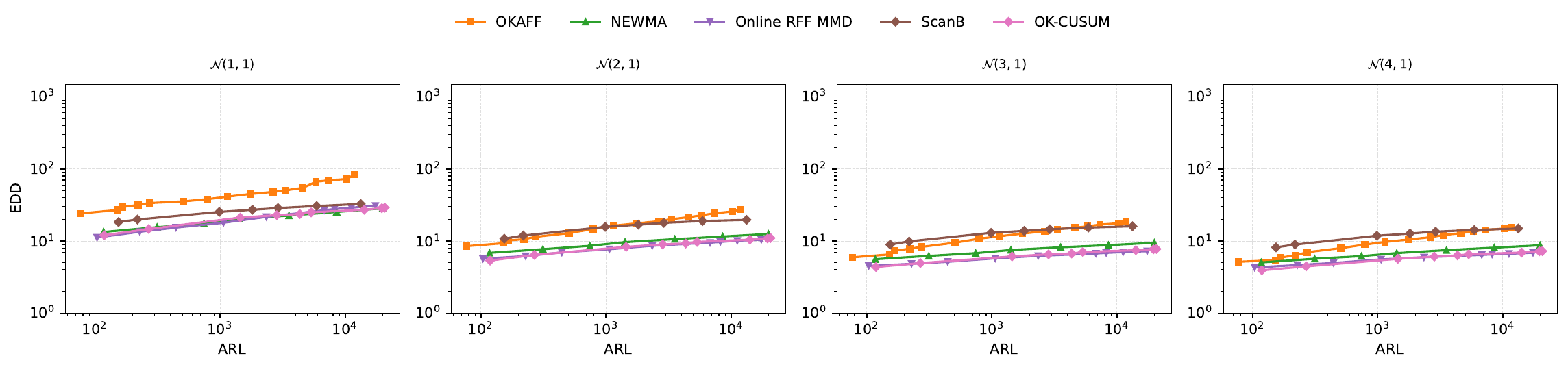}
  \caption{Univariate mean shift:
  \(q=\mathcal{N}(\delta,1)\), where
  \(\delta\in\{1,2,3,4\}\).}
\end{subfigure}

\caption{EDD versus ARL for the univariate mean-shift setting when the
change occurs at \(k=100\).}
\label{fig:edd-arl-mean-shifts-d1}

\end{figure}

\begin{figure}[!htbp]
\centering

\begin{subfigure}[t]{0.96\linewidth}
  \centering
  \includegraphics[
    width=\linewidth
  ]{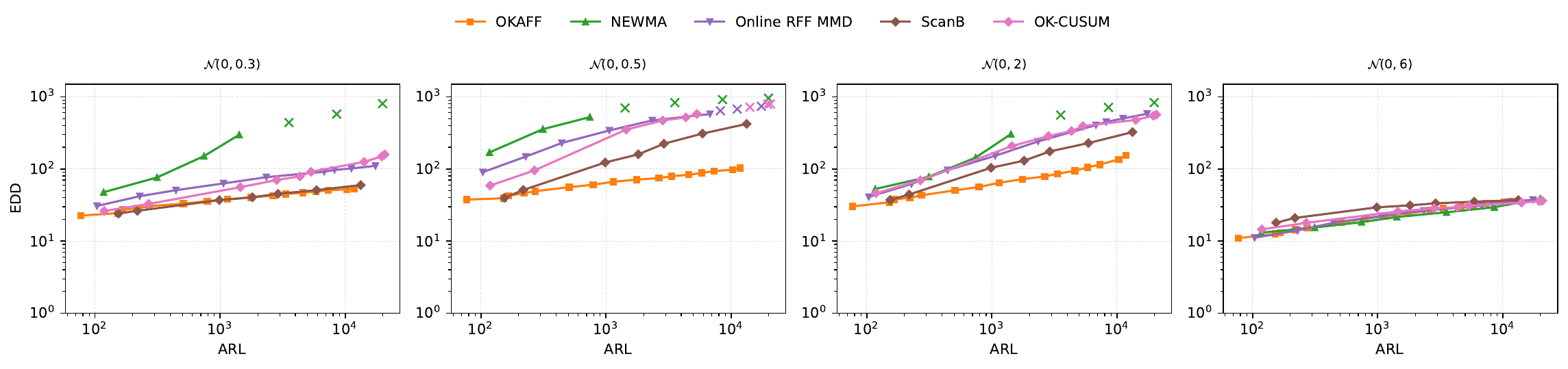}
  \caption{Univariate variance shift:
  \(q=\mathcal{N}(0,\delta)\), where
  \(\delta\in\{0.3,0.5,2,6\}\).}
\end{subfigure}

\caption{EDD versus ARL for the univariate variance-shift setting when
the change occurs at \(k=100\).}
\label{fig:edd-arl-covariance-shifts-d1}

\end{figure}

\begin{figure}[!htbp]
\centering

\begin{subfigure}[t]{0.96\linewidth}
  \centering
  \includegraphics[
    width=\linewidth
  ]{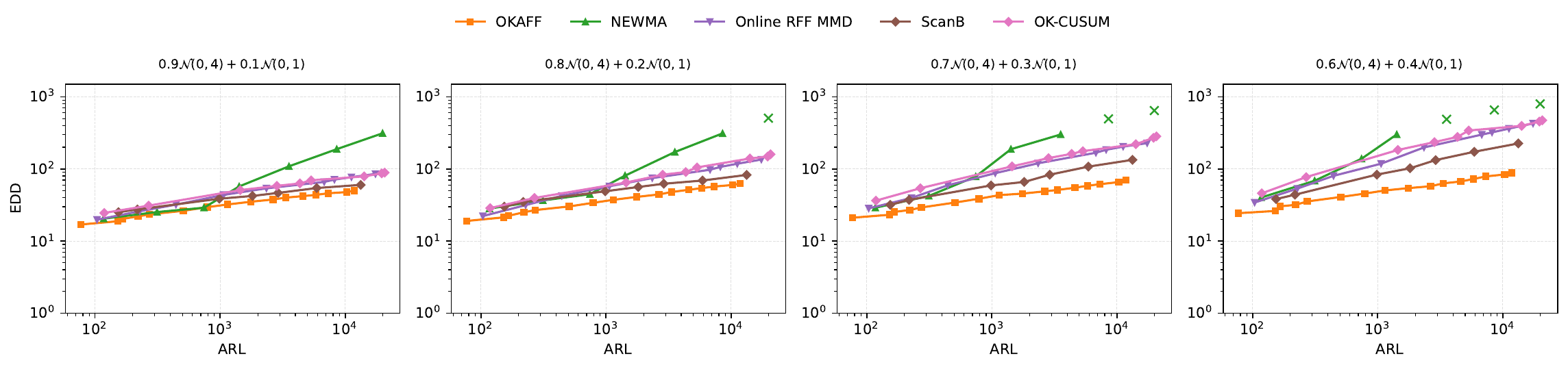}
  \caption{Gaussian-mixture setting for univariate data.}
\end{subfigure}

\medskip

\begin{subfigure}[t]{0.96\linewidth}
  \centering
  \includegraphics[
    width=\linewidth
  ]{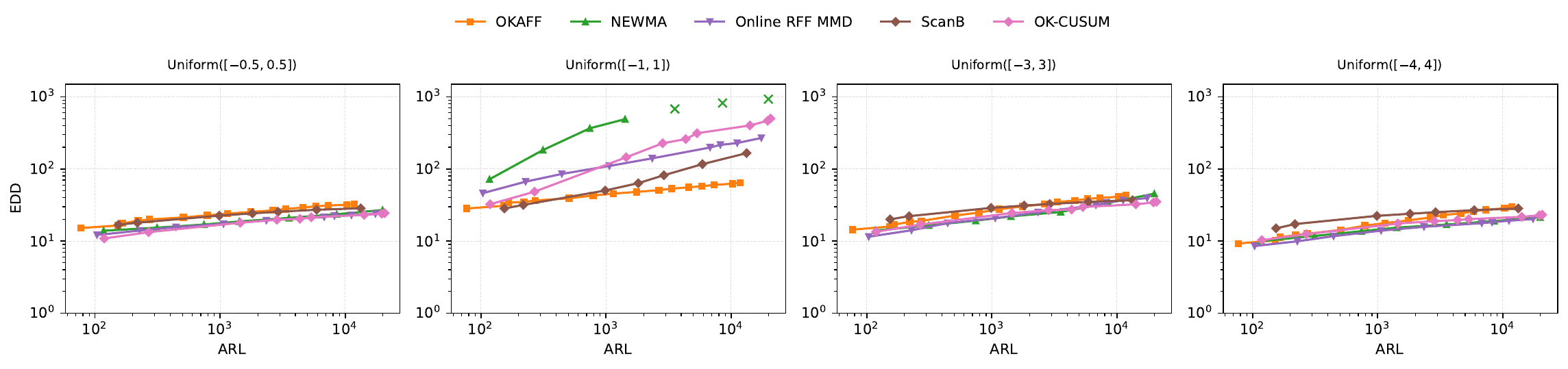}
  \caption{Uniform setting for univariate data.}
\end{subfigure}

\medskip

\begin{subfigure}[t]{0.96\linewidth}
  \centering
  \includegraphics[
    width=\linewidth
  ]{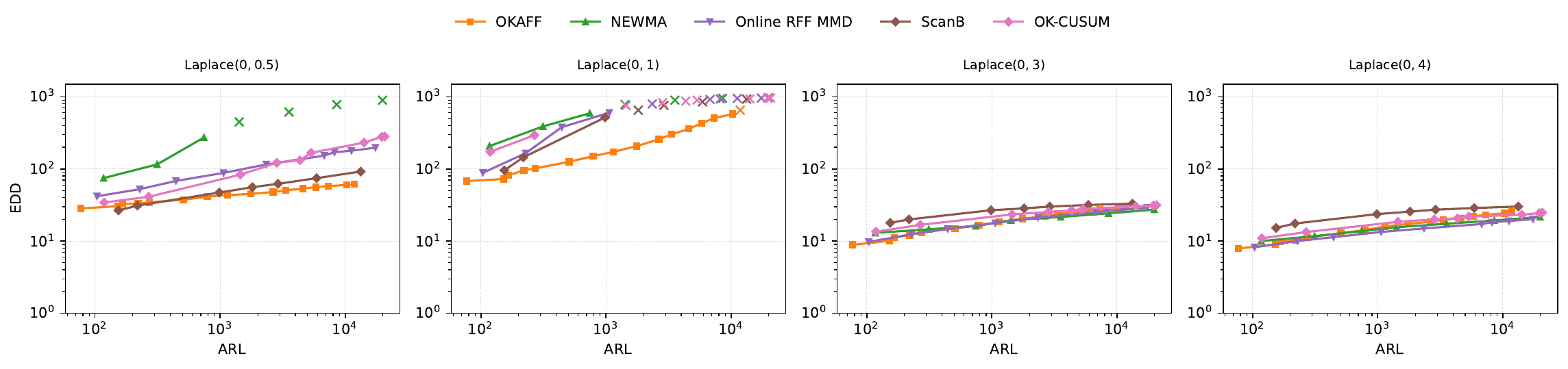}
  \caption{Laplace setting for univariate data.}
\end{subfigure}

\caption{EDD versus ARL for the univariate Gaussian-mixture, uniform,
and Laplace settings when the change occurs at \(k=100\).}
\label{fig:edd-arl-distribution-shifts-d1}

\end{figure}

Figures~\ref{fig:edd-arl-mean-shifts-d20}--\ref{fig:edd-arl-distribution-shifts-d1}
present the EDD-versus-ARL results for the \(d=20\) and univariate experiments, complementing those reported in Section~\ref{fixthr} for simulations with fixed thresholds. The results are grouped according to mean, covariance, and general distributional changes. In all experiments, the change occurs at \(k=100\). To specify the dimension, we use subscripts for the normal distributions in these figures. For example, we write $\mathcal{N}(\boldsymbol{0}_d, I_d)$ instead of $\mathcal{N}(\boldsymbol{0}, I)$, and $\mathcal{N}((\delta\boldsymbol{1}_{5}, \boldsymbol{0}_{d-5})^{\top},I_d)$ instead of $\mathcal{N}((\underbrace{\delta,\ldots,\delta}_{5}, \underbrace{0,\ldots,0}_{d-5})^{\top}, I)$.

\subsubsection{Simulations with adaptive thresholds} 

\label{adpthr}

In this section, we describe the adaptive threshold construction proposed by \cite{keriven2020newma}, which allows the threshold to evolve continuously during online monitoring.  We conduct simulations in which the adaptive thresholds are calibrated using pre-change observations to achieve the same prescribed in-control ARL across all methods. As in the fixed-threshold experiments, this common ARL calibration ensures a fair and meaningful comparison of detection performance.

Let \(S_t\) denote the monitoring statistic. OKAFF uses a two-sided threshold based on exponentially weighted first and
second moments:
\begin{align*}
u_{1,t}=(1-\rho)u_{1,t-1}+\rho S_t,\qquad
u_{2,t}=(1-\rho)u_{2,t-1}+\rho S_t^2,\qquad
\sigma_t=\sqrt{u_{2,t}-u_{1,t}^2},
\end{align*}
with the initial value $u_{1,0}=u$, $u_{2,0}=\sigma^2+u^2$, here $u$, $\sigma$ are determined by \eqref{u}-\eqref{zetagauss}.
It signals a change whenever
\[
S_t \not \in 
\bigl[u_{1,t}-L\sigma_t,\,
      u_{1,t}+L\sigma_t\bigr] = 
      \mathcal{R}_{L, \rho, t}', 
\]
where the adaptive rejection region is defined as 
\begin{equation}
    \mathcal{R}_{L, \rho, t} 
    :=
[0,u_{1,t}-L\sigma_t) \cup (u_{1,t}+L \sigma_t,\infty).
\label{eqn:RLrhot}
\end{equation}
NEWMA, MMDEW, and Online RFF MMD use a one-sided threshold
based on exponentially weighted moments of \(S_t^2\):
\begin{align*}
\mu_{2,t}=(1-\rho)\mu_{2,t-1}+\rho S_t^2,\qquad
\mu_{4,t}=(1-\rho)\mu_{4,t-1}+\rho S_t^4,\qquad
\sigma_t=
\sqrt{\mu_{4,t}-\mu_{2,t}^{2}},
\end{align*}
with the initial value $u_{2,0}=u_{4,0}=0$. They signal a change whenever
\[
S_t^2>\mu_{2,t}+L\sigma_t.
\]

For all methods, we set \(\rho=0.1\) during the burn-in period. This enables the threshold to be learnt quickly during the burn-in period.
After the burn-in period ends, OKAFF, Online RFF MMD, and MMDEW use \(\rho=0.01\) when monitoring for a change. For NEWMA, \(\rho\) is tied to the slow update rate \(\lambda\), as specified by its authors, during the monitoring period.
Select the control limit \(L=\Phi^{-1}(\theta)\), where \(\Phi^{-1}(\theta)\) denotes the \(\theta\)-quantile of the standard normal distribution, to achieve a prescribed in-control ARL.

Monitoring begins after
50 burn-in observations for OKAFF, Online RFF MMD, MMDEW, and NEWMA. For each method, the first 50 burn-in
observations constitute the reference sample used to select the
Gaussian-kernel bandwidth by the median heuristic. All remaining parameters
are identical to those used in the fixed-threshold experiments. ScanB and OK-CUSUM require a relatively large number of pre-change observations to construct their reference blocks. We therefore exclude these methods from the adaptive-threshold and real-data experiments, where only limited pre-change data are assumed to be available. While such a requirement may be reasonable in applications such as manufacturing process control and industrial quality monitoring, where extensive in-control data can often be collected, it may be unrealistic in other settings, such as econometrics and environmental monitoring, where a large pre-change reference sample is rarely available.

We consider both univariate and 20-dimensional settings. For each method, \(L\) is calibrated using 500 Monte Carlo replications so that the estimated in-control ARL satisfies the prescribed constraint. In both settings, the change occurs at \(k=100\) and is followed by at most 1000 sequential post-change observations. The EDD is estimated from 500 Monte Carlo replications as the average detection delay among successful detections. 

Recall that the average run length  under $H_0$ is defined by 
\[
\operatorname{ARL}:=\mathbb{E}_{\infty}[T].
\]
Conditional on no false alarm before the change, (conditional) expected detection delay is defined by 
\[
\operatorname{EDD}:=\mathbb{E}_k[T-k|T>k].
\]
The performance metric reported in Table~\ref{tab:d1-mid-arl-edd-sff}--Table~\ref{tab:d20-mid-arl-edd-sff} are \(\operatorname{EDD}\), False alarm  (proportion of sequences with an alarm raised before the true
changepoint), Success (the proportion of successful detections, with no false alarm and the changepoint is detected within $1000$ post-change observations), and Failure (failure to detect the
change within 1000 post-change observations).
\vspace{0.2cm}

\noindent\textbf{Univariate setting.}\enspace
The pre-change distribution is standard normal. The univariate post-change
distributions are \(q=\mathcal{N}(\delta,1)\), \(q=\mathcal{N}(0,\delta)\), \(q=\operatorname{Uniform}([-\delta,\delta])\), \(q=\operatorname{Laplace}(0,\delta)\),  \(q= \operatorname{Exp}(\delta)\).
For each method, \(L\) is calibrated so that the estimated in-control ARL is approximately equal to the target value of 500. The resulting ARL estimates are reported below, with standard errors in parentheses:
\begin{align*}
\text{OKAFF}&=539.90\;(\mathrm{se.}=13.49),\\
\text{NEWMA}&=502.43\;(\mathrm{se.}=16.42),\\
\text{Online RFF MMD}&=469.82\;(\mathrm{se.}=23.12),\\
\text{MMDEW}&=513.56\;(\mathrm{se.}=52.31).
\end{align*}
\noindent\textbf{Multivariate setting.}\enspace
The pre-change distribution is the \(d\)-dimensional standard normal
distribution with \(d=20\). The \(d\)-dimensional post-change distributions are
\(q=\mathcal{N}(\delta\mathbf{1},I)\),
\(q=\mathcal{N}(\boldsymbol{0},\delta I)\),
\(q=\operatorname{Uniform}([-\delta,\delta]^d)\) is a \(d\)-dimensional uniform distribution, \(q=\operatorname{Laplace}(0,\delta)^d\) is a \(d\)-dimensional Laplace distribution, 
and \(q=\operatorname{Exp}(\delta)^d\) is a \(d\)-dimensional exponential distribution. The coordinates of the uniform,
exponential, and Laplace distributions are mutually independent, each following an \(\operatorname{Uniform}([-\delta,\delta])\)  distribution,  \(\operatorname{Laplace}(0,\delta)\)  distribution, and \(\operatorname{Exp}(\delta)\) distribution, respectively.
For each method, \(L\) is calibrated so that the estimated in-control ARL is approximately equal to the target value of 500. The resulting ARL estimates are reported below, with standard errors in parentheses:
\begin{align*}
\text{OKAFF}&=517.32\;(\mathrm{se.}=8.85),\\
\text{NEWMA}&=518.19\;(\mathrm{se.}=18.80),\\
\text{Online RFF MMD}&=473.56\;(\mathrm{se.}=13.81),\\
\text{MMDEW}&=532.96\;(\mathrm{se.}=20.72).
\end{align*}
The results are reported in Tables~\ref{tab:d1-mid-arl-edd-sff}
and~\ref{tab:d20-mid-arl-edd-sff}. Success, false-alarm, and failure outcomes
are reported as proportions. Across most post-change alternatives in both
the univariate and 20-dimensional settings, the proposed method, OKAFF,
provides the strongest overall balance between EDD and detection success
while maintaining low false-alarm and failure proportions. Table~\ref{tab:d1-mid-arl-edd-sff} shows that OKAFF provides the most consistent
performance in the univariate experiments. Under weak changes, it maintains
high success proportions and low false-alarm proportions. By contrast, Online RFF MMD records low successful detections, high false alarm and failure rate in several weak-change
settings. Under stronger changes, Online RFF MMD or MMDEW can achieve a smaller EDD among
successful detections, but this advantage is often accompanied by a low proportion of successful detections. EDD should therefore be interpreted jointly with the
success.

The 20-dimensional results in Table~\ref{tab:d20-mid-arl-edd-sff} exhibit a
similar pattern. OKAFF remains reliable under weak  changes. Under strong alternatives, Online RFF MMD, MMDEW, and OKAFF
generally achieve high success proportions and short detection delays, with MMDEW often attaining the smallest EDD. Under weak change signals, however, MMDEW exhibits substantially larger detection delays and may fail to detect the change within the simulation horizon. Overall, adaptive calibration reveals a clear trade-off between detection
speed and reliability. OKAFF is the most robust method across weak and
moderate changes, whereas Online RFF MMD and MMDEW can respond more quickly
to strong changes. Because EDD is computed only from successful detections, it
can appear favorable when the success proportion is small; all four metrics
should therefore be considered together.

\clearpage

\begingroup
\scriptsize
\setlength{\tabcolsep}{2.5pt}
\renewcommand{\arraystretch}{1.08}
\begin{longtable}{@{\hspace{0.4em}}p{0.240\textwidth}p{0.11\textwidth}*{4}{>{\centering\arraybackslash}p{0.145\textwidth}}@{\hspace{0.4em}}}
\caption{Performance comparison of four methodsunder univariate
post-change alternatives, with the in-control ARL constrained to
\(500\pm40\). The symbol ``--'' indicates that the detection procedure does not successfully detect the change. For EDD, bold entries identify the smallest EDD among methods with Success \(\ge 0.5\). If the absolute smallest EDD is attained by a method with Success \(<0.5\), that EDD is underlined instead. For the remaining metrics, bold entries indicate the highest Success and the lowest False alarm and Failure proportions.}
\label{tab:d1-mid-arl-edd-sff}\\
\toprule
& Metric & OKAFF & NEWMA & Online RFF MMD & MMDEW\\
\midrule
\endfirsthead

\multicolumn{6}{c}{\tablename~\thetable{} continued}\\
\toprule
& Metric & OKAFF & NEWMA & Online RFF MMD & MMDEW\\
\midrule
\endhead

\midrule
\multicolumn{6}{r}{Continued on next page}\\
\endfoot

\bottomrule
\endlastfoot

\multicolumn{6}{@{}l}{\bfseries\boldmath
\(\mathcal{N}(0,1)\to\mathcal{N}(\delta,1)\)}\\
\addlinespace[2pt]

\textbf{\(\delta=3\)} & EDD
  & \makebox[5.0em][l]{\makebox[3.4em][r]{8.69}}
  & \makebox[5.0em][l]{\makebox[3.4em][r]{355.24}}
  & {\bfseries \makebox[5.0em][l]{\makebox[3.4em][r]{5.89}}}
  & \makebox[5.0em][l]{\makebox[3.4em][r]{378.65}}\\
& Success
  & {\bfseries \makebox[5.0em][l]{\makebox[3.4em][r]{0.916}}}
  & \makebox[5.0em][l]{\makebox[3.4em][r]{0.666}}
  & \makebox[5.0em][l]{\makebox[3.4em][r]{0.682}}
  & \makebox[5.0em][l]{\makebox[3.4em][r]{0.558}}\\
& False alarm
  & {\bfseries \makebox[5.0em][l]{\makebox[3.4em][r]{0.084}}}
  & \makebox[5.0em][l]{\makebox[3.4em][r]{0.318}}
  & \makebox[5.0em][l]{\makebox[3.4em][r]{0.312}}
  & \makebox[5.0em][l]{\makebox[3.4em][r]{0.380}}\\
& Failure
  & {\bfseries \makebox[5.0em][l]{\makebox[3.4em][r]{0.000}}}
  & \makebox[5.0em][l]{\makebox[3.4em][r]{0.016}}
  & \makebox[5.0em][l]{\makebox[3.4em][r]{0.006}}
  & \makebox[5.0em][l]{\makebox[3.4em][r]{0.062}}\\

\addlinespace[3pt]

\textbf{\(\delta=-4\)} & EDD
  & \makebox[5.0em][l]{\makebox[3.4em][r]{7.03}}
  & \makebox[5.0em][l]{\makebox[3.4em][r]{323.55}}
  & {\bfseries \makebox[5.0em][l]{\makebox[3.4em][r]{5.34}}}
  & \makebox[5.0em][l]{\makebox[3.4em][r]{385.32}}\\
& Success
  & {\bfseries \makebox[5.0em][l]{\makebox[3.4em][r]{0.926}}}
  & \makebox[5.0em][l]{\makebox[3.4em][r]{0.674}}
  & \makebox[5.0em][l]{\makebox[3.4em][r]{0.710}}
  & \makebox[5.0em][l]{\makebox[3.4em][r]{0.546}}\\
& False alarm
  & {\bfseries \makebox[5.0em][l]{\makebox[3.4em][r]{0.074}}}
  & \makebox[5.0em][l]{\makebox[3.4em][r]{0.322}}
  & \makebox[5.0em][l]{\makebox[3.4em][r]{0.290}}
  & \makebox[5.0em][l]{\makebox[3.4em][r]{0.348}}\\
& Failure
  & {\bfseries \makebox[5.0em][l]{\makebox[3.4em][r]{0.000}}}
  & \makebox[5.0em][l]{\makebox[3.4em][r]{0.004}}
  & {\bfseries \makebox[5.0em][l]{\makebox[3.4em][r]{0.000}}}
  & \makebox[5.0em][l]{\makebox[3.4em][r]{0.106}}\\

\addlinespace[4pt]
\multicolumn{6}{@{}l}{\bfseries\boldmath
\(\mathcal{N}(0,1)\to\mathcal{N}(0,\delta)\)}\\
\addlinespace[2pt]

\(\delta=0.1\) & EDD
  & {\bfseries \makebox[5.0em][l]{\makebox[3.4em][r]{25.63}}}
  & \makebox[5.0em][l]{\makebox[3.4em][r]{805.05}}
  & \makebox[5.0em][l]{\makebox[3.4em][r]{\underline{18.67}}}
  & \makebox[5.0em][l]{\makebox[3.4em][r]{351.26}}\\
& Success
  & {\bfseries \makebox[5.0em][l]{\makebox[3.4em][r]{0.898}}}
  & \makebox[5.0em][l]{\makebox[3.4em][r]{0.556}}
  & \makebox[5.0em][l]{\makebox[3.4em][r]{\underline{0.086}}}
  & \makebox[5.0em][l]{\makebox[3.4em][r]{0.608}}\\
& False alarm
  & {\bfseries \makebox[5.0em][l]{\makebox[3.4em][r]{0.102}}}
  & \makebox[5.0em][l]{\makebox[3.4em][r]{0.358}}
  & \makebox[5.0em][l]{\makebox[3.4em][r]{\underline{0.310}}}
  & \makebox[5.0em][l]{\makebox[3.4em][r]{0.392}}\\
& Failure
  & {\bfseries \makebox[5.0em][l]{\makebox[3.4em][r]{0.000}}}
  & \makebox[5.0em][l]{\makebox[3.4em][r]{0.086}}
  & \makebox[5.0em][l]{\makebox[3.4em][r]{\underline{0.604}}}
  & {\bfseries \makebox[5.0em][l]{\makebox[3.4em][r]{0.000}}}\\

\addlinespace[3pt]

\(\delta=6\) & EDD
  & {\bfseries \makebox[5.0em][l]{\makebox[3.4em][r]{21.08}}}
  & \makebox[5.0em][l]{\makebox[3.4em][r]{661.90}}
  & \makebox[5.0em][l]{\makebox[3.4em][r]{\underline{18.24}}}
  & \makebox[5.0em][l]{\makebox[3.4em][r]{26.43}}\\
& Success
  & {\bfseries \makebox[5.0em][l]{\makebox[3.4em][r]{0.906}}}
  & \makebox[5.0em][l]{\makebox[3.4em][r]{0.688}}
  & \makebox[5.0em][l]{\makebox[3.4em][r]{\underline{0.346}}}
  & \makebox[5.0em][l]{\makebox[3.4em][r]{0.140}}\\
& False alarm
  & {\bfseries \makebox[5.0em][l]{\makebox[3.4em][r]{0.086}}}
  & \makebox[5.0em][l]{\makebox[3.4em][r]{0.308}}
  & \makebox[5.0em][l]{\makebox[3.4em][r]{\underline{0.322}}}
  & \makebox[5.0em][l]{\makebox[3.4em][r]{0.378}}\\
& Failure
  & \makebox[5.0em][l]{\makebox[3.4em][r]{0.008}}
  & {\bfseries \makebox[5.0em][l]{\makebox[3.4em][r]{0.004}}}
  & \makebox[5.0em][l]{\makebox[3.4em][r]{\underline{0.332}}}
  & \makebox[5.0em][l]{\makebox[3.4em][r]{0.482}}\\

\addlinespace[4pt]
\multicolumn{6}{@{}l}{\bfseries\boldmath
\(\mathcal{N}(0,1)\to\operatorname{Uniform}([-\delta,\delta])\)}\\
\addlinespace[2pt]

\(\delta=0.5\) & EDD
  & {\bfseries \makebox[5.0em][l]{\makebox[3.4em][r]{24.61}}}
  & \makebox[5.0em][l]{\makebox[3.4em][r]{814.76}}
  & \makebox[5.0em][l]{\makebox[3.4em][r]{\underline{20.21}}}
  & \makebox[5.0em][l]{\makebox[3.4em][r]{740.94}}\\
& Success
  & {\bfseries \makebox[5.0em][l]{\makebox[3.4em][r]{0.912}}}
  & \makebox[5.0em][l]{\makebox[3.4em][r]{0.542}}
  & \makebox[5.0em][l]{\makebox[3.4em][r]{\underline{0.048}}}
  & \makebox[5.0em][l]{\makebox[3.4em][r]{0.598}}\\
& False alarm
  & {\bfseries \makebox[5.0em][l]{\makebox[3.4em][r]{0.088}}}
  & \makebox[5.0em][l]{\makebox[3.4em][r]{0.372}}
  & \makebox[5.0em][l]{\makebox[3.4em][r]{\underline{0.346}}}
  & \makebox[5.0em][l]{\makebox[3.4em][r]{0.366}}\\
& Failure
  & {\bfseries \makebox[5.0em][l]{\makebox[3.4em][r]{0.000}}}
  & \makebox[5.0em][l]{\makebox[3.4em][r]{0.086}}
  & \makebox[5.0em][l]{\makebox[3.4em][r]{\underline{0.606}}}
  & \makebox[5.0em][l]{\makebox[3.4em][r]{0.036}}\\

\addlinespace[3pt]

\(\delta=5\) & EDD
  & \makebox[5.0em][l]{\makebox[3.4em][r]{12.07}}
  & \makebox[5.0em][l]{\makebox[3.4em][r]{645.19}}
  & {\bfseries \makebox[5.0em][l]{\makebox[3.4em][r]{11.48}}}
  & \makebox[5.0em][l]{\makebox[3.4em][r]{176.25}}\\
& Success
  & {\bfseries \makebox[5.0em][l]{\makebox[3.4em][r]{0.920}}}
  & \makebox[5.0em][l]{\makebox[3.4em][r]{0.634}}
  & \makebox[5.0em][l]{\makebox[3.4em][r]{0.570}}
  & \makebox[5.0em][l]{\makebox[3.4em][r]{0.162}}\\
& False alarm
  & {\bfseries \makebox[5.0em][l]{\makebox[3.4em][r]{0.076}}}
  & \makebox[5.0em][l]{\makebox[3.4em][r]{0.358}}
  & \makebox[5.0em][l]{\makebox[3.4em][r]{0.286}}
  & \makebox[5.0em][l]{\makebox[3.4em][r]{0.344}}\\
& Failure
  & {\bfseries \makebox[5.0em][l]{\makebox[3.4em][r]{0.004}}}
  & \makebox[5.0em][l]{\makebox[3.4em][r]{0.008}}
  & \makebox[5.0em][l]{\makebox[3.4em][r]{0.144}}
  & \makebox[5.0em][l]{\makebox[3.4em][r]{0.494}}\\

\addlinespace[4pt]
\multicolumn{6}{@{}l}{\bfseries\boldmath
\(\mathcal{N}(0,1)\to\operatorname{Laplace}(0,\delta)\)}\\
\addlinespace[2pt]

\(\delta=0.3\) & EDD
  & {\bfseries \makebox[5.0em][l]{\makebox[3.4em][r]{29.94}}}
  & \makebox[5.0em][l]{\makebox[3.4em][r]{782.20}}
  & \makebox[5.0em][l]{\makebox[3.4em][r]{\underline{22.41}}}
  & \makebox[5.0em][l]{\makebox[3.4em][r]{174.51}}\\
& Success
  & {\bfseries \makebox[5.0em][l]{\makebox[3.4em][r]{0.940}}}
  & \makebox[5.0em][l]{\makebox[3.4em][r]{0.626}}
  & \makebox[5.0em][l]{\makebox[3.4em][r]{\underline{0.068}}}
  & \makebox[5.0em][l]{\makebox[3.4em][r]{0.678}}\\
& False alarm
  & {\bfseries \makebox[5.0em][l]{\makebox[3.4em][r]{0.060}}}
  & \makebox[5.0em][l]{\makebox[3.4em][r]{0.328}}
  & \makebox[5.0em][l]{\makebox[3.4em][r]{\underline{0.326}}}
  & \makebox[5.0em][l]{\makebox[3.4em][r]{0.322}}\\
& Failure
  & {\bfseries \makebox[5.0em][l]{\makebox[3.4em][r]{0.000}}}
  & \makebox[5.0em][l]{\makebox[3.4em][r]{0.046}}
  & \makebox[5.0em][l]{\makebox[3.4em][r]{\underline{0.606}}}
  & {\bfseries \makebox[5.0em][l]{\makebox[3.4em][r]{0.000}}}\\

\addlinespace[3pt]

\(\delta=3\) & EDD
  & {\bfseries \makebox[5.0em][l]{\makebox[3.4em][r]{15.15}}}
  & \makebox[5.0em][l]{\makebox[3.4em][r]{701.90}}
  & \makebox[5.0em][l]{\makebox[3.4em][r]{13.80}}
  & \makebox[5.0em][l]{\makebox[3.4em][r]{\underline{2.60}}}\\
& Success
  & {\bfseries \makebox[5.0em][l]{\makebox[3.4em][r]{0.922}}}
  & \makebox[5.0em][l]{\makebox[3.4em][r]{0.656}}
  & \makebox[5.0em][l]{\makebox[3.4em][r]{0.384}}
  & \makebox[5.0em][l]{\makebox[3.4em][r]{\underline{0.144}}}\\
& False alarm
  & {\bfseries \makebox[5.0em][l]{\makebox[3.4em][r]{0.068}}}
  & \makebox[5.0em][l]{\makebox[3.4em][r]{0.342}}
  & \makebox[5.0em][l]{\makebox[3.4em][r]{0.304}}
  & \makebox[5.0em][l]{\makebox[3.4em][r]{\underline{0.362}}}\\
& Failure
  & \makebox[5.0em][l]{\makebox[3.4em][r]{0.010}}
  & {\bfseries \makebox[5.0em][l]{\makebox[3.4em][r]{0.002}}}
  & \makebox[5.0em][l]{\makebox[3.4em][r]{0.312}}
  & \makebox[5.0em][l]{\makebox[3.4em][r]{\underline{0.494}}}\\

\addlinespace[4pt]
\multicolumn{6}{@{}l}{\bfseries\boldmath
\(\mathcal{N}(0,1)\to\operatorname{Exp}(\delta)\)}\\
\addlinespace[2pt]

\(\delta=1\) & EDD
  & {\bfseries \makebox[5.0em][l]{\makebox[3.4em][r]{46.61}}}
  & \makebox[5.0em][l]{\makebox[3.4em][r]{689.09}}
  & \makebox[5.0em][l]{\makebox[3.4em][r]{\underline{16.95}}}
  & \makebox[5.0em][l]{\makebox[3.4em][r]{242.49}}\\
& Success
  & {\bfseries \makebox[5.0em][l]{\makebox[3.4em][r]{0.928}}}
  & \makebox[5.0em][l]{\makebox[3.4em][r]{0.650}}
  & \makebox[5.0em][l]{\makebox[3.4em][r]{\underline{0.222}}}
  & \makebox[5.0em][l]{\makebox[3.4em][r]{0.568}}\\
& False alarm
  & {\bfseries \makebox[5.0em][l]{\makebox[3.4em][r]{0.072}}}
  & \makebox[5.0em][l]{\makebox[3.4em][r]{0.348}}
  & \makebox[5.0em][l]{\makebox[3.4em][r]{\underline{0.304}}}
  & \makebox[5.0em][l]{\makebox[3.4em][r]{0.354}}\\
& Failure
  & {\bfseries \makebox[5.0em][l]{\makebox[3.4em][r]{0.000}}}
  & \makebox[5.0em][l]{\makebox[3.4em][r]{0.002}}
  & \makebox[5.0em][l]{\makebox[3.4em][r]{\underline{0.474}}}
  & \makebox[5.0em][l]{\makebox[3.4em][r]{0.078}}\\

\addlinespace[3pt]

\(\delta=8\) & EDD
  & {\bfseries \makebox[5.0em][l]{\makebox[3.4em][r]{8.80}}}
  & \makebox[5.0em][l]{\makebox[3.4em][r]{659.38}}
  & \makebox[5.0em][l]{\makebox[3.4em][r]{9.08}}
  & \makebox[5.0em][l]{\makebox[3.4em][r]{\underline{1.30}}}\\
& Success
  & {\bfseries \makebox[5.0em][l]{\makebox[3.4em][r]{0.916}}}
  & \makebox[5.0em][l]{\makebox[3.4em][r]{0.632}}
  & \makebox[5.0em][l]{\makebox[3.4em][r]{0.592}}
  & \makebox[5.0em][l]{\makebox[3.4em][r]{\underline{0.162}}}\\
& False alarm
  & {\bfseries \makebox[5.0em][l]{\makebox[3.4em][r]{0.084}}}
  & \makebox[5.0em][l]{\makebox[3.4em][r]{0.342}}
  & \makebox[5.0em][l]{\makebox[3.4em][r]{0.290}}
  & \makebox[5.0em][l]{\makebox[3.4em][r]{\underline{0.404}}}\\
& Failure
  & {\bfseries \makebox[5.0em][l]{\makebox[3.4em][r]{0.000}}}
  & \makebox[5.0em][l]{\makebox[3.4em][r]{0.026}}
  & \makebox[5.0em][l]{\makebox[3.4em][r]{0.118}}
  & \makebox[5.0em][l]{\makebox[3.4em][r]{\underline{0.434}}}\\

\end{longtable}
\endgroup

\begingroup
\scriptsize
\setlength{\tabcolsep}{2.5pt}
\renewcommand{\arraystretch}{0.95}
\begin{longtable}{@{}>{\hspace{0.4em}}p{0.265\textwidth}p{0.11\textwidth}*{4}{>{\centering\arraybackslash}p{0.145\textwidth}}@{}}
\caption{Performance comparison of four methods under 20-dimensional
post-change alternatives, with the in-control ARL constrained to
\(500\pm40\). The symbol ``--'' indicates that the detection procedure does not successfully detect the change. For EDD, bold entries identify the smallest EDD among methods with Success $\ge 0.5$. If the absolute smallest EDD is attained by a method with Success $<0.5$, that EDD is underlined instead. For the remaining metrics, bold entries indicate the highest Success and the lowest False alarm and Failure proportions.}
\label{tab:d20-mid-arl-edd-sff}\\
\toprule
& Metric & OKAFF & NEWMA & Online RFF MMD & MMDEW\\
\midrule
\endfirsthead
\multicolumn{6}{c}{\tablename~\thetable{} continued}\\
\toprule
& Metric & OKAFF & NEWMA & Online RFF MMD & MMDEW\\
\midrule
\endhead
\midrule
\multicolumn{6}{r}{Continued on next page}\\
\endfoot
\bottomrule
\endlastfoot

\multicolumn{6}{@{}l}{\hspace{0.4em}\bfseries\boldmath
\(\mathcal{N}(\boldsymbol{0},I)
\to\mathcal{N}(\delta\boldsymbol{1},I)\)}\\
\addlinespace[2pt]

\(\delta=1\) & EDD
& \makebox[5.0em][l]{\makebox[3.4em][r]{9.11}}
& \makebox[5.0em][l]{\makebox[3.4em][r]{764.28}}
& {\bfseries \makebox[5.0em][l]{\makebox[3.4em][r]{4.07}}}
& \makebox[5.0em][l]{\makebox[3.4em][r]{179.08}}\\
& Success
& {\bfseries \makebox[5.0em][l]{\makebox[3.4em][r]{0.950}}}
& \makebox[5.0em][l]{\makebox[3.4em][r]{0.582}}
& \makebox[5.0em][l]{\makebox[3.4em][r]{0.912}}
& \makebox[5.0em][l]{\makebox[3.4em][r]{0.782}}\\
& False alarm
& {\bfseries \makebox[5.0em][l]{\makebox[3.4em][r]{0.050}}}
& \makebox[5.0em][l]{\makebox[3.4em][r]{0.414}}
& \makebox[5.0em][l]{\makebox[3.4em][r]{0.088}}
& \makebox[5.0em][l]{\makebox[3.4em][r]{0.142}}\\
& Failure
& {\bfseries \makebox[5.0em][l]{\makebox[3.4em][r]{0.000}}}
& \makebox[5.0em][l]{\makebox[3.4em][r]{0.004}}
& {\bfseries \makebox[5.0em][l]{\makebox[3.4em][r]{0.000}}}
& \makebox[5.0em][l]{\makebox[3.4em][r]{0.076}}\\

\addlinespace[3pt]

\(\delta=-4\) & EDD
& \makebox[5.0em][l]{\makebox[3.4em][r]{3.39}}
& \makebox[5.0em][l]{\makebox[3.4em][r]{71.74}}
& \makebox[5.0em][l]{\makebox[3.4em][r]{1.01}}
& {\bfseries \makebox[5.0em][l]{\makebox[3.4em][r]{1.00}}}\\
& Success
& {\bfseries \makebox[5.0em][l]{\makebox[3.4em][r]{0.934}}}
& \makebox[5.0em][l]{\makebox[3.4em][r]{0.580}}
& \makebox[5.0em][l]{\makebox[3.4em][r]{0.888}}
& \makebox[5.0em][l]{\makebox[3.4em][r]{0.844}}\\
& False alarm
& {\bfseries \makebox[5.0em][l]{\makebox[3.4em][r]{0.066}}}
& \makebox[5.0em][l]{\makebox[3.4em][r]{0.414}}
& \makebox[5.0em][l]{\makebox[3.4em][r]{0.112}}
& \makebox[5.0em][l]{\makebox[3.4em][r]{0.156}}\\
& Failure
& {\bfseries \makebox[5.0em][l]{\makebox[3.4em][r]{0.000}}}
& \makebox[5.0em][l]{\makebox[3.4em][r]{0.006}}
& {\bfseries \makebox[5.0em][l]{\makebox[3.4em][r]{0.000}}}
& {\bfseries \makebox[5.0em][l]{\makebox[3.4em][r]{0.000}}}\\

\addlinespace[4pt]

\multicolumn{6}{@{}l}{\hspace{0.4em}\bfseries\boldmath
\(\mathcal{N}(\boldsymbol{0},I)
\to\mathcal{N}(\boldsymbol{0},\delta I)\)}\\
\addlinespace[2pt]

\(\delta=0.3\) & EDD
& {\bfseries \makebox[5.0em][l]{\makebox[3.4em][r]{12.67}}}
& \makebox[5.0em][l]{\makebox[3.4em][r]{861.44}}
& \makebox[5.0em][l]{\makebox[3.4em][r]{25.03}}
& \makebox[5.0em][l]{\makebox[3.4em][r]{671.29}}\\
& Success
& {\bfseries \makebox[5.0em][l]{\makebox[3.4em][r]{0.942}}}
& \makebox[5.0em][l]{\makebox[3.4em][r]{0.592}}
& \makebox[5.0em][l]{\makebox[3.4em][r]{0.880}}
& \makebox[5.0em][l]{\makebox[3.4em][r]{0.828}}\\
& False alarm
& {\bfseries \makebox[5.0em][l]{\makebox[3.4em][r]{0.058}}}
& \makebox[5.0em][l]{\makebox[3.4em][r]{0.386}}
& \makebox[5.0em][l]{\makebox[3.4em][r]{0.106}}
& \makebox[5.0em][l]{\makebox[3.4em][r]{0.168}}\\
& Failure
& {\bfseries \makebox[5.0em][l]{\makebox[3.4em][r]{0.000}}}
& \makebox[5.0em][l]{\makebox[3.4em][r]{0.022}}
& \makebox[5.0em][l]{\makebox[3.4em][r]{0.014}}
& \makebox[5.0em][l]{\makebox[3.4em][r]{0.004}}\\

\addlinespace[3pt]

\( \delta=2.5\) & EDD
& \makebox[5.0em][l]{\makebox[3.4em][r]{6.33}}
& \makebox[5.0em][l]{\makebox[3.4em][r]{808.32}}
& \makebox[5.0em][l]{\makebox[3.4em][r]{4.61}}
& {\bfseries \makebox[5.0em][l]{\makebox[3.4em][r]{2.06}}}\\
& Success
& {\bfseries \makebox[5.0em][l]{\makebox[3.4em][r]{0.938}}}
& \makebox[5.0em][l]{\makebox[3.4em][r]{0.618}}
& \makebox[5.0em][l]{\makebox[3.4em][r]{0.906}}
& \makebox[5.0em][l]{\makebox[3.4em][r]{0.896}}\\
& False alarm
& {\bfseries \makebox[5.0em][l]{\makebox[3.4em][r]{0.062}}}
& \makebox[5.0em][l]{\makebox[3.4em][r]{0.382}}
& \makebox[5.0em][l]{\makebox[3.4em][r]{0.094}}
& \makebox[5.0em][l]{\makebox[3.4em][r]{0.102}}\\
& Failure
& {\bfseries \makebox[5.0em][l]{\makebox[3.4em][r]{0.000}}}
& {\bfseries \makebox[5.0em][l]{\makebox[3.4em][r]{0.000}}}
& {\bfseries \makebox[5.0em][l]{\makebox[3.4em][r]{0.000}}}
& \makebox[5.0em][l]{\makebox[3.4em][r]{0.002}}\\

\addlinespace[4pt]

\multicolumn{6}{@{}l}{\hspace{0.4em}\bfseries\boldmath
\(\mathcal{N}(\boldsymbol{0},I)
\to\operatorname{Uniform}([-\delta,\delta]^{20})\)}\\
\addlinespace[2pt]

\(\delta=0.5\) & EDD
& {\bfseries \makebox[5.0em][l]{\makebox[3.4em][r]{8.87}}}
& \makebox[5.0em][l]{\makebox[3.4em][r]{966.67}}
& \makebox[5.0em][l]{\makebox[3.4em][r]{15.63}}
& \makebox[5.0em][l]{\makebox[3.4em][r]{--}}\\
& Success
& {\bfseries \makebox[5.0em][l]{\makebox[3.4em][r]{0.928}}}
& \makebox[5.0em][l]{\makebox[3.4em][r]{0.086}}
& \makebox[5.0em][l]{\makebox[3.4em][r]{0.898}}
& \makebox[5.0em][l]{\makebox[3.4em][r]{0.000}}\\
& False alarm
& {\bfseries \makebox[5.0em][l]{\makebox[3.4em][r]{0.072}}}
& \makebox[5.0em][l]{\makebox[3.4em][r]{0.438}}
& \makebox[5.0em][l]{\makebox[3.4em][r]{0.102}}
& \makebox[5.0em][l]{\makebox[3.4em][r]{0.128}}\\
& Failure
& {\bfseries \makebox[5.0em][l]{\makebox[3.4em][r]{0.000}}}
& \makebox[5.0em][l]{\makebox[3.4em][r]{0.476}}
& {\bfseries \makebox[5.0em][l]{\makebox[3.4em][r]{0.000}}}
& \makebox[5.0em][l]{\makebox[3.4em][r]{0.872}}\\

\addlinespace[3pt]

\(\delta=1\) & EDD
& {\bfseries \makebox[5.0em][l]{\makebox[3.4em][r]{14.23}}}
& \makebox[5.0em][l]{\makebox[3.4em][r]{843.44}}
& \makebox[5.0em][l]{\makebox[3.4em][r]{27.74}}
& \makebox[5.0em][l]{\makebox[3.4em][r]{800.76}}\\
& Success
& {\bfseries \makebox[5.0em][l]{\makebox[3.4em][r]{0.940}}}
& \makebox[5.0em][l]{\makebox[3.4em][r]{0.536}}
& \makebox[5.0em][l]{\makebox[3.4em][r]{0.842}}
& \makebox[5.0em][l]{\makebox[3.4em][r]{0.568}}\\
& False alarm
& {\bfseries \makebox[5.0em][l]{\makebox[3.4em][r]{0.060}}}
& \makebox[5.0em][l]{\makebox[3.4em][r]{0.452}}
& \makebox[5.0em][l]{\makebox[3.4em][r]{0.122}}
& \makebox[5.0em][l]{\makebox[3.4em][r]{0.142}}\\
& Failure
& {\bfseries \makebox[5.0em][l]{\makebox[3.4em][r]{0.000}}}
& \makebox[5.0em][l]{\makebox[3.4em][r]{0.012}}
& \makebox[5.0em][l]{\makebox[3.4em][r]{0.036}}
& \makebox[5.0em][l]{\makebox[3.4em][r]{0.290}}\\

\addlinespace[4pt]

\multicolumn{6}{@{}l}{\hspace{0.4em}\bfseries\boldmath
\(\mathcal{N}(\boldsymbol{0},I)
\to\operatorname{Laplace}(0,\delta)^{20}\)}\\
\addlinespace[2pt]

\(\delta=0.3\) & EDD
& {\bfseries \makebox[5.0em][l]{\makebox[3.4em][r]{10.53}}}
& \makebox[5.0em][l]{\makebox[3.4em][r]{914.23}}
& \makebox[5.0em][l]{\makebox[3.4em][r]{19.66}}
& \makebox[5.0em][l]{\makebox[3.4em][r]{614.26}}\\
& Success
& {\bfseries \makebox[5.0em][l]{\makebox[3.4em][r]{0.940}}}
& \makebox[5.0em][l]{\makebox[3.4em][r]{0.432}}
& \makebox[5.0em][l]{\makebox[3.4em][r]{0.900}}
& \makebox[5.0em][l]{\makebox[3.4em][r]{0.862}}\\
& False alarm
& {\bfseries \makebox[5.0em][l]{\makebox[3.4em][r]{0.060}}}
& \makebox[5.0em][l]{\makebox[3.4em][r]{0.414}}
& \makebox[5.0em][l]{\makebox[3.4em][r]{0.100}}
& \makebox[5.0em][l]{\makebox[3.4em][r]{0.138}}\\
& Failure
& {\bfseries \makebox[5.0em][l]{\makebox[3.4em][r]{0.000}}}
& \makebox[5.0em][l]{\makebox[3.4em][r]{0.154}}
& {\bfseries \makebox[5.0em][l]{\makebox[3.4em][r]{0.000}}}
& {\bfseries \makebox[5.0em][l]{\makebox[3.4em][r]{0.000}}}\\

\addlinespace[3pt]

\(\delta=2\) & EDD
& \makebox[5.0em][l]{\makebox[3.4em][r]{3.59}}
& \makebox[5.0em][l]{\makebox[3.4em][r]{773.63}}
& \makebox[5.0em][l]{\makebox[3.4em][r]{1.22}}
& {\bfseries \makebox[5.0em][l]{\makebox[3.4em][r]{1.02}}}\\
& Success
& {\bfseries \makebox[5.0em][l]{\makebox[3.4em][r]{0.932}}}
& \makebox[5.0em][l]{\makebox[3.4em][r]{0.552}}
& \makebox[5.0em][l]{\makebox[3.4em][r]{0.910}}
& \makebox[5.0em][l]{\makebox[3.4em][r]{0.846}}\\
& False alarm
& {\bfseries \makebox[5.0em][l]{\makebox[3.4em][r]{0.068}}}
& \makebox[5.0em][l]{\makebox[3.4em][r]{0.406}}
& \makebox[5.0em][l]{\makebox[3.4em][r]{0.090}}
& \makebox[5.0em][l]{\makebox[3.4em][r]{0.154}}\\
& Failure
& {\bfseries \makebox[5.0em][l]{\makebox[3.4em][r]{0.000}}}
& \makebox[5.0em][l]{\makebox[3.4em][r]{0.042}}
& {\bfseries \makebox[5.0em][l]{\makebox[3.4em][r]{0.000}}}
& {\bfseries \makebox[5.0em][l]{\makebox[3.4em][r]{0.000}}}\\

\addlinespace[4pt]

\multicolumn{6}{@{}l}{\hspace{0.4em}\bfseries\boldmath
\(\mathcal{N}(\boldsymbol{0},I)
\to\operatorname{Exp}(\delta)^{20}\)}\\
\addlinespace[2pt]

\(\delta=0.1\) & EDD
& {\bfseries \makebox[5.0em][l]{\makebox[3.4em][r]{8.11}}}
& \makebox[5.0em][l]{\makebox[3.4em][r]{--}}
& \makebox[5.0em][l]{\makebox[3.4em][r]{10.48}}
& \makebox[5.0em][l]{\makebox[3.4em][r]{--}}\\
& Success
& {\bfseries \makebox[5.0em][l]{\makebox[3.4em][r]{0.940}}}
& \makebox[5.0em][l]{\makebox[3.4em][r]{0.000}}
& \makebox[5.0em][l]{\makebox[3.4em][r]{0.916}}
& \makebox[5.0em][l]{\makebox[3.4em][r]{0.000}}\\
& False alarm
& {\bfseries \makebox[5.0em][l]{\makebox[3.4em][r]{0.060}}}
& \makebox[5.0em][l]{\makebox[3.4em][r]{0.422}}
& \makebox[5.0em][l]{\makebox[3.4em][r]{0.084}}
& \makebox[5.0em][l]{\makebox[3.4em][r]{0.130}}\\
& Failure
& {\bfseries \makebox[5.0em][l]{\makebox[3.4em][r]{0.000}}}
& \makebox[5.0em][l]{\makebox[3.4em][r]{0.578}}
& {\bfseries \makebox[5.0em][l]{\makebox[3.4em][r]{0.000}}}
& \makebox[5.0em][l]{\makebox[3.4em][r]{0.870}}\\

\addlinespace[3pt]

\(\delta=2\) & EDD
& \makebox[5.0em][l]{\makebox[3.4em][r]{3.55}}
& \makebox[5.0em][l]{\makebox[3.4em][r]{666.16}}
& \makebox[5.0em][l]{\makebox[3.4em][r]{1.17}}
& {\bfseries \makebox[5.0em][l]{\makebox[3.4em][r]{1.01}}}\\
& Success
& {\bfseries \makebox[5.0em][l]{\makebox[3.4em][r]{0.954}}}
& \makebox[5.0em][l]{\makebox[3.4em][r]{0.572}}
& \makebox[5.0em][l]{\makebox[3.4em][r]{0.904}}
& \makebox[5.0em][l]{\makebox[3.4em][r]{0.864}}\\
& False alarm
& {\bfseries \makebox[5.0em][l]{\makebox[3.4em][r]{0.046}}}
& \makebox[5.0em][l]{\makebox[3.4em][r]{0.414}}
& \makebox[5.0em][l]{\makebox[3.4em][r]{0.096}}
& \makebox[5.0em][l]{\makebox[3.4em][r]{0.136}}\\
& Failure
& {\bfseries \makebox[5.0em][l]{\makebox[3.4em][r]{0.000}}}
& \makebox[5.0em][l]{\makebox[3.4em][r]{0.014}}
& {\bfseries \makebox[5.0em][l]{\makebox[3.4em][r]{0.000}}}
& {\bfseries \makebox[5.0em][l]{\makebox[3.4em][r]{0.000}}}\\

\end{longtable}
\endgroup

\clearpage

\subsubsection{Comparison of average runtime}

\label{supp:sec:avruntime}

\paragraph*{Average runtime under different sequence lengths.}
To compare runtime as sequence length increases, we fix \(d=20\) and consider sequence lengths ranging from 10 to \(200{,}000\) observations. For each of 10 repetitions, we generate an independent standard normal data stream, with all methods evaluated on the same stream within each repetition. Figure~\ref{fig:runtime-total-comparison-d20} reports the average streaming runtime, excluding initialization costs. OKAFF and NewMA achieve the lowest streaming runtimes, while Online RFF MMD incurs moderately higher computational cost and MMDEW is substantially slower. ScanB and OK-CUSUM are considerably more expensive. In the Figure, both axes are displayed on logarithmic scales. Each observation retains its original values $(x,y)$, but its position is determined by
\(
\bigl(\log_{10}x,\log_{10}y\bigr),
\)
where $x$ denotes the sequence length and $y$ denotes the average streaming runtime in seconds. These results show that OKAFF combines competitive  detection performance with low computational cost, making it well suited for online changepoint detection.

\par\medskip

{\centering
\includegraphics[
    width=0.6\linewidth
]{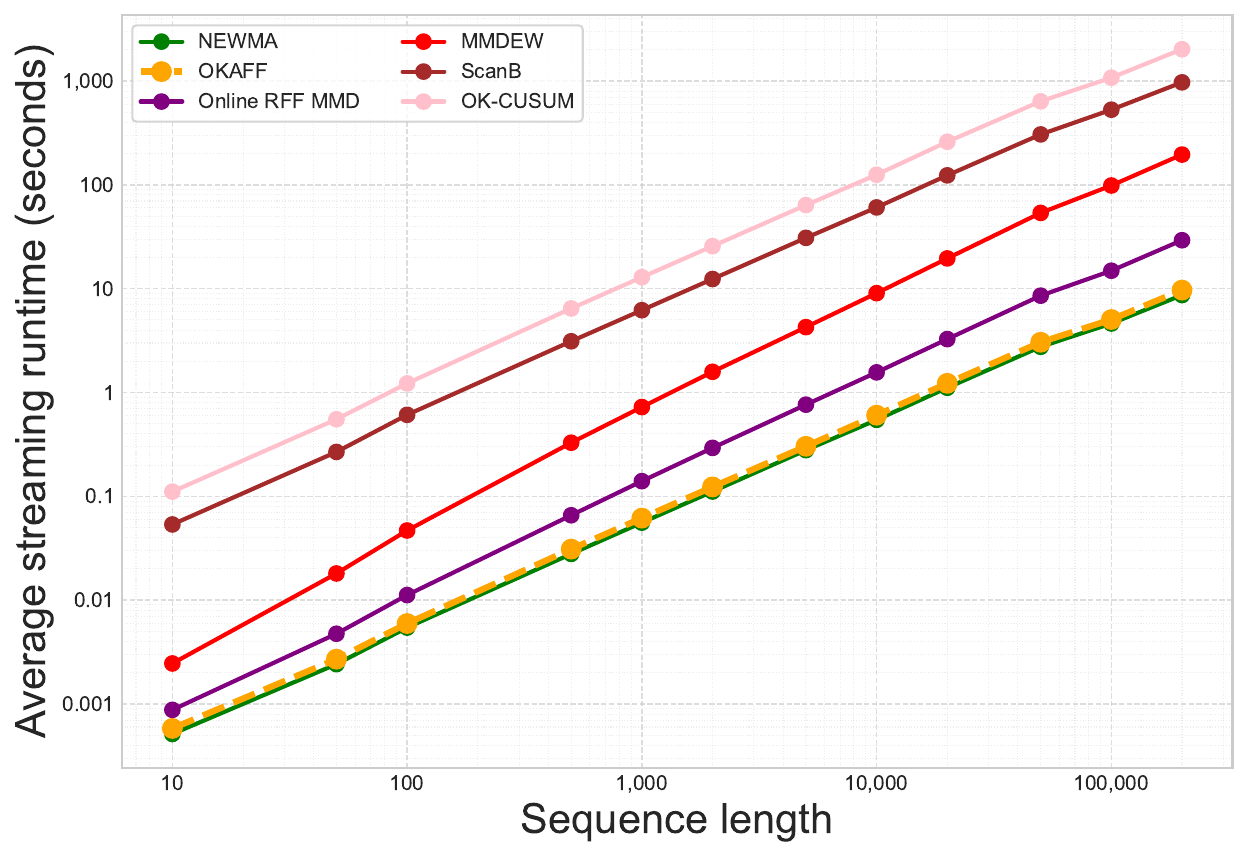}

\captionof{figure}{Average streaming runtime versus sequence length for
standard normal data with \(d=20\), comparing  NEWMA, OKAFF, Online RFF MMD, MMDEW, ScanB, and OK-CUSUM.}
\label{fig:runtime-total-comparison-d20}
\par}

\paragraph*{Average runtime across different dimensions}
Figure~\ref{fig:runtime-dimension-comparison} examines how computational cost
changes with dimension. We consider
\(d\in\{1,2,5,10,20,50,100,1000,5000\}\). For each dimension, a monitoring
stream of length \(T=20{,}000\) is generated from the standard multivariate normal
distribution. The reported runtimes are averages over 10 repetitions. Within
each repetition, every method receives the same reference sample and monitoring
stream. The panel reports streaming time, excluding initialization, for OKAFF, MMDEW,
NEWMA, and Online RFF MMD, ScanB and
OK-CUSUM.

{\centering
\includegraphics[
    width=0.6\linewidth
]{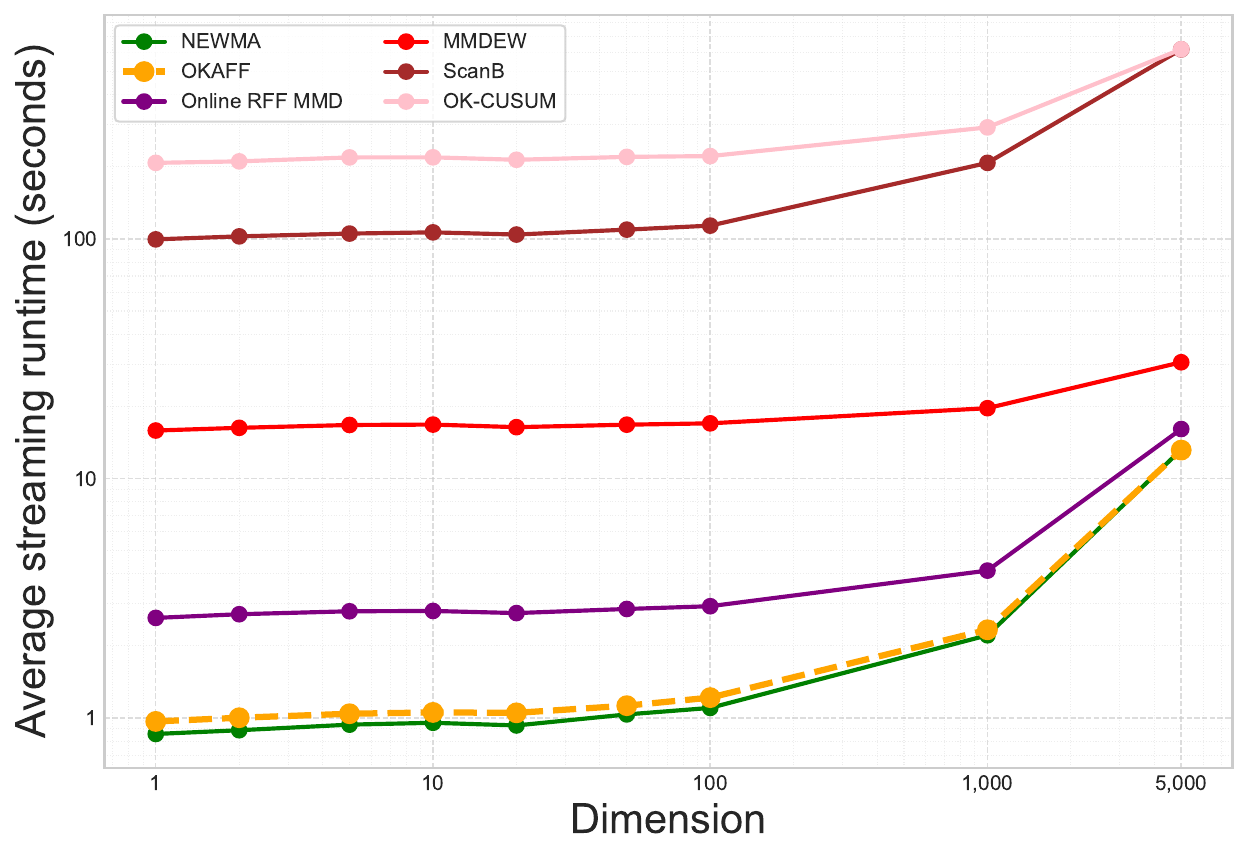}

\captionof{figure}{Average streaming runtime versus dimension, comparing NewMA, OKAFF,Online RFF MMD, MMDEW, ScanB and OK-CUSUM. The sequence length is
\(T=20{,}000\), and runtimes are averaged over 10 repetitions.}
\label{fig:runtime-dimension-comparison}
\par}

\subsubsection{Comparison of different numbers of random Fourier features}

\label{rffnumcom}

\paragraph*{Dimension $d=20$.} The following experiments examine the effect of the number of random features on detection performance. We consider \(500\), \(1000\), and \(5000\) random features in the multivariate setting. The null distribution is \(\mathcal{N}(\mathbf{0},I)\), and the alternative distribution is \(\mathcal{N}(\mathbf{0},0.3I)\). The results are shown in Figure~\ref{fig:rff-number-covdiag03}. Detection performance remains generally similar as the number of random features increases from \(500\) to \(5000\), suggesting that \(500\) random features are sufficient for the considered alternative and that further increases provide improvement in detection performance. Similar behavior is observed for the other alternatives considered in the \(d=20\) experiments. We therefore use \(500\) random features in the main experiments.

\par\medskip

{\centering
\includegraphics[
    width=0.6\linewidth
]{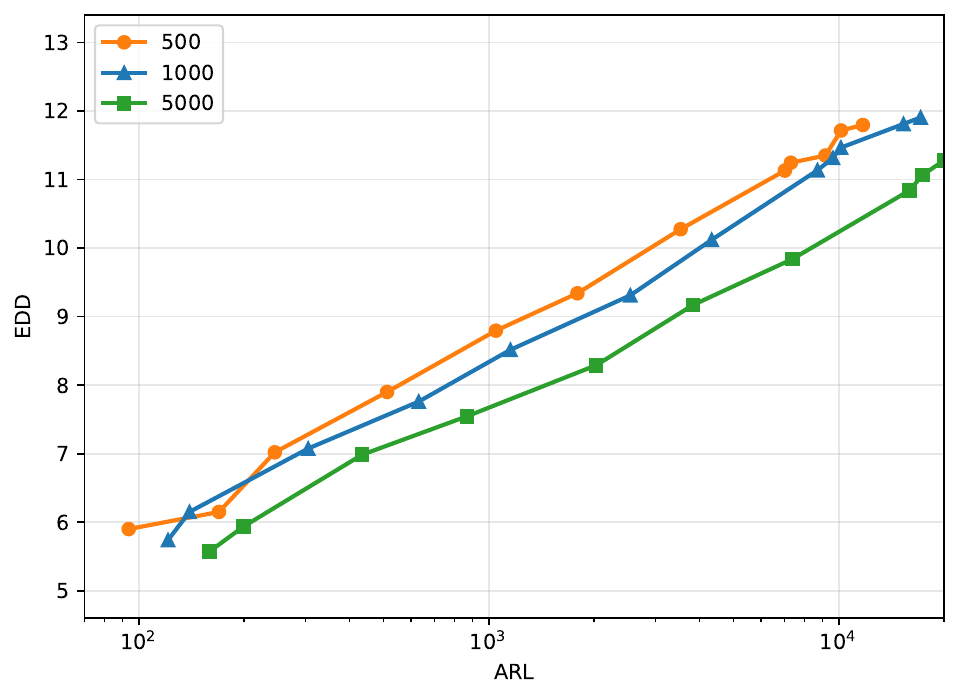}

\captionof{figure}{Effect of the number of random features on detection
performance for multivariate data. The panel compares \(500\), \(1000\), and
\(5000\) random features. The horizontal axis showing ARL is log-scaled,
whereas the vertical axis showing EDD unscaled.}
\label{fig:rff-number-covdiag03}
\par}

\paragraph*{Dimension $d=1$.} We further examine the effect of the number of random features on detection performance in the univariate setting in Figure~\ref{fig:rff-number-mean03-d1}. We consider \(500\), \(1000\), and \(5000\) random features, with null distribution \(\mathcal{N}(0,1)\) and alternative distribution \(\mathcal{N}(-3,1)\). The results show that detection performance remains generally similar as the number of random features increases from \(500\) to \(5000\). This suggests that \(500\) random features are sufficient for the considered setting, with little additional improvement obtained from using a larger number of features. Similar behavior is observed for the other univariate alternatives considered in our experiments.

\begin{figure}[!htbp]
\centering

\includegraphics[
  width=0.6\linewidth
]{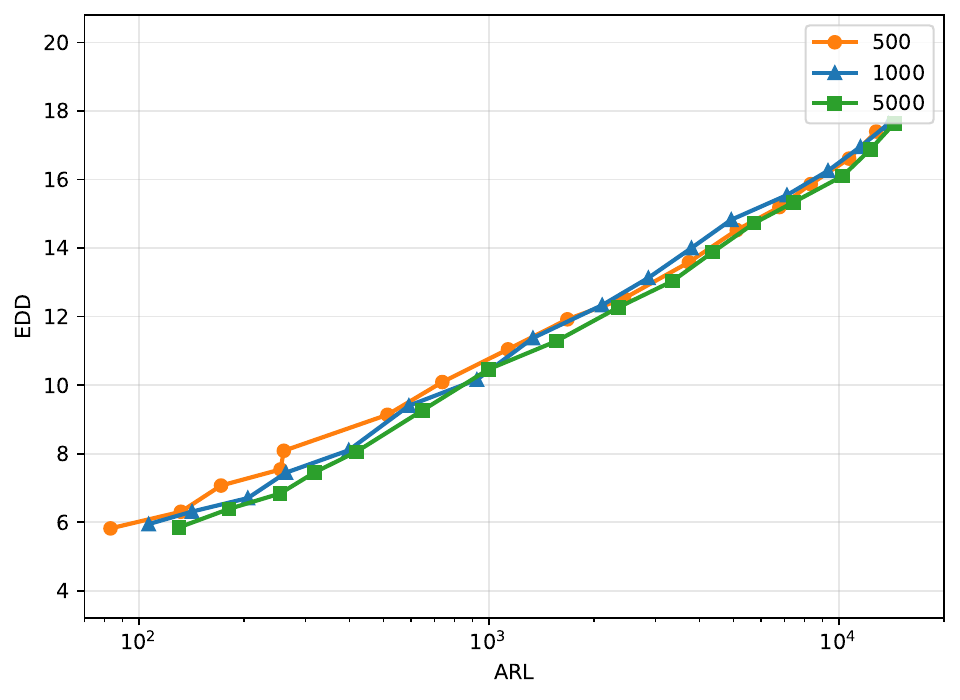}

\caption{Effect of the number of random features on detection performance
for univariate data. The panel compares \(500\), \(1000\), and \(5000\)
random features. The horizontal axis showing ARL is log-scaled,
whereas the vertical axis showing EDD unscaled.}
\label{fig:rff-number-mean03-d1}

\end{figure}

\newpage

\subsubsection{Behaviour of the monitoring statistic under null and alternative distributions}

\label{sec:supp:monstat}

\begin{figure}[!htbp]
\centering

\includegraphics[
  width=0.49\linewidth
]{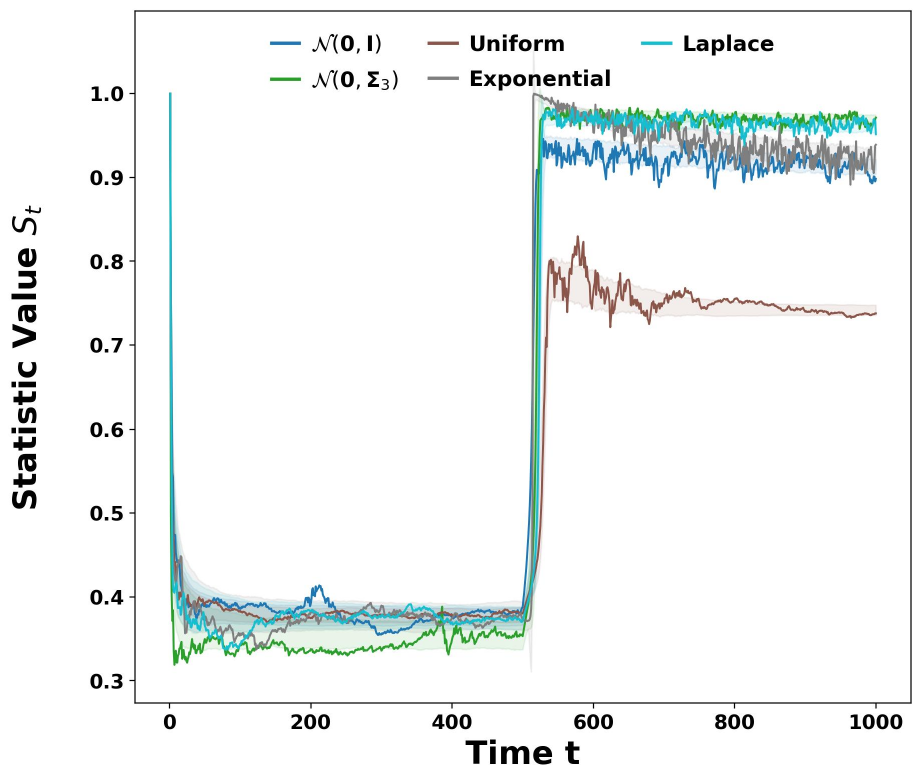}
\hfill
\includegraphics[
  width=0.49\linewidth
]{Figures/boxplot/diffpre/all_prechange_cases_average_paths_empirical_only.pdf}

\caption{Sample paths of the monitoring statistic under different pre-change
distributions and the same post-change distribution, with the change
point occurring at \(t=501\). The left panel shows an individual sample
path with estimated standard-deviation bands, while the right panel shows
the average sample paths with the corresponding estimated
standard-deviation bands.}
\label{fig:average-individual-samplepath-diffpre}

\end{figure}

\begin{figure}[!htbp]
\centering

\includegraphics[
  width=0.49\linewidth
]{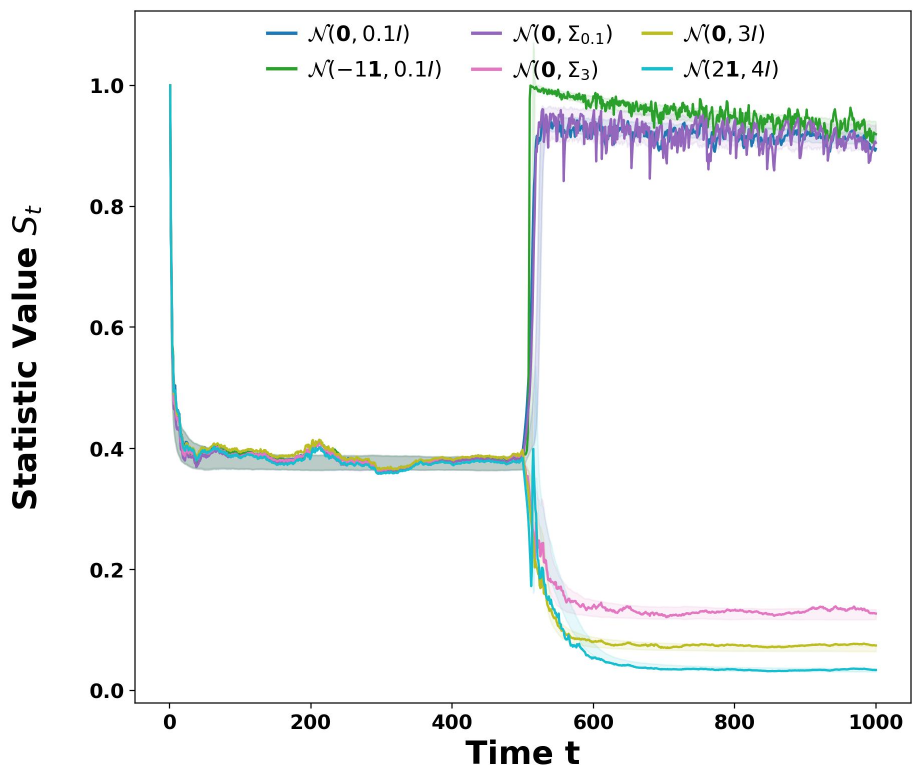}
\hfill
\includegraphics[
  width=0.49\linewidth
]{Figures/boxplot/Gaussian/all_gaussian_cases_average_paths.pdf}

\caption{Sample paths of the monitoring statistic under a common pre-change Gaussian distribution
and different post-change Gaussian distributions, with the changepoint occurring at
\(t=501\). The left panel shows an individual sample path with estimated
standard-deviation bands, while the right panel shows the average sample
paths with the corresponding estimated standard-deviation bands.}
\label{fig:average-individual-samplepath-gaussian}

\end{figure}

In Section~\ref{bemostatnullalt} we visualized the distribution of the OKAFF monitoring statistic $S_t$ before and after a changepoint across various scenarios. This provided a qualitative 
assessment of the $S_t$'s ability to react to a changepoint.
This section provides additional plots showing the behaviour of the statistic $S_t$ for different types of changes in distribution.
Greater separation between the pre- and post-change distributions facilitates reliable detection, allowing a suitably calibrated threshold to achieve high detection power while controlling false alarms. We further consider uniform, exponential, and Laplace null distributions. For each setting, empirical distributions at successive post-change times illustrate the evolving separation from the null distribution, while individual and average sample paths show the temporal evolution and variability of this response. These experiments complement the results in Section \ref{adpthr} and provide further evidence of the statistic's performance across a broader range of underlying distributions. Figure~\ref{fig:boxplot-thresholds-510-540-difpre} and Figure~\ref{fig:average-individual-samplepath-diffpre} compare the empirical distributions of the monitoring statistic under several null distributions and the same alternative distribution. In contrast, Figure~\ref{fig:average-individual-samplepath-gaussian} illustrates sample paths of the monitoring statistic under the same Gaussian null distribution but different Gaussian alternatives.

\vspace{0.2cm}
\noindent\textbf{Uniform distribution.}\enspace
We consider a stream whose coordinates are independent and uniformly
distributed. The pre-change distribution $p$ is
\(\operatorname{Uniform}([-1,1]^{d}),
\)
$d=20$.
The six post-change distributions are displayed in the following order: $\operatorname{Uniform}([-0.25,0]^{d})$, $\operatorname{Uniform}([-0.25,0.25]^{d})$, $\operatorname{Uniform}([-0.5,0.5]^{d})$, $\operatorname{Uniform}([-2,2]^{d})$, $\operatorname{Uniform}([-3,3]^{d})$, and $\operatorname{Uniform}([5,12]^{d})$. Here, \(\operatorname{Uniform}([a,b]^{d})\) denotes the joint distribution of $d$ independent \(\operatorname{Uniform}([a,b])\) variables.
The alternatives above include changes in
dispersion alone as well as simultaneous changes in location and dispersion. The results are in Figure~\ref{fig:uniform-sample-paths} and Figure~\ref{fig:uniform-boxplots-510-540}.

\begin{figure}[!htbp]
\centering

\includegraphics[
  width=0.49\linewidth
]{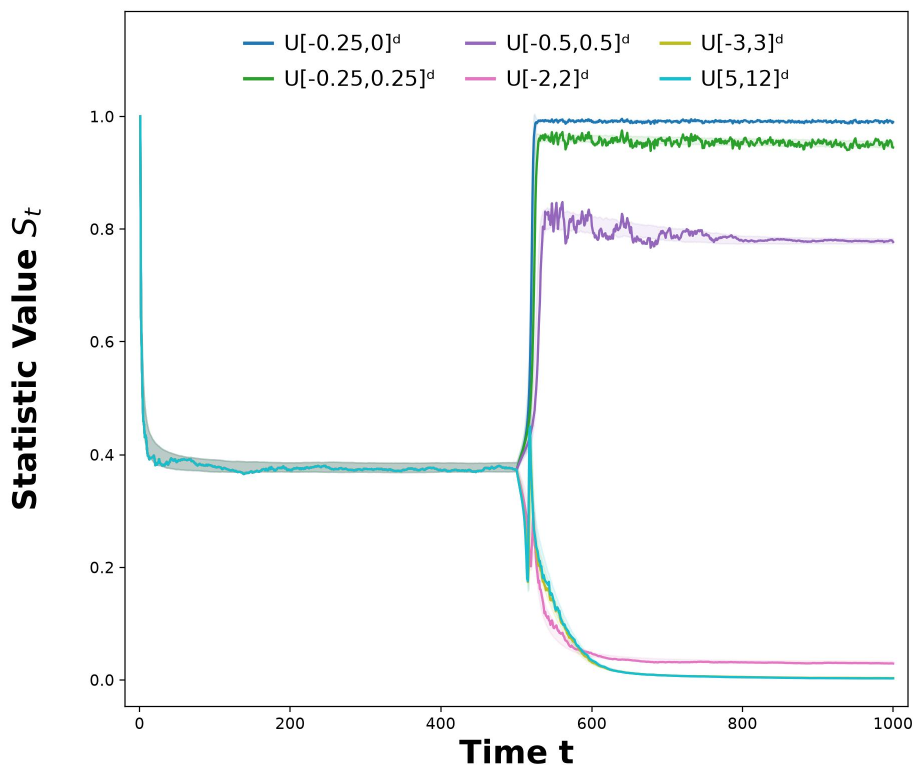}
\hfill
\includegraphics[
  width=0.49\linewidth
]{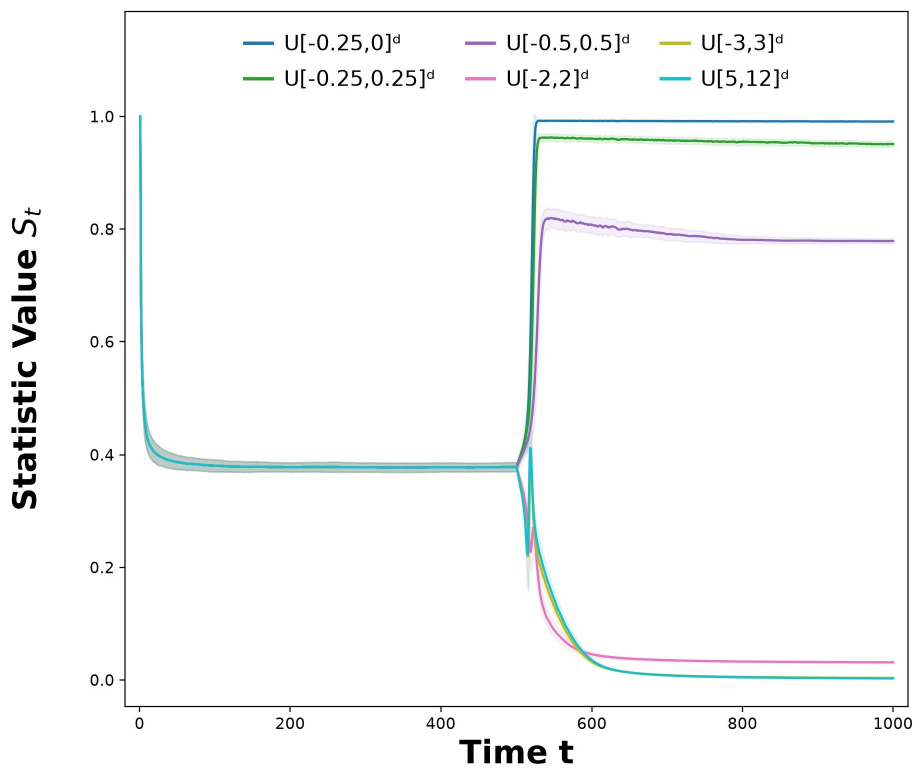}

\caption{Sample paths of the monitoring statistic under a common pre-change uniform distribution and  different
post-change uniform distributions, with the changepoint occurring at \(t=501\). The left panel
shows an individual sample path, with shaded regions representing the
estimated standard-deviation bands. The right panel shows the average
sample paths and their corresponding estimated standard-deviation bands.}
\label{fig:uniform-sample-paths}

\end{figure}

\noindent\textbf{Exponential distribution.}\enspace
We consider a stream whose coordinates are independent and exponentially
distributed. Before the change, all 20 coordinates follow independent
exponential distributions with scale parameter 1. After the changepoint occurs, the
coordinates remain independent and the six post-change alternatives are displayed in the following order:
from $\operatorname{exp}(0.1)^{d}$, $\operatorname{exp}(0.25)^{d}$, $\operatorname{exp}(0.5)^{d}$, $(\operatorname{exp}(0.25)-0.25)^{d}$, $\operatorname{exp}(2)^{d}$, and $(\operatorname{exp}(2)-2)^{d}$.
The results are reported in Figure~\ref{fig:exponential-sample-paths} and Figure \ref{fig:exponential-boxplots-510-540}.
\vspace{0.2cm}

\begin{figure*}[!t]
\centering
\includegraphics[width=0.49\textwidth]{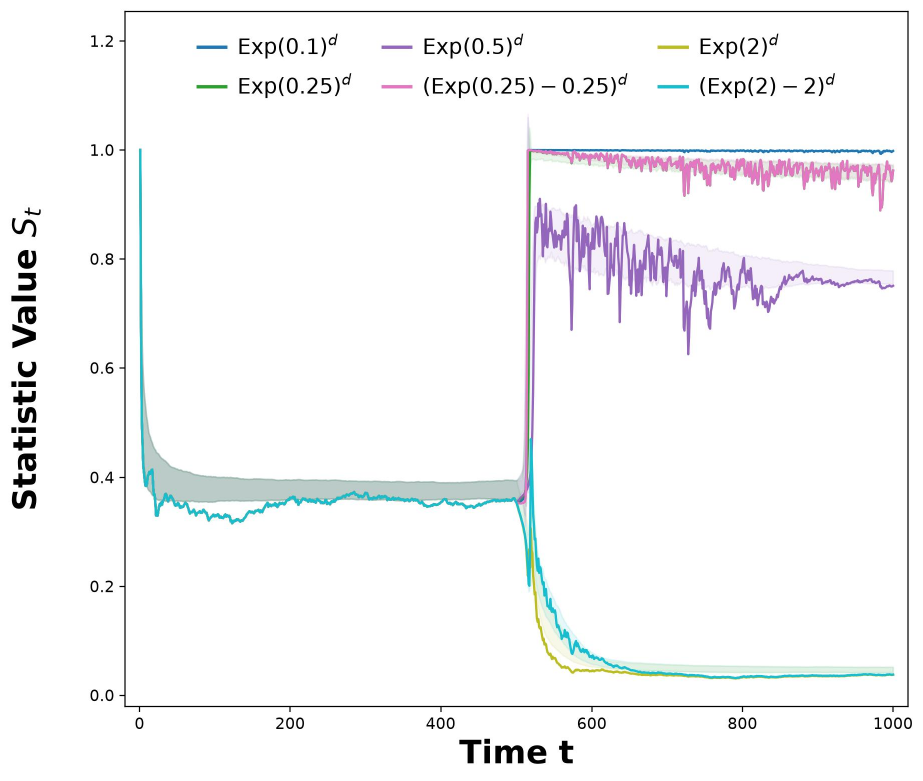}\hfill
\includegraphics[width=0.49\textwidth]{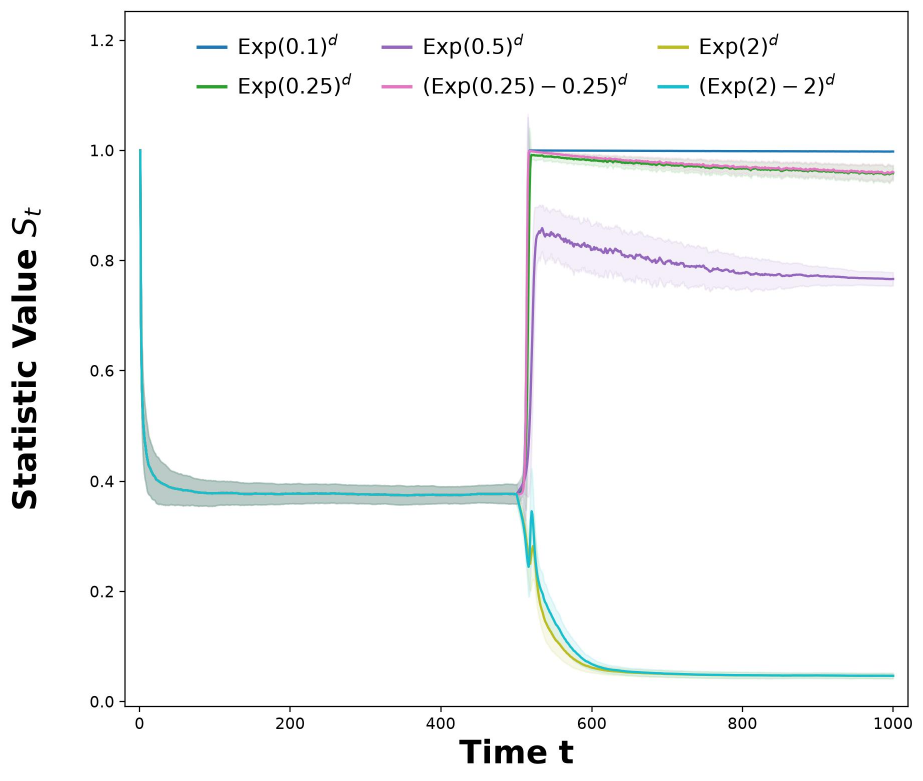}
\caption{Sample paths of the monitoring statistic under a common pre-change exponential distribution and different
         post-change exponential distributions, with the changepoint occurring at
        \(t=501\). The first panel shows an individual sample path, with
        shaded regions representing the estimated standard-deviation
        bands. The second panel shows the average sample paths and their
        corresponding estimated standard-deviation bands.}
 \label{fig:exponential-sample-paths}
\end{figure*}

\noindent\textbf{Laplace distribution.}\enspace
We consider a stream whose coordinates are independent and follow Laplace distributions.
The pre-change distribution is
\(
p=\operatorname{Laplace}(0,1)^{d},
\)
with $d=20$,
and the six post-change alternatives are displayed in the following order:
$\operatorname{Laplace}(0,0.25)^{d}$,
$\operatorname{Laplace}(0,0.5)^{d}$,
$\operatorname{Laplace}(0,2)^{d}$,
$\operatorname{Laplace}(-1,0.25)^{d}$, $\operatorname{Laplace}(-2,0.5)^{d}$, and $\operatorname{Laplace}(3,2)^{d}$.
Here, \(\operatorname{Laplace}(a,b)^{d}\) denotes the joint distribution of $d$ independent Laplace variables with location \(a\) and scale \(b\). The alternatives include changes in scale as well as simultaneous changes in location and scale (see, Figure~\ref{fig:laplace-sample-paths} and Figures~\ref{fig:laplace-boxplots-510-540}).

\begin{figure*}[!t]
\centering
\includegraphics[width=0.49\textwidth]{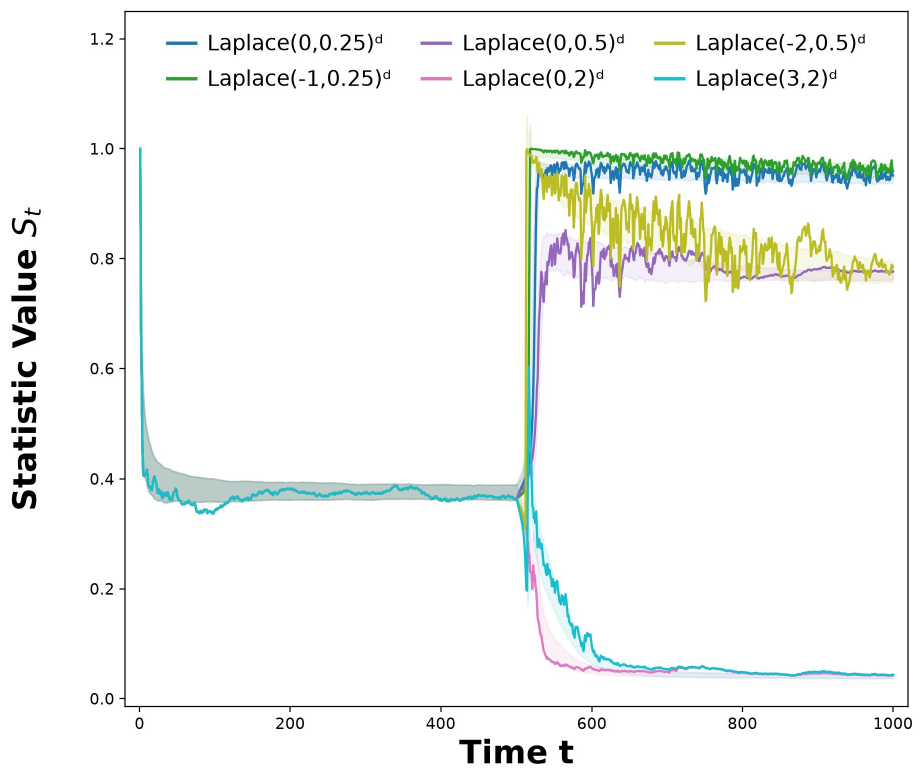}\hfill
\includegraphics[width=0.49\textwidth]{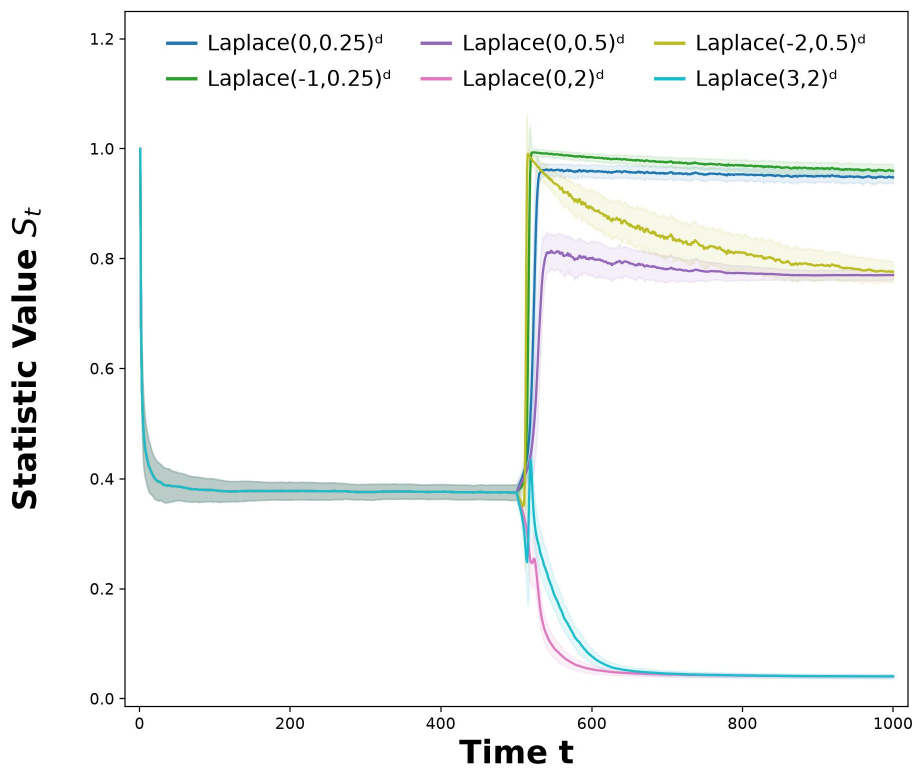}
\caption{Sample paths of the monitoring statistic under a common pre-change Laplace distribution and different
         post-change Laplace distributions, with the changepoint occurring at
        \(t=501\). The first panel shows an individual sample path, with
        shaded regions representing the estimated standard-deviation
        bands. The second panel shows the average sample paths and their
        corresponding estimated standard-deviation bands.}
 \label{fig:laplace-sample-paths}
\end{figure*}

\begin{figure}[!htbp]
\centering

\includegraphics[
  width=0.32\linewidth,
  height=0.155\textheight,
  keepaspectratio
]{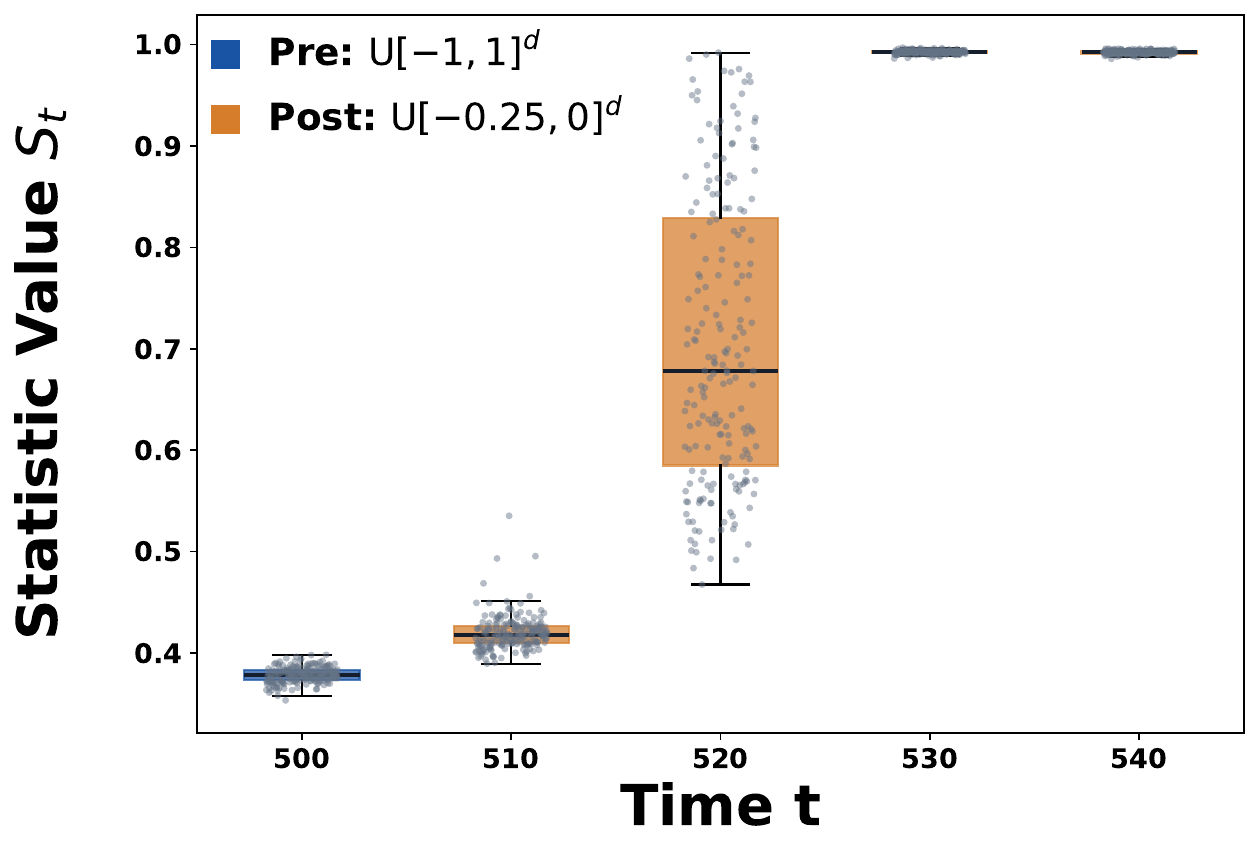}
\hfill
\includegraphics[
  width=0.32\linewidth,
  height=0.155\textheight,
  keepaspectratio
]{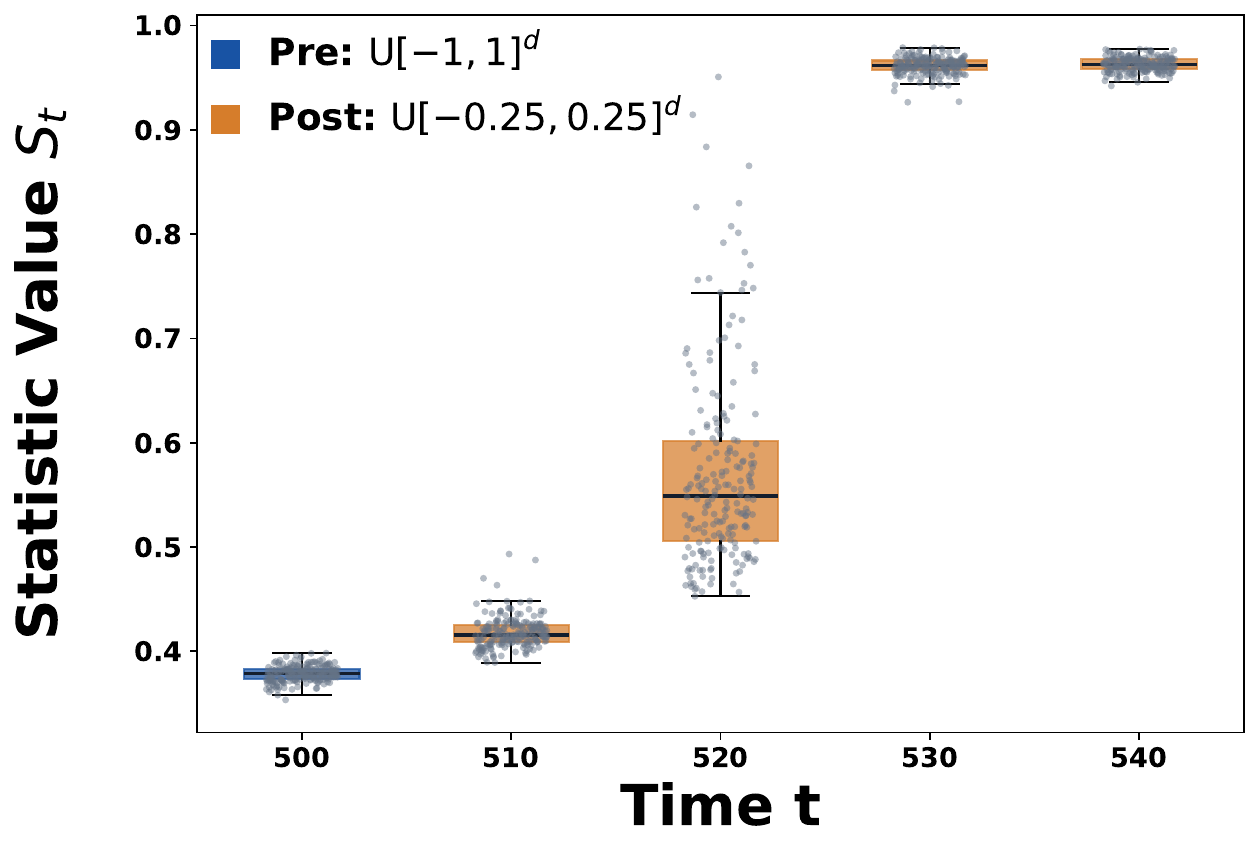}
\hfill
\includegraphics[
  width=0.32\linewidth,
  height=0.155\textheight,
  keepaspectratio
]{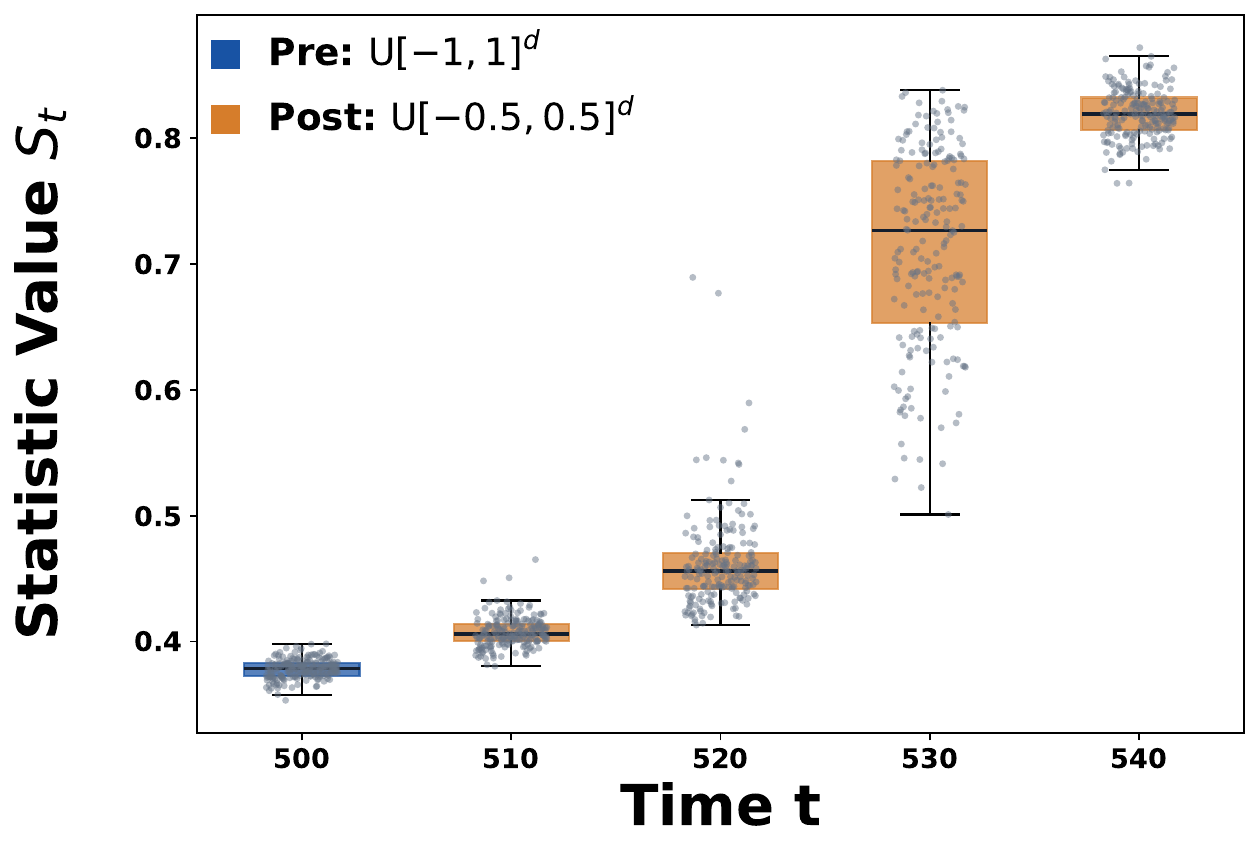}

\smallskip

\includegraphics[
  width=0.32\linewidth,
  height=0.155\textheight,
  keepaspectratio
]{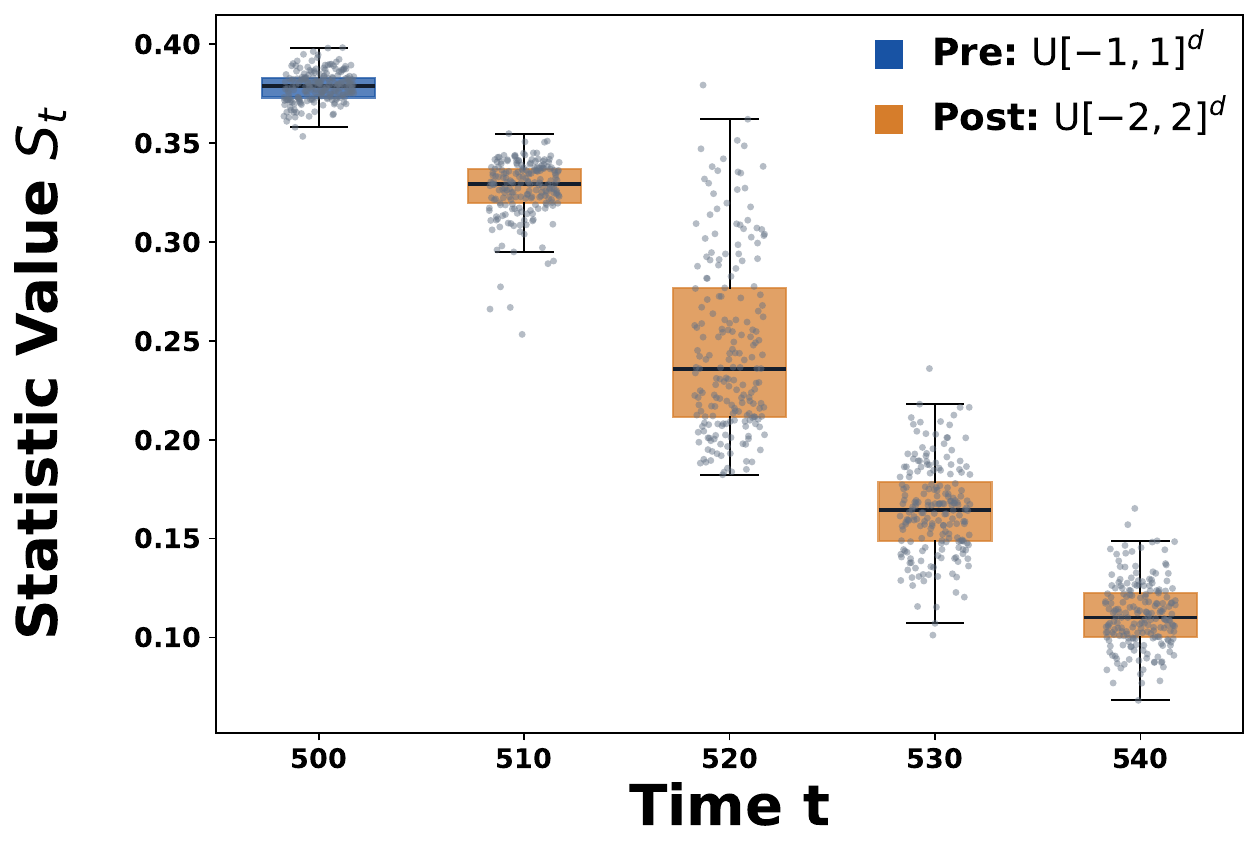}
\hfill
\includegraphics[
  width=0.32\linewidth,
  height=0.155\textheight,
  keepaspectratio
]{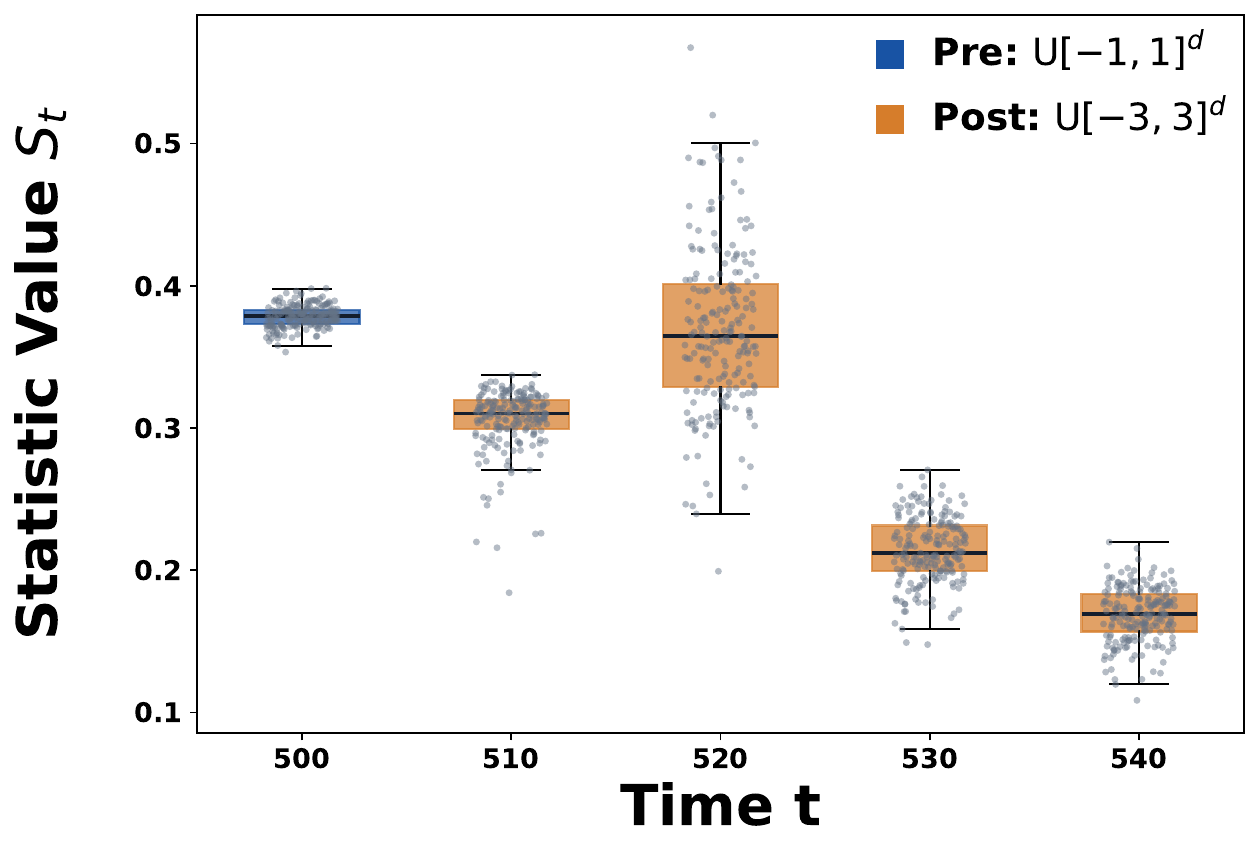}
\hfill
\includegraphics[
  width=0.32\linewidth,
  height=0.155\textheight,
  keepaspectratio
]{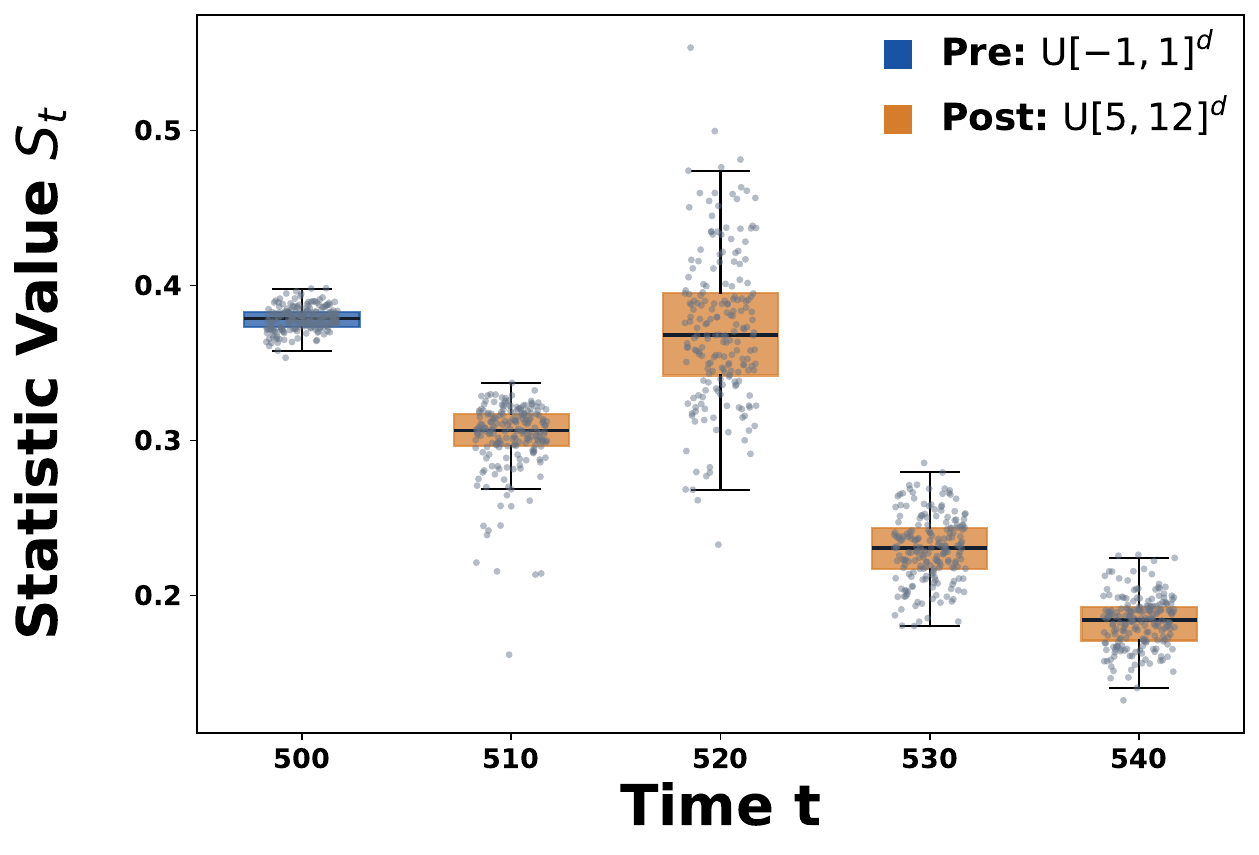}

\caption{Empirical distributions of the monitoring statistic under a common
pre-change uniform distribution and different post-change uniform distributions in dimension \(d=20\), with the
changepoint occurring at \(t=501\). The statistic under the pre-change
distribution is evaluated at \(t=500\), while the statistics under the
post-change distributions are evaluated at \(t=510,520,530,540\).
The top row presents, from left to right, the post-change distributions 
\(\operatorname{Uniform}([-0.25,0]^d)\),
\(\operatorname{Uniform}([-0.25,0.25]^d)\), and
\(\operatorname{Uniform}([-0.5,0.5]^d)\). The bottom row presents, from
left to right, the post-change distributions
\(\operatorname{Uniform}([-2,2]^d)\),
\(\operatorname{Uniform}([-3,3]^d)\), and
\(\operatorname{Uniform}([5,12]^d)\).}
\label{fig:uniform-boxplots-510-540}

\end{figure}

\begin{figure}[!htbp]
\centering

\includegraphics[
  width=0.32\linewidth,
  height=0.155\textheight,
  keepaspectratio
]{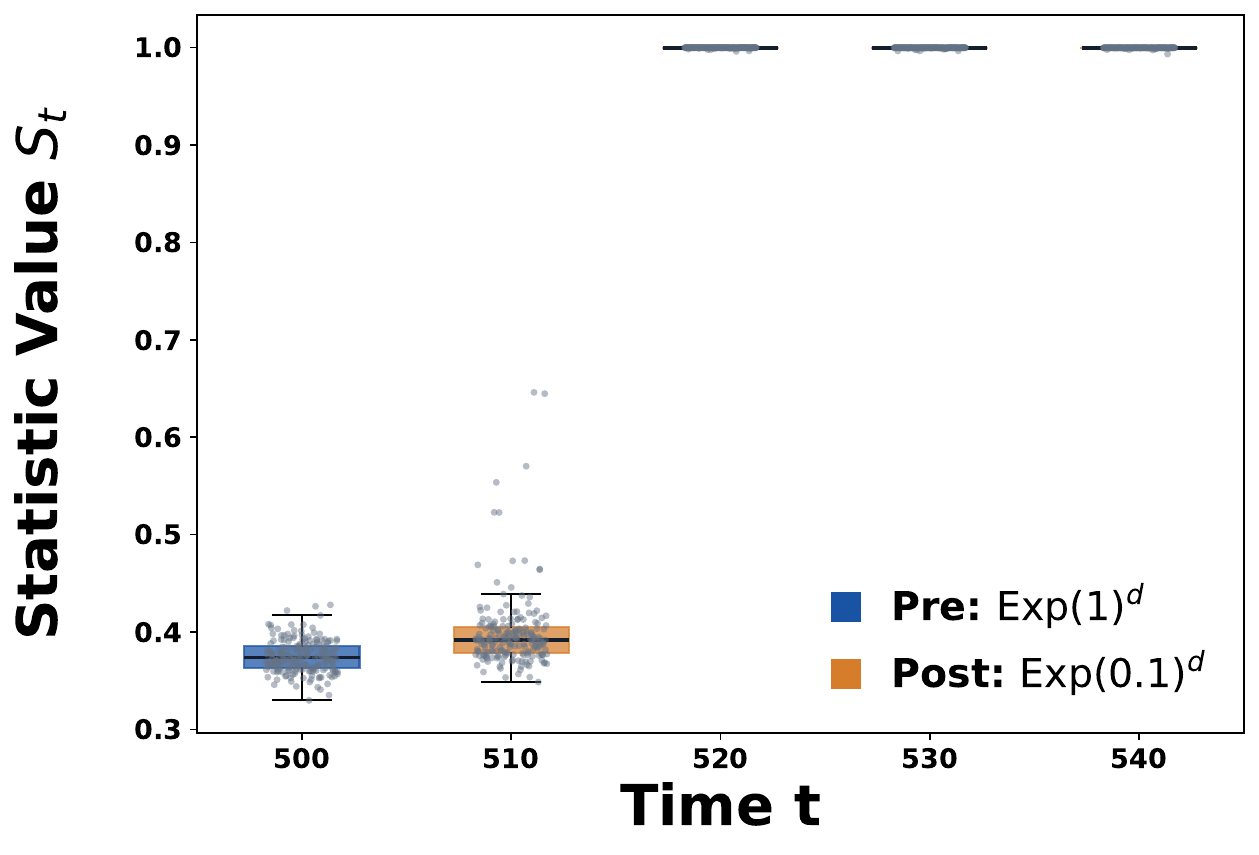}
\hfill
\includegraphics[
  width=0.32\linewidth,
  height=0.155\textheight,
  keepaspectratio
]{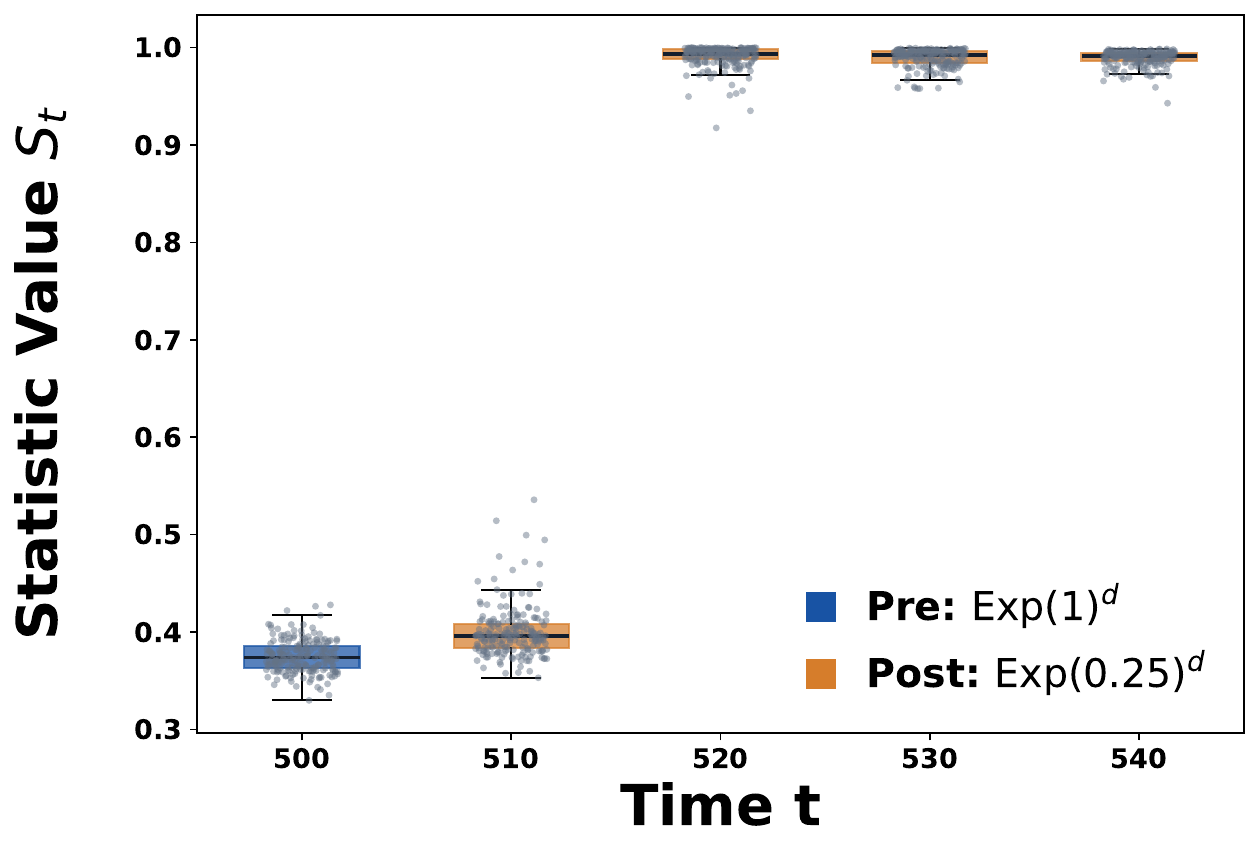}
\hfill
\includegraphics[
  width=0.32\linewidth,
  height=0.155\textheight,
  keepaspectratio
]{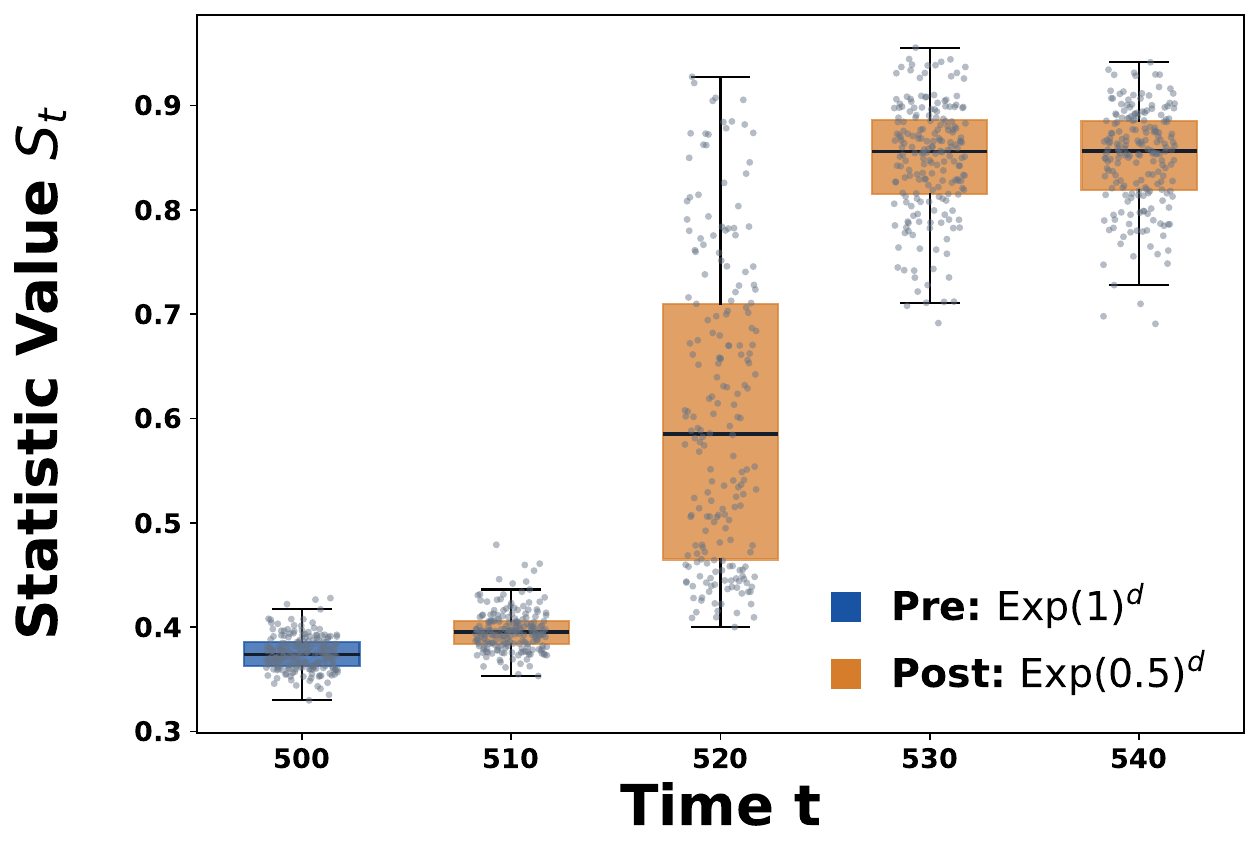}

\smallskip

\includegraphics[
  width=0.32\linewidth,
  height=0.155\textheight,
  keepaspectratio
]{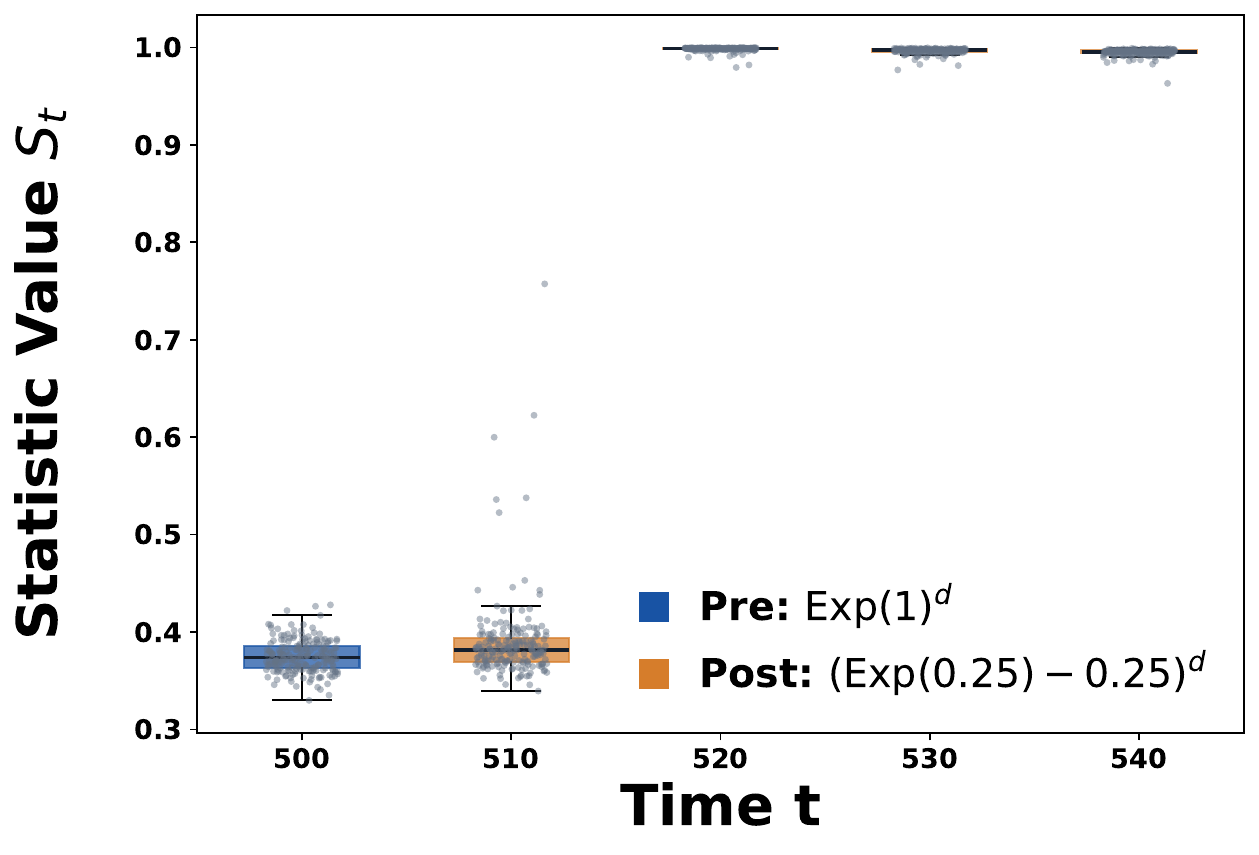}
\hfill
\includegraphics[
  width=0.32\linewidth,
  height=0.155\textheight,
  keepaspectratio
]{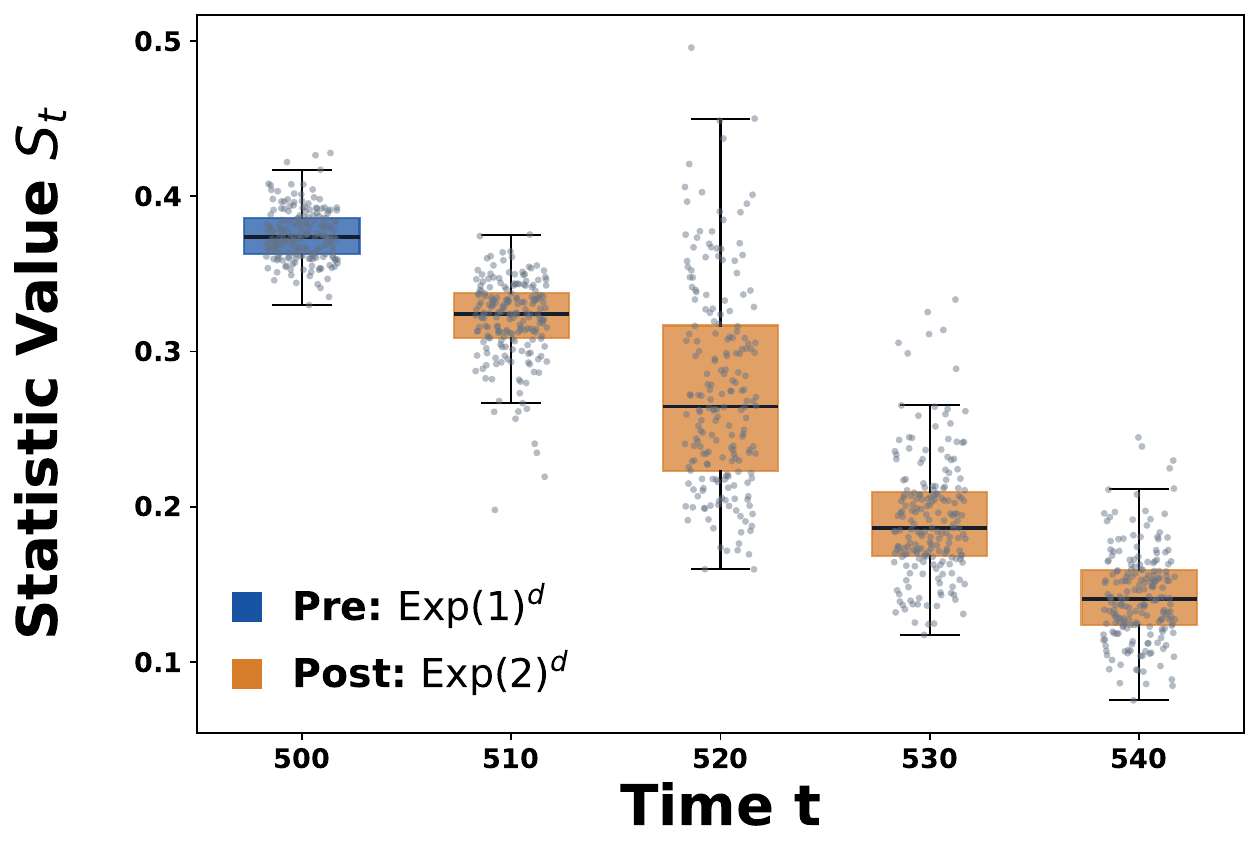}
\hfill
\includegraphics[
  width=0.32\linewidth,
  height=0.155\textheight,
  keepaspectratio
]{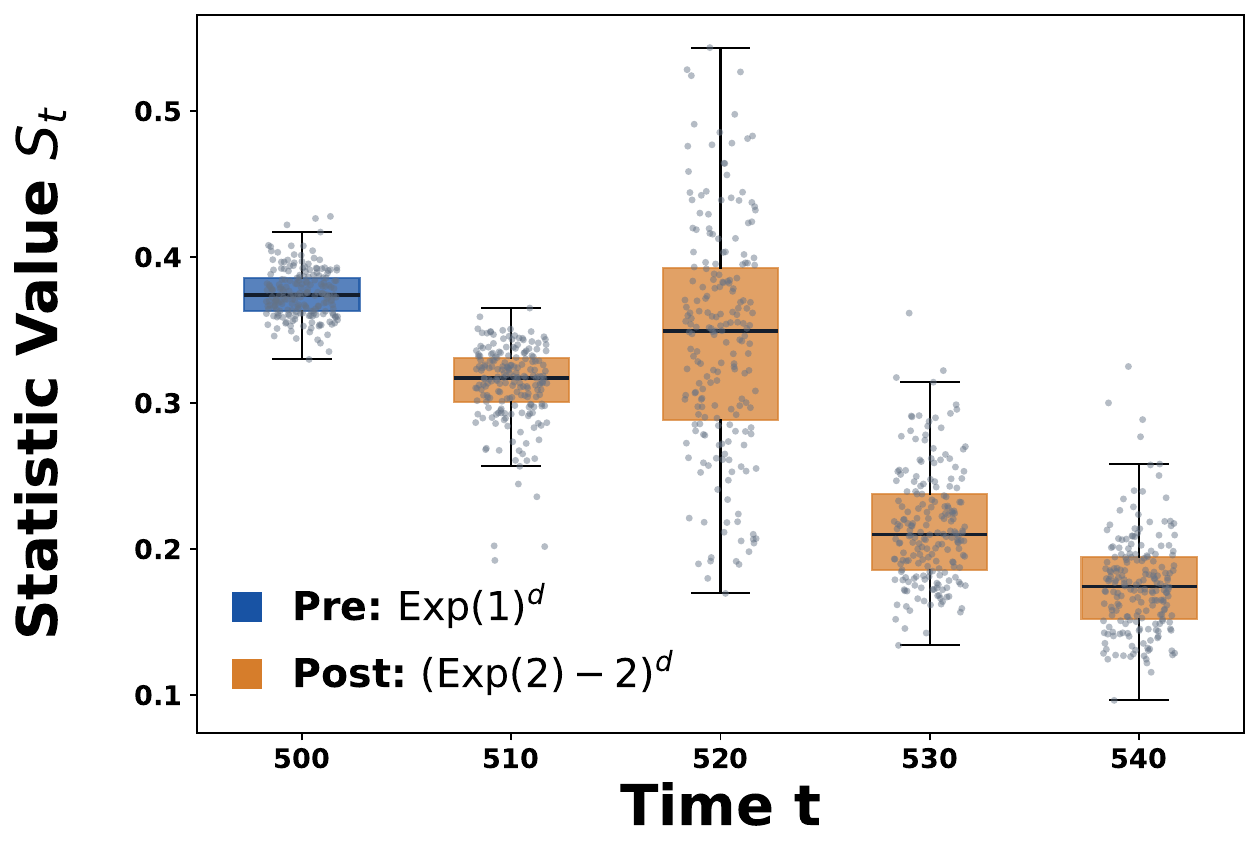}

\caption{Empirical distributions of the monitoring statistic under a common
pre-change exponential distribution and different post-change exponential distributions in dimension \(d=20\), with
the changepoint occurring at \(t=501\). The statistic under the pre-change
distribution is evaluated at \(t=500\), while the statistics under the
post-change distributions are evaluated at \(t=510,520,530,540\).
The top row presents, from left to right, the post-change distributions 
\(\operatorname{Exp}(0.1)^d\), \(\operatorname{Exp}(0.25)^d\), and
\(\operatorname{Exp}(0.5)^d\). The bottom row presents, from left to
right, the post-change distributions 
\((\operatorname{Exp}(0.25)-0.25)^d\),
\(\operatorname{Exp}(2)^d\), and
\((\operatorname{Exp}(2)-2)^d\).}
\label{fig:exponential-boxplots-510-540}

\end{figure}

\begin{figure}[!htbp]
\centering

\includegraphics[
  width=0.32\linewidth,
  height=0.155\textheight,
  keepaspectratio
]{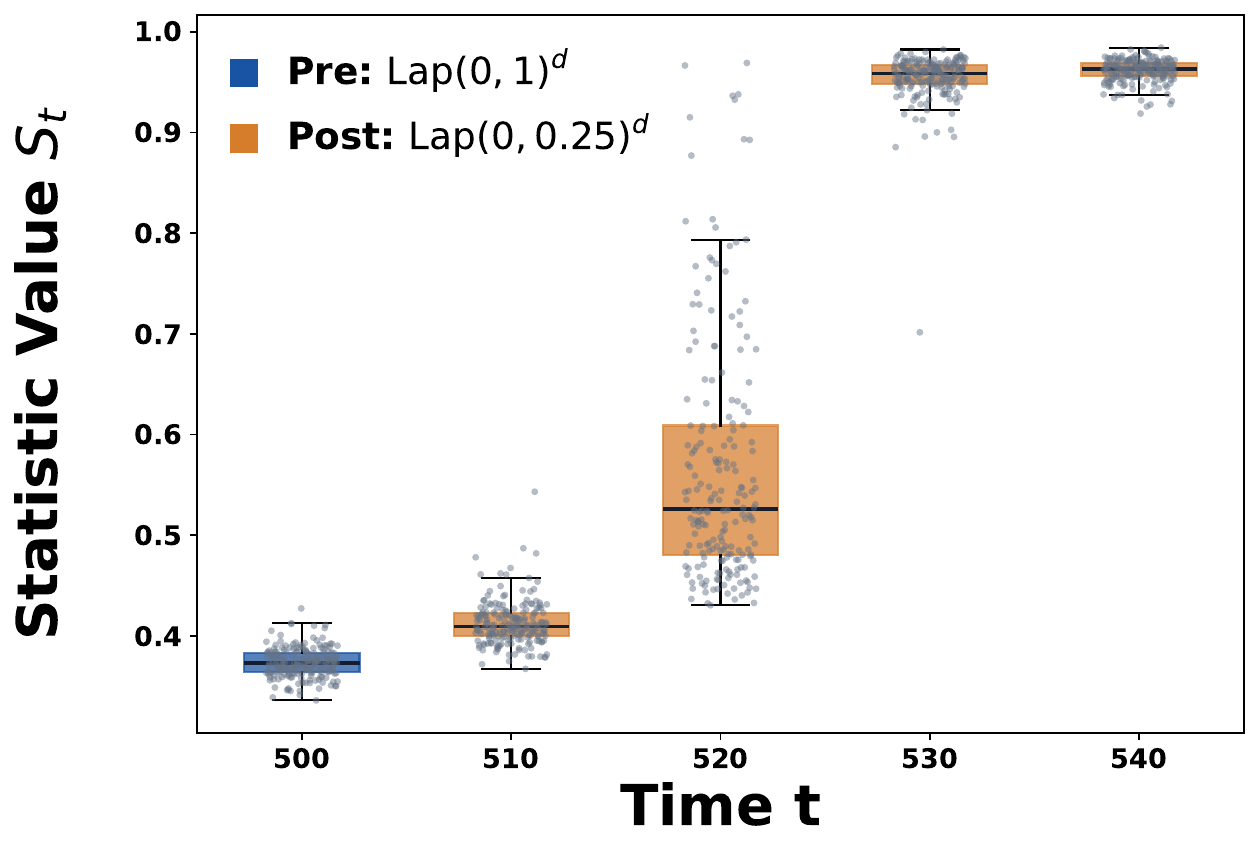}
\hfill
\includegraphics[
  width=0.32\linewidth,
  height=0.155\textheight,
  keepaspectratio
]{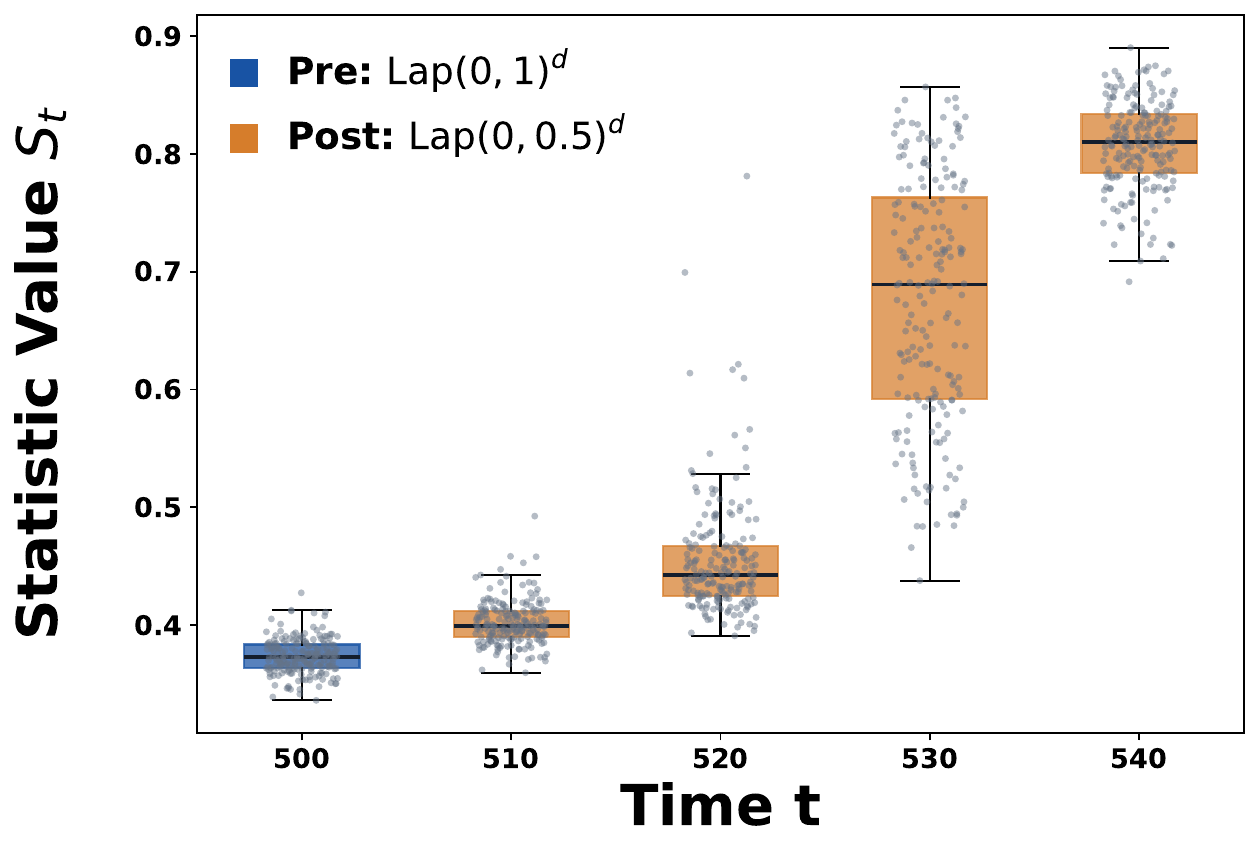}
\hfill
\includegraphics[
  width=0.32\linewidth,
  height=0.155\textheight,
  keepaspectratio
]{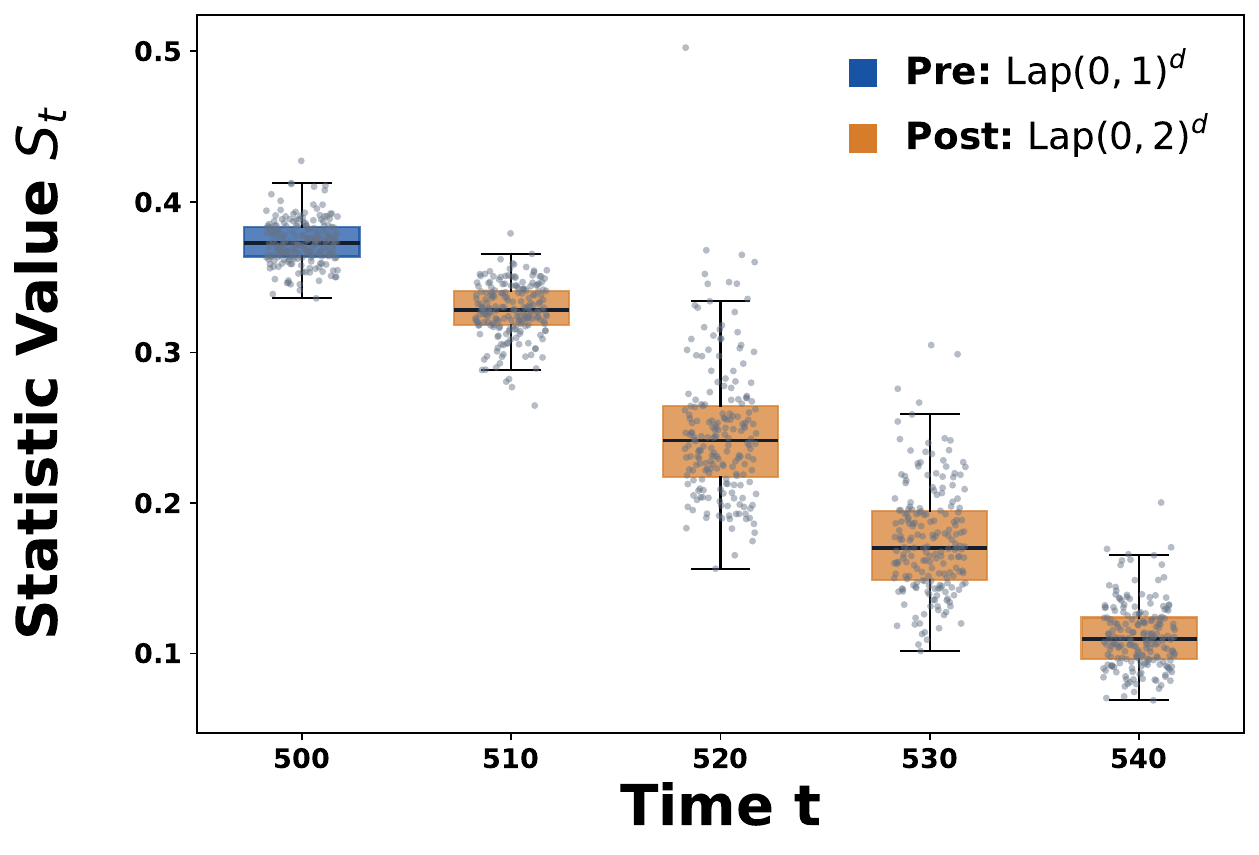}

\smallskip

\includegraphics[
  width=0.32\linewidth,
  height=0.155\textheight,
  keepaspectratio
]{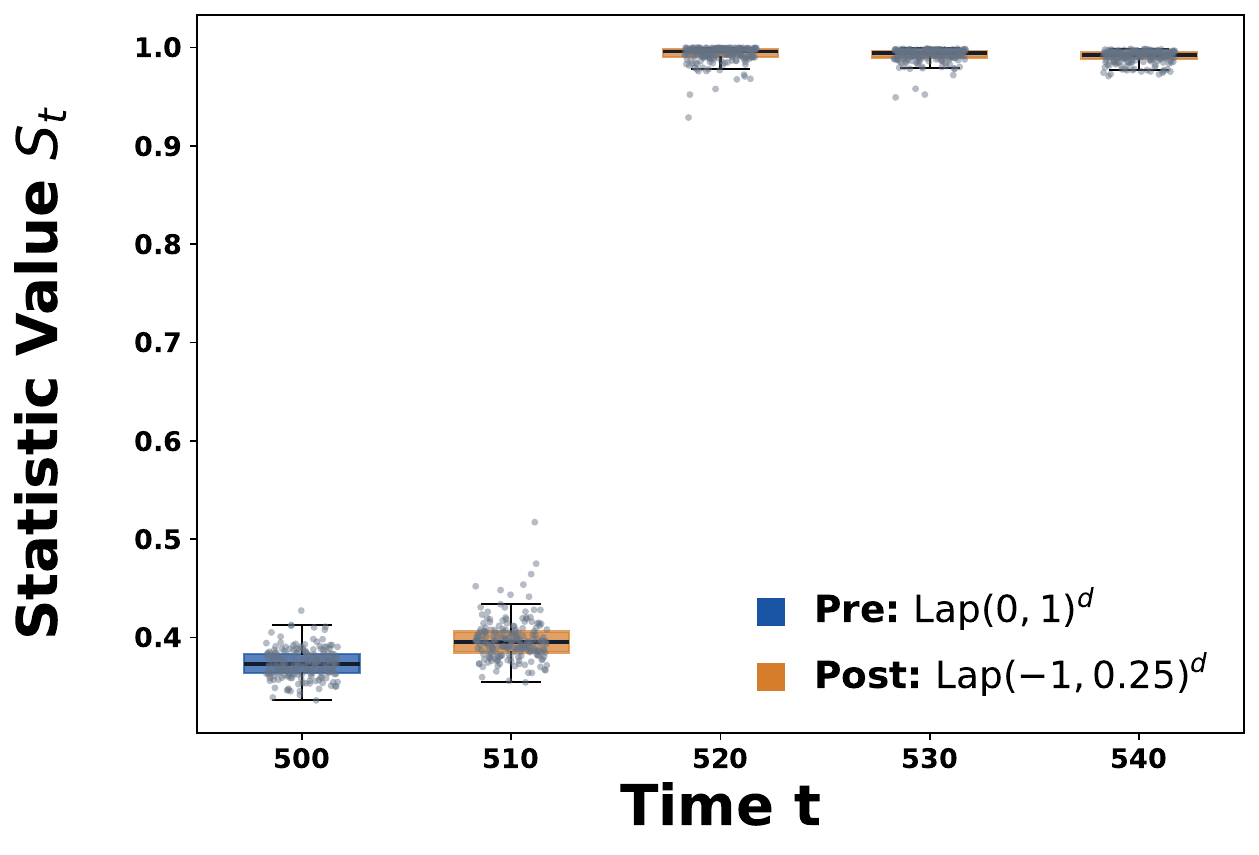}
\hfill
\includegraphics[
  width=0.32\linewidth,
  height=0.155\textheight,
  keepaspectratio
]{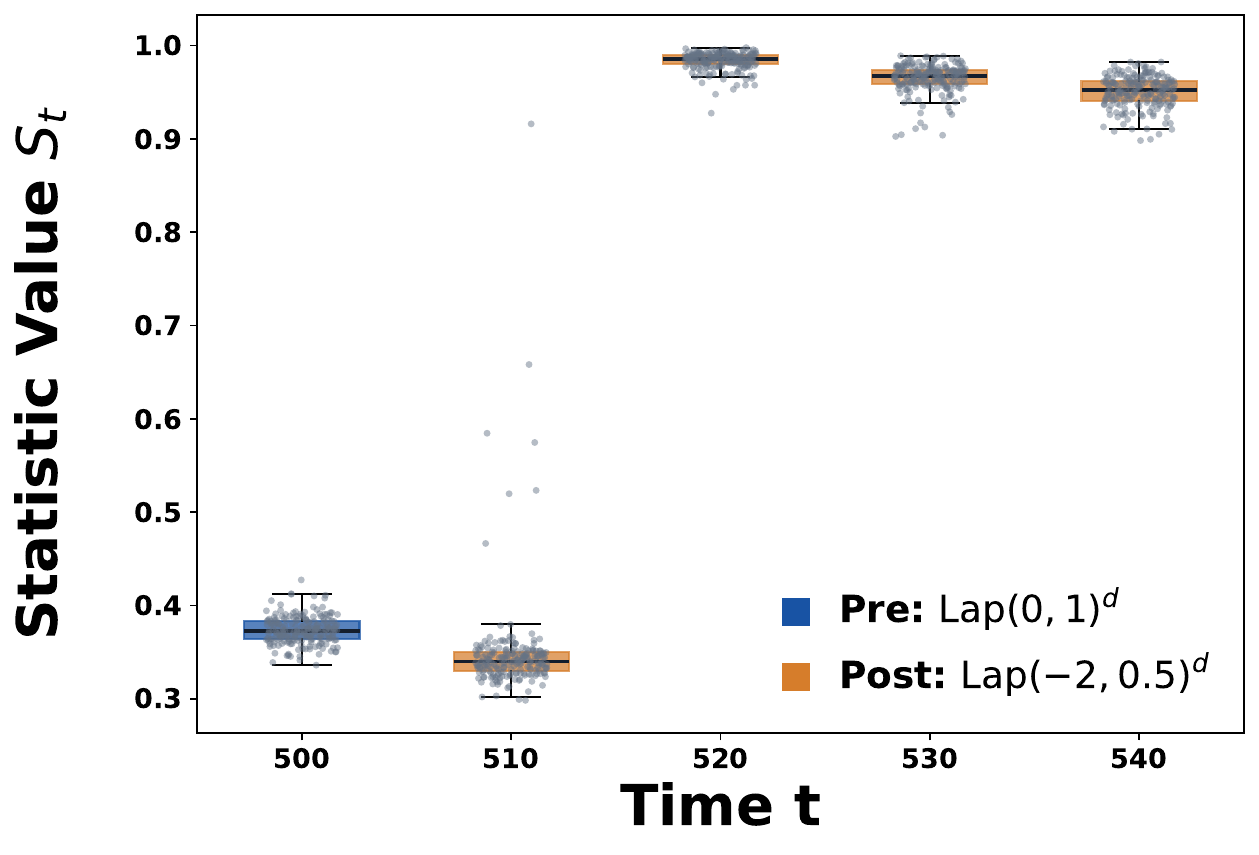}
\hfill
\includegraphics[
  width=0.32\linewidth,
  height=0.155\textheight,
  keepaspectratio
]{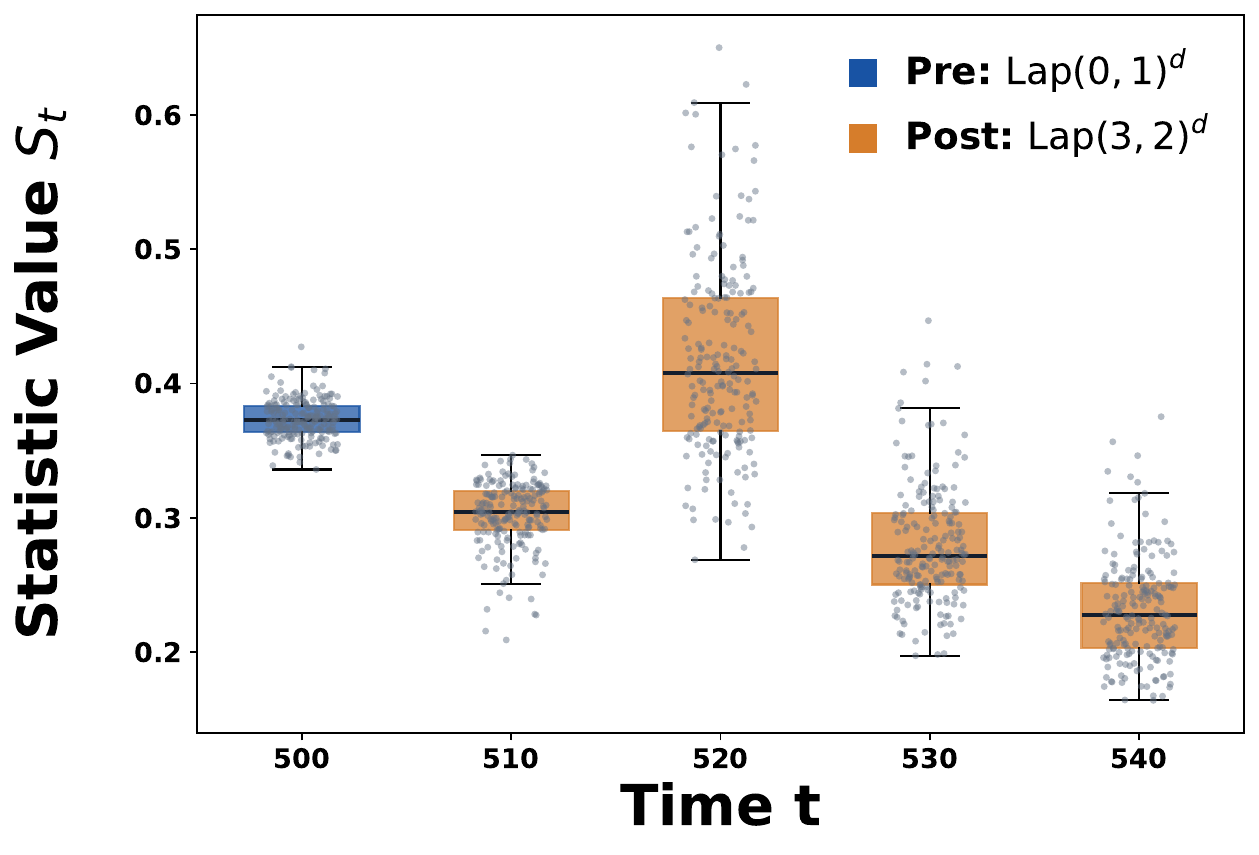}

\caption{Empirical distributions of the monitoring statistic under a common
pre-change Laplace distribution and different post-change Laplace distributions in dimension \(d=20\), with the
changepoint occurring at \(t=501\). The statistic under the pre-change
distribution is evaluated at \(t=500\), while the statistics under the
post-change distributions are evaluated at \(t=510,520,530,540\).
The top row presents, from left to right, the post-change distributions
\(\operatorname{Laplace}(0,0.25)^d\),
\(\operatorname{Laplace}(0,0.5)^d\), and
\(\operatorname{Laplace}(0,2)^d\). The bottom row presents, from left to
right, the post-change distributions
\(\operatorname{Laplace}(-1,0.25)^d\),
\(\operatorname{Laplace}(-2,0.5)^d\), and
\(\operatorname{Laplace}(3,2)^d\).}
\label{fig:laplace-boxplots-510-540}

\end{figure}

\end{document}